\documentclass[11pt]{article}
\usepackage{graphicx}
\usepackage{amsmath,amsbsy,bm}
\usepackage{amssymb}
\usepackage{float}
\usepackage{amsthm}
\usepackage{enumerate}
\usepackage[normalem]{ulem}
\usepackage{tikz, pgfplots}
\usetikzlibrary{arrows}
\usepackage[all]{nowidow}
\usepackage[left=2cm, right=2cm, top=2.6cm,bottom=2.6cm]{geometry}
\usepackage{stmaryrd}
\usepackage[export]{adjustbox}
\usepackage[colorlinks=true,urlcolor=blue,citecolor=blue,linkcolor=blue,backref=page,breaklinks]{hyperref}
\usepackage{url}
\usepackage[capitalise]{cleveref}
\usepackage[color=red]{todonotes}
\usepackage{caption}
\usepackage{subcaption}
\usepackage{physics}
\usepackage{stmaryrd}
\usepackage{wasysym}
\usepackage{comment}
\usepackage[T1]{fontenc}
\usepackage{blindtext}
\usepackage{titlesec}
\usepackage[affil-it]{authblk}
\usepackage{multicol}
\usepackage{tocloft}

\hypersetup{final}
\usepackage{bbm}

\newcommand{\Hbf}{{\boldsymbol{H}}}
\newcommand{\Obf}{{\boldsymbol{O}}}
\newcommand{\Kbf}{{\boldsymbol{K}}}
\newcommand{\Abf}{{\boldsymbol{A}}}

\newcommand{\Ibf}{{\boldsymbol{I}}}

\newcommand{\sigmabf}{{\boldsymbol{\sigma}}}
\newcommand{\mubf}{{\boldsymbol{\mu}}}

\newcommand{\lind}{{\mathcal{L}}}
\newcommand{\EV}{{\mathbb{E}}}
\newcommand{\bTr}{{\overline{\Tr}}}

\newcommand{\aloc}{{a_{\mathrm{loc}}}}
\newcommand{\alocPower}[1]{{a_{\mathrm{loc}}^{#1}}}
\newcommand{\aac}{{a_{\mathrm{ac}}}}

\newcommand{\Hloc}{{\norm{\Hbf}_{\mathrm{loc}}}}
\newcommand{\Hglo}{{\norm{\Hbf}_{\mathrm{glo}}}}
\newcommand{\HgloPower}[1]{{\norm{\Hbf}_{\mathrm{glo}}^{#1}}}
\newcommand{\HlocPower}[1]{{\norm{\Hbf}_{\mathrm{loc}}^{#1}}}

\newcommand{\id}{{\boldsymbol{I}}}

\newcommand{\labs}[1]{\left\vert {#1} \right\vert}

\newcommand{\sss}{\scriptscriptstyle}

\newcommand{\adj}[3]{\raisebox{#1}{\scalebox{#2}{{#3}}}}

\renewcommand{\vec}{\bm}
\newcommand{\CA}{\mathcal{A}}
\newcommand{\CB}{\mathcal{B}}
\newcommand{\CC}{\mathcal{C}}

\newcommand{\CD}{\mathcal{D}}
\newcommand{\CE}{\mathcal{E}}
\newcommand{\BE}{\mathbb{E}}

\newcommand{\CL}{\mathcal{L}}
\newcommand{\CM}{\mathcal{M}}

\newcommand{\CO}{\mathcal{O}}

\newcommand{\CR}{\mathcal{R}}

\newcommand{\CT}{\mathcal{T}}

\newcommand{\CW}{\mathcal{W}}

\newcommand{\vA}{\bm{A}}
\newcommand{\vB}{\bm{B}}
\newcommand{\vC}{\bm{C}}

\newcommand{\vH}{\bm{H}}
\newcommand{\vI}{\bm{I}}

\newcommand{\vK}{\bm{K}}

\newcommand{\vO}{\bm{O}}

\newcommand{\vsigma}{\bm{ \sigma}}

\newcommand{\vrho}{\bm{ \rho}}
\renewcommand{\L}{\left}
\newcommand{\R}{\right}

\newcommand{\symg}{{\mathfrak{S}}}

\DeclareMathOperator{\Unif}{Unif}

\newcommand{\vertiii}[1]{{\left\vert\kern-0.25ex\left\vert\kern-0.25ex\left\vert #1 \right\vert\kern-0.25ex\right\vert\kern-0.25ex\right\vert}}

\newcommand{\e}{\mathrm{e}}

\newcommand{\rd}{\mathrm{d}}
\newcommand*{\poly}{\mathrm{Poly}}

\newcommand{\indicator}{\mathbbm{1}}
\newcommand{\Lword}[1]{\text{Lindbladian}}

\newcommand{\undersetbrace}[2]{ \underset{#1}{\underbrace{#2}}}

\newcounter{pathindex}
\newcounter{pathlength}
\def\xspacing{4em}
\def\yspacing{2em}
\def\radius{0.3em}
\newcommand{\drawDiagram}[1]{
\begin{tikzpicture}
\setcounter{pathindex}{1}
\setcounter{pathlength}{0}
\foreach  \p/\t in #1
{
    \setcounter{pathlength}{0}
    \foreach \j in \p
    {
        \draw[fill=white] (\value{pathindex}*\xspacing,\value{pathlength}*\yspacing) circle (\radius);
        \node at (\value{pathindex}*\xspacing-4*\radius,\value{pathlength}*\yspacing) {\j} ;
        \ifthenelse{\value{pathlength}>0}{
            \draw (\value{pathindex}*\xspacing,\value{pathlength}*\yspacing-\radius) -- (\value{pathindex}*\xspacing,\value{pathlength}*\yspacing-\yspacing+\radius);
        }{}
        \stepcounter{pathlength}
    }
    \ifthenelse{\t = 0}
    {
        \draw[->] (\value{pathindex}*\xspacing,\value{pathlength}*\yspacing-\yspacing+\radius) -- (\value{pathindex}*\xspacing,\value{pathlength}*\yspacing-\yspacing/2);
    }{}
    \ifthenelse{\t = \value{pathindex}}
    {
        \draw[fill] (\value{pathindex}*\xspacing,\value{pathlength}*\yspacing-\yspacing) circle (\radius);
    }
    {}
    \ifthenelse{\t > \value{pathindex}}
    {
        \draw[dashed] (\value{pathindex}*\xspacing+\radius,\value{pathlength}*\yspacing-\yspacing) -- (\value{pathindex}*\xspacing+\radius+\xspacing/4,\value{pathlength}*\yspacing-\yspacing) -- (\value{pathindex}*\xspacing+\radius+\xspacing/4,\value{pathindex}*\yspacing/4 - \t*\yspacing/4) -- (\value{pathindex}*\xspacing/2 + \t*\xspacing/2 + \xspacing/8, \value{pathindex}*\yspacing/2 - \t*\yspacing/2) -- (\t*\xspacing, 0);
    }
    {}
    \stepcounter{pathindex}
}
\end{tikzpicture}
}

\newcommand{\drawDiagramIndexed}[1]{
\begin{tikzpicture}
\setcounter{pathindex}{1}
\setcounter{pathlength}{0}
\foreach  \p/\t in #1
{
    \setcounter{pathlength}{0}
    \foreach \j/\l in \p
    {
        \draw[fill=white] (\value{pathindex}*\xspacing,\value{pathlength}*\yspacing) circle (\radius);
        \node at (\value{pathindex}*\xspacing-4*\radius,\value{pathlength}*\yspacing) {\j} ;
        \ifthenelse{\value{pathlength}>0}{
            \draw (\value{pathindex}*\xspacing,\value{pathlength}*\yspacing-\radius) -- (\value{pathindex}*\xspacing,\value{pathlength}*\yspacing-\yspacing+\radius);
            \node at (\value{pathindex}*\xspacing+2*\radius,\value{pathlength}*\yspacing-\yspacing/2) {\l};
        }{}
        \stepcounter{pathlength}
    }
    \ifthenelse{\t = 0}
    {
        \draw[->] (\value{pathindex}*\xspacing,\value{pathlength}*\yspacing-\yspacing+\radius) -- (\value{pathindex}*\xspacing,\value{pathlength}*\yspacing-\yspacing/2);
    }{}
    \ifthenelse{\t = \value{pathindex}}
    {
        \draw[fill] (\value{pathindex}*\xspacing,\value{pathlength}*\yspacing-\yspacing) circle (\radius);
    }
    {}
    \ifthenelse{\t > \value{pathindex}}
    {
        \draw[dashed] (\value{pathindex}*\xspacing+\radius,\value{pathlength}*\yspacing-\yspacing) -- (\value{pathindex}*\xspacing+\radius+\xspacing/4,\value{pathlength}*\yspacing-\yspacing) -- (\value{pathindex}*\xspacing+\radius+\xspacing/4,\value{pathindex}*\yspacing/4 - \t*\yspacing/4) -- (\value{pathindex}*\xspacing/2 + \t*\xspacing/2 + \xspacing/8, \value{pathindex}*\yspacing/2 - \t*\yspacing/2) -- (\t*\xspacing, 0);
    }
    {}
    \stepcounter{pathindex}
}
\end{tikzpicture}
}

\usepackage{mathtools}

\DeclarePairedDelimiter{\bb}{\llbracket}{\rrbracket}

\makeatletter
\newcommand{\linethrough}{\mathpalette\@thickbar}
\newcommand{\@thickbar}[2]{{#1\mkern0mu\vbox{
    \sbox\z@{$#1#2\mkern-1.5mu$}%
    \dimen@=\dimexpr\ht\tw@-\ht\z@+2\p@\relax 
    \hrule\@height0.5\p@ 
    \vskip\dimen@
    \box\z@}}
}
\makeatother

\newtheorem{theorem}{Theorem}

\newtheorem{lemma}[theorem]{Lemma}
\newtheorem{definition}[theorem]{Definition}

\newtheorem{condition}[theorem]{Condition}
\newtheorem{proposition}[theorem]{Proposition}

\newtheorem{innercustomthm}{}
\newenvironment{customthm}[1]
{\renewcommand\theinnercustomthm{#1}\innercustomthm}
  {\endinnercustomthm}

\numberwithin{equation}{section}
\numberwithin{theorem}{section}

\title{Optimizing random local Hamiltonians by dissipation}

\author[1,2]{Joao Basso\thanks{joao.basso@berkeley.edu}}
\author[3,4]{Chi-Fang Chen\thanks{achifchen@gmail.com}}
\author[1]{Alexander M. Dalzell\thanks{dalzel@amazon.com}}

\affil[1]{\footnotesize AWS Center for Quantum Computing, Pasadena, CA, USA}

\affil[2]{\footnotesize Department of Mathematics, University of California, Berkeley, CA, USA}

\affil[3]{\footnotesize University of California, Berkeley, CA, USA}
\affil[4]{\footnotesize Massachusetts Institute of Technology,
Cambridge, USA}
\date{\today}

\begin{document}

\maketitle

\begin{abstract}

A central challenge in quantum simulation is to prepare low-energy states of strongly interacting many-body systems. In this work, we study the problem of preparing a quantum state that optimizes a random all-to-all, sparse or dense, spin or fermionic $k$-local Hamiltonian. We prove that a simplified quantum Gibbs sampling algorithm achieves a $\Omega(\frac{1}{k})$-fraction approximation of the optimum, giving an exponential improvement on the $k$-dependence over the prior best (both classical and quantum) algorithmic guarantees. Combined with the circuit lower bound for such states, our results suggest that finding low-energy states for sparsified (quasi)local spin and fermionic models is quantumly easy but classically nontrivial. This further indicates tha 
quantum Gibbs sampling may be a suitable metaheuristic for optimization problems. 
\end{abstract}

\vspace{0in}
\setlength{\columnsep}{25pt}
\begin{multicols}{2}
{\small \tableofcontents}
\end{multicols}
\thispagestyle{empty}

\clearpage

\setcounter{page}{1}
\section{Introduction}
The low-energy properties of many-body quantum systems are central targets of quantum and classical simulation algorithms. While classical simulation algorithms have found extensive practical usage, there has not been a go-to quantum algorithmic principle for preparing ground states or thermal states at large scales. Due to the limited size of experimental or numerical data, it is difficult to obtain an empirical evaluation of the performance of any new quantum algorithmic proposal. Ideally, then, a new proposal should be accompanied by a rigorous mathematical proof of success. However, the QMA-hardness~\cite{kitaev2002classical} of worst-case instances precludes a general theoretical guarantee; to make progress, one must first restrict to a certain class of instances.

Low-energy state preparation can be precisely stated as a quantum optimization problem: given the description of a many-body Hamiltonian, efficiently prepare a quantum state such that the energy is as \textit{high} as possible.\footnote{Note the sign flip; maximizing the energy is equivalent to minimizing the energy of the negated Hamiltonian.} For simplicity, we will focus on models with as little spatial structure as necessary: random, few-body Hamiltonian problems with all-to-all connectivity, where the ensemble is fully specified by the system size $n$, the locality parameter $k$ (i.e., the number of sites on which each Hamiltonian term acts nontrivially), the sparsity $m$ (i.e., the total number of terms), and whether the degree of freedom at each site is spin or fermion (\cref{sec:models}). These models are natural quantum analogues of classical spin glass models~\cite{SK_model_75,derrida1980randomEnergyModel}, the crucial difference being that the random terms in the quantum models do not all commute. While these models lack structure---e.g., spatial locality and symmetry---seen in real-world chemistry or condensed matter systems, we hope that the average-case nature of this problem still covers some generic aspects of strongly interacting quantum systems and, at the very least, serves as a concrete theoretical testbed of new algorithmic ideas that may extend beyond quantum simulation.

So far, very little is known about the low-energy or low-temperature states of these random Hamiltonians, unlike their classical cousins. In physics terms, random spin models exhibit signs of glassy behaviors at low temperatures just like classical glasses, while random fermionic models may not~\cite{baldwin2020quenched,anshuetz2024strongly}. In terms of optimization, for random $\CO(1)$-local spin models, a constant-ratio approximation to the optimal energy is achieved by product states, which can be computed and represented efficiently  on a classical computer~\cite{Harrow2017extremaleigenvalues}. Meanwhile, for random $\CO(1)$-local fermionic models, the only quantum algorithm  with rigorous performance bounds is due to~\cite{hastings2022optimizing} (whose guarantee is improved in~\cite{herasymenko2023optimizing}), which also efficiently prepares a quantum state achieving a constant-ratio approximation. Interestingly, unlike the spin case, there is no known efficient classical ansatz with matching performance for the random fermionic model. 

Nevertheless, in both cases where an algorithmic guarantee on the approximation has been shown, the guarantee decays \textit{exponentially} with the locality $k$, if we regard $k$ as a growing parameter, suggesting that these optimization strategies become rapidly less effective the more the model becomes nonlocal. It is perhaps a surprising fact, then, that in the fully nonlocal, random-matrix regime where $k\sim \Theta(n)$, the sparsified random Pauli (spin) models~\cite{chen2023sparse} can be optimized to arbitrarily good energies\footnote{More precisely, the guaranteed energy improves as the sparsity $m$ (number of Pauli terms) is made larger.  } merely by performing phase estimation on the maximally mixed state, a remarkably simple quantum algorithm. This algorithm works because the spectral density of sparse nonlocal models approximately follows a semicircle distribution, implying that a constant fraction of the energy levels lie within a constant-ratio approximation of the optimum. Yet, the states that achieve the constant-ratio approximation are highly entangled, owing to a circuit lower bound in \cite{chen2023sparse}. Unfortunately, the phase estimation strategy is no longer efficient when the locality $k$ is reduced; the spectral density departs from the semicircle distribution and develops long tails (the critical locality where this transition happens is argued to be $k=\Theta(\sqrt{n})$~\cite{erdHos2014phase}). A  motivation of this work is to understand how the random matrix features and algorithmic efficiency interpolate between the local and nonlocal models and how that contrasts with classical intuition.

In classical optimization, a very successful metaheuristic is \textit{simulated annealing}, an adaptation of the Markov chain Monte Carlo (MCMC) algorithm with an annealing schedule that gradually ``lowers the temperature,'' as in metallurgy. Incidentally, there has been a new line of ``quantum MCMC'' algorithms inspired by the dissipative cooling process in nature~\cite{terhal2000problem,temme2011quantum, yung2012quantum,moussa2019low,chen2023QThermalStatePrep,shtanko2021preparing,chen2021fast,wocjan2023szegedy,chen2023efficient, jiang2024quantum,ding2024single,gilyen2024quantum}. While their performance (the \textit{mixing time} problem) on interesting real-world instances remains open, they are known to prepare nontrivial Gibbs states and ground states in finely tuned Hamiltonians~\cite{chen2024local,bergamaschi2024quantum,rajakumar2024gibbs}. It is tempting to ask,
\begin{align*}
    \textit{Can quantum Gibbs samplers efficiently optimize quantum Hamiltonians?}
\end{align*}
The interest in this question is not only about finding better quantum algorithms for solving quantum Hamiltonians, but also about whether quantum MCMC algorithms are a valid general-purpose metaheuristic for solving strongly interacting quantum systems. Physically, the close connection to Nature's thermalization also hints at how such a strongly interacting system may behave in a low-temperature environment.

In this work, we show that a simplified version of quantum MCMC readily achieves a performance that exceeds prior classical and quantum provable guarantees for both spin and fermionic $k$-local models (and their sparsifications) by substantially improving the $k$-dependence of the achieved energy (\cref{thm:intro_algo_performance}). This refined $k$-dependence points to a particularly favorable parameter regime where finding an approximate solution is quantumly easy yet classically nontrivial: sparsified quasilocal spin and local fermionic models, where the approximate solutions must be very entangled in the sense of quantum circuit lower bounds, yet the quantum dissipative algorithm remains efficient and effective.

\section{The $k$-local models}\label{sec:models}
The main models we wish to optimize are random, all-to-all connected, $k$-local spin and fermionic systems on $n$ sites. We will be interested in both the $n$ and the $k$ dependence. While the models may differ in their physical properties, we highlight in~\cref{sec:general_requirement} the relevant features shared among the above models that are the only requirements for our algorithmic guarantee. 

\paragraph{Random Paulis}---
Consider a sum over all $k$-body Paulis on $n$-qubits~\cite{erdHos2014phase},
\begin{align}
    \Hbf = \sum_{\sigmabf:\labs{S(\sigmabf)}=k}g_{\sigmabf} \sigmabf\label{eq:random_Pauli},
\end{align}
where $g_{\sigmabf}$ are i.i.d.~Gaussians with mean 0 and variance $3^{-k} \binom{n}{k}^{-1}$. Here, $\sigmabf$ is an $n$-qubit Pauli operator, and $S(\sigmabf)$ denotes the qubit indices on which $\sigmabf$ acts nontrivially, so that $|S(\sigmabf)|$ is the weight of the Pauli operator. These models are natural quantum analogues of spin glass models~\cite{SK_model_75, derrida1980randomEnergyModel}, where instead of noncommuting Paulis, we have commuting $\vsigma^z$s. 
In the physics literature, these models exhibit signatures of spin-glass phase transition at low enough temperatures~\cite{Sachdev1993,baldwin2020quenched}; thus, preparing states (or compact classical descriptions of states) that achieve a $(1-o(1))$-ratio  approximation of the ground state energy might be challenging for both classical and quantum computers. Nevertheless, making a more refined complexity theoretical characterization at low-temperature remains open. As of today, there is an efficient classical product state that achieves a $\frac{1}{\e^{\Omega(k)}}$-ratio approximation of the optimal energy, and it was not known whether quantum algorithms could be qualitatively better, in terms of the $k$-dependence. Other models of quantum optimization problems that we do not discuss include, e.g., transverse field Ising models~\cite{leschke2021existence} and quantum Max-Cut~\cite{lee2022optimizing,king2023improved,lee2024improved}.

\paragraph{Random fermions}---
The Sachdev--Ye--Kitaev (SYK) model is a random $k$-body fermionic Hamiltonian~\cite{Sachdev1993,Kitaev2015}: for even $k$, 
\begin{align}
    \Hbf = (i)^{k/2} \sum_{S:\labs{S}=k} g_{S}\chi_S, \quad \text{where}\quad \chi_S:= \prod_{i \in S} \chi_{i} \quad \text{and}\quad \chi_i \chi_j + \chi_j \chi_i = 2 \delta_{ij} \label{eq:random_Fermion}
\end{align}
and $g_S$ are i.i.d.~Gaussians with mean 0 and variance $\binom{n}{k}^{-1}$. Here, $\chi_i$ denotes a Majorana operator on site $i \in [n]$,  and $S$ denotes a subset of $[n]$. Originally, the SYK model (and its cousins) was proposed as a toy model for chaotic many-body systems or quantum black holes that thermalize even at low temperatures. 
More quantitative evidence suggests that the model most likely does not exhibit a spin-glass phase transition at very low temperatures~\cite{anshuetz2024strongly,baldwin2020quenched}. For our purposes, we also regard it as a minimal model of interacting fermions~\cite{hastings2022optimizing,herasymenko2023optimizing}. 
Despite the wealth of physics arguments (e.g.,~\cite{maldacena2016remarks}), the best provable efficient algorithm~\cite{hastings2022optimizing,herasymenko2023optimizing} achieves a $\Omega(\frac{1}{e^{\Omega(k)}})$-fraction\footnote{The $k$ dependence is implicit in the result of~\cite[Lemma 27]{herasymenko2023optimizing}, but there is at least an exponential loss in $k$ in ~\cite[Eq. 152]{herasymenko2023optimizing}.} approximation of the optimum.

\paragraph{Sparisified models}---
Based on the above fully connected models, one may sparsify (or ``re-sample'') them to introduce yet another family of random Hamiltonians; this introduces another parameter, being the number of sampled terms $m$ (which can be much smaller than the total number of $k$-body subsets $\binom{n}{k}\sim n^k$). We sample $m$ uniformly random $k$-body terms with random signs (and with replacement) 

\begin{align}
    \Hbf &= \frac{1}{\sqrt{m}} \sum_{j=1}^m \sigmabf_j \quad \text{where}\quad \sigmabf_j \stackrel{i.i.d.}{\sim} \Unif\left(\{\pm\sigmabf: S(\vsigma)=k\}\right), \label{eq:sampled_Pauli}\\
    \Hbf &= (i)^{k/2}\frac{1}{\sqrt{m}} \sum_{j=1}^m \xi_j \quad \text{where}\quad \xi_j \stackrel{i.i.d.}{\sim} \Unif\left(\{\pm\chi_S: \labs{S}=k\}\right). \label{eq:sampled_Fermion}
\end{align}
The benefit of sparsification is that the model can be efficiently simulated on a quantum computer due to a small number of terms $m$, while at the same time retaining relevant properties of the dense model (known as \textit{universality})~\cite{xu2020sparse,chen2023sparse,anschuetz2024bounds}. For our usage, sparsification allows for a parameter regime where $k$ grows with the system size
\begin{align}
    k \gg 1 \quad\text{while}\quad m \ll n^k,
\end{align}
which seems be to an interesting regime where quantum algorithms can produce entangled states that are qualitatively better than what can be achieved with product states.
In the above models, the normalizations are chosen so that the ``variance'' behaves as
\begin{align}
    \BE \vH^2 = \vI,
\end{align}
which will make the energy nonextensive, differing from the physics convention.

\section{Main results}\label{sec:main_results}
\renewcommand{\arraystretch}{1.5}
\begin{table}
\centering
\begin{adjustbox}{width=1\textwidth}
    \begin{tabular}{c|c|c|c|c|c}
         & model & locality $k$ & energy & number of terms $m$ & algorithm \\
         \hline
       \hline       \cite{hastings2022optimizing, herasymenko2023optimizing} & SYK & $\ge 4$ & $\frac{\sqrt{n}}{f(k)}$ & $\binom{n}{k}$, Guassian coeffs. & short-time dynamics \\ \hline
       \cite{herasymenko2023optimizing} & sparse fermion & 4 & $\Omega(\sqrt{n})$ & $\CO(n)$, Gaussian coeffs. & Gaussian states \\ 
         \hline\cite{Harrow2017extremaleigenvalues}  & applied to random Pauli~\eqref{eq:sampled_Pauli} & any & $\pm\frac{\sqrt{n}}{f(k)}$ & $m \ge cn\log(n)/k$ & product states \\
       \hline
       \cite{chen2023sparse} & random Pauli strings & mostly $\Theta(n)$ & $(1-o(1))\cdot \BE\norm{\vH}$ & $\Omega(\poly(n))$ & phase estimation\\
       \hline 
        This paper & Pauli~\eqref{eq:sampled_Pauli} and fermion~\eqref{eq:sampled_Fermion} & any & $\Omega(\frac{\sqrt{n}}{k})$ & $m \ge cn\log(n)/k$ & (simplified) Gibbs sampling\\
       \hline
    \end{tabular}
\end{adjustbox}
\caption{ Summary of prior work considering slight variants of the random fermion/Pauli models, normalized such that the variance is $\BE\vH^2= \vI$. The algorithm~\cite{Harrow2017extremaleigenvalues} actually works for non-random models (but may produce positive or negative energies), so we applied their bounds to the random sampled Pauli model~\eqref{eq:sampled_Pauli}. The two implicit functions $f(k)$ are functions that grow at least exponentially with $k$.
}
\label{table:different_models}
\end{table}
Our main result is a rigorous state preparation guarantee for achieving a good approximation ratio for the random Hamiltonians in~\cref{sec:models}. The particular dissipative algorithm, which is to simulate evolution by a simple Lindbladian chosen to depend on the random Hamiltonian, is defined in~\cref{sec:algorithm}. Our current guarantee for this crude Lindbladian is only proved to work for a short, prescribed evolution time and for the maximally mixed initial state. For the more refined algorithmic Lindbladians (for example, equipped with quantum detailed balance~\cite{chen2023efficient}), we may keep running the algorithm for longer times hoping for better solutions. 

\begin{theorem}[Optimization by dissipation]\label{thm:intro_algo_performance}
    Suppose that $\Hbf$ is drawn randomly from the Gaussian ensembles (\cref{eq:random_Pauli} or \cref{eq:random_Fermion} with $1<k<n$), or the sparsified ensemble (\cref{eq:sampled_Pauli} or \cref{eq:sampled_Fermion} with $m\ge c n\log(n)/k$ for some constant $c$). 
    Then, simulating dynamics by a certain Lindbladian $\lind$ (defined in \cref{eq:lindbladian_def,eq:Kagamma}) starting with the maximally mixed state $\vec{\mu}$ for time $t = \Theta \L(\frac{1}{k}\R)$ achieves average energy 
    \begin{align}
        \EV \Tr[ \e^{\lind t } [\mubf] \Hbf ] \ge \Omega\L(\frac{\sqrt{n}}{k}\R) \ge \Omega\L(\frac{\BE \norm{\vH} }{k}\R),
    \end{align}
where expectation is over the random choice of Hamiltonian.
\end{theorem}

The dissipative algorithm achieves a $\Omega(\frac{1}{k})$-fraction of the optimal energy, giving an exponential improvement on the $k$-dependence over the previous best classical and quantum algorithm guarantees~\cite{Harrow2017extremaleigenvalues,hastings2022optimizing,herasymenko2023optimizing}; see~\cref{table:different_models}. The optimum is provably bounded by $\BE\norm{\Hbf} \le \CO(\sqrt{n})$, 
but actually, we conjecture that its true scaling is $\Theta(\frac{\sqrt{n}}{k})$, which would imply that our dissipative algorithm achieves an $\Omega(1)$ approximation to the optimum on average. The reason to believe this conjecture is that it matches the heuristic prediction of a transition to semicircular ensembles at $ k = \Theta(\sqrt{n})$, where the optimum should be independent of $n$; the displayed bound on the expected optimum $\BE\norm{\vH}\le \CO(\sqrt{n})$ is a loose matrix Hoeffings' that does not take into account the non-commutativity of the summands (see~\Cref{prop:max_eigenvalues} of~\cref{sec:matrix_hoeffding}).

Our improved $k$-dependence is particularly interesting when combined with existing circuit lower bounds from~\cite[Appendix D]{chen2023sparse},~\cite[Appendix E]{anshuetz2024strongly}, which have a qualitatively different scaling with $k$. 
\begin{lemma}[Circuit size lower bounds on low-energy states~\cite{chen2023sparse,anshuetz2024strongly}]\label{lem:counting_lower_bounds}
Fix a description of quantum states (e.g., a sequence of discrete gates or tensor network), and let $G$ be the length of the description. Then, with probability $0.999$, the following hold: \textit{any} high-energy state $\vrho$ that may depend on the Hamiltonian such that 
\begin{align}
     \Tr[\vrho\vH]
    \ge \epsilon \frac{\sqrt{n}}{k}
\end{align}
must have a description length $G$ lower bounded by
\begin{align}
&n \epsilon^2 \e^{\Omega(k)} &\text{for the random Pauli ensemble~\eqref{eq:random_Pauli}},\\
&\Omega_k\L( \epsilon^2 n^{k/2+1}\R)&\text{for the random fermion ensemble~\eqref{eq:random_Fermion}},\\
&\min\left( \Omega(\epsilon \sqrt{nm}/k), n\epsilon^2 \e^{\Omega(k)}\right)&\text{for the sparsified Pauli ensemble~\eqref{eq:sampled_Pauli}},\\
&\min \left( \Omega(\epsilon \sqrt{nm}/k), \Omega_k(\epsilon^2 n^{k/2+1})\right)&\text{for the sparsified fermion ensemble~\eqref{eq:sampled_Fermion}},
\end{align}
where the notation $\Omega_k$ assumes $k$ to be fixed and $n$ growing.
\end{lemma}
For concreteness, let us describe states produced in a fixed circuit architecture of discrete one- and two-qubit gates.\footnote{For continuous gates, an additional $\epsilon$-net argument would be needed to cover all possible states~\cite[Appendix D]{chen2023sparse}.} Then, a product state produced by $n$ single-qubit discrete gates can only achieve a $\frac{\sqrt{n}}{e^{\Omega(k)}}$-ratio approximation of the energy for the spin models; this is rooted in the noncommutativity of the summand and does not have a classical analogue. Furthermore, \textit{any} state that achieves the same approximation ratio as \cref{thm:intro_algo_performance} requires a circuit size growing \textit{(super)exponentially} with $k$, suggesting that the dissipative algorithm has prepared a classically nontrivial state.
While the dense model has a substantial algorithmic cost (the number of terms growing as $\sim n^k/k!$), the sparsified model appears to be in a sweet spot where the quantum algorithm cost remains tractable while retaining sufficient circuit complexity. Indeed, we see that setting, e.g., $k = \log(n)^2, m = \poly(n)$ for the sparse Pauli model allows an arbitrarily large circuit complexity to grow with the number of term $\sqrt{m}$. That is, the states are very entangled, and nevertheless, our algorithm can find them using $\poly(m,n)$ runtime.

This circuit depth lower bound is even stronger for fermions, growing polynomially even at strictly constant values of $k$ (e.g., $k=4$)~\cite{anshuetz2024strongly}; this difference from spins diminishes at logarithmically large values of $k$. Interestingly, this distinction between low-energy states of fermions and spins does not appear in our proof of algorithmic correctness and only appears as obstructions for classical ansatz.

\section{The dissipative algorithm}\label{sec:algorithm}
In this section, we describe our dissipative algorithm, which has a physical origin in the Markovian model of system-bath interactions. We optimize the random Hamiltonian by applying an evolution for a short-time $t\ge 0$ by the \textit{Lindbladian} $\lind$ on the maximally mixed initial state (denoted by $\mubf$), which produces the state
\begin{align}
    \e^{\lind t }(\mubf) \quad \text{where} \quad \mubf := \frac{\Ibf}{\Tr(\Ibf)}.
\end{align}
Recall that a Lindbladian is a quantum continuous-time Markov chain that generates a completely positive and trace-preserving map.
We aim to choose a Lindbladian $\lind$ such that the state obtained has a large energy, given by the Hilbert--Schmidt inner product with the operator $\Hbf$:
\begin{equation}
    \Tr(\Hbf \e^{\lind t }(\mubf) ) = \Tr(\e^{\lind^\dag t }(\Hbf)\mubf ) =\bTr(\e^{\lind^\dag t }(\Hbf))
\end{equation}
where we define the normalized trace
\begin{equation}
    \bTr(\Obf) := \frac{\Tr(\Obf)}{\Tr(\Ibf)}\,.
\end{equation}

The Lindbladian we will choose is inspired by recent proposals for some quantum MCMC algorithms \cite{chen2023QThermalStatePrep,chen2023efficient,ding2024single,gilyen2024quantum, ding2024efficient}, but only keeping minimal features to guarantee a constant-ratio approximation. To define the main Lindbladian, we first choose a set of jump operators $\{\Abf^a\}_{a \in A}$, where $A$ denotes the set of labels for these operators. For the spin models in \cref{eq:random_Pauli,eq:sampled_Pauli}, we choose $\{\Abf^a\}_{a \in A}$ to be the set of all single-site Pauli operators, so $|A| = 3n$. For the fermionic models in \cref{eq:random_Fermion,eq:sampled_Fermion}, we choose $\{\Abf^a\}_{a \in A}$ to be the set of all single-site Majorana operators, so $|A| = n$. Then, we consider the Lindbladian\footnote{We define its adjoint, since we work in the Heisenberg picture, where operators evolve as $\Obf \mapsto \e^{\lind^\dag t}(\Obf)$. }
\begin{equation}\label{eq:lindbladian_def}
    \lind^\dag := \sum_{a\in A} \lind^{a\dag} \quad \text{where}\quad \lind^{a \dag}(\Obf)
    := \left( \Kbf^{a\dag}\Obf\Kbf^a -\frac{1}{2}\{\Kbf^{a\dag}\Kbf^{a},\Obf\}\right),
\end{equation}
where we denote the anticomutator by $\acomm{\vA}{\vB} = \vA\vB+\vB\vA$. 
In our algorithm, we choose the Lindblad operators $\vK^a$ to depend on the Hamiltonian terms, as follows
\begin{equation}
\Kbf^a:= \Abf^a + y [\Abf^a,\Hbf] \quad \text{for tunable}\quad y \in \mathbb{R}. \label{eq:Kagamma}
\end{equation}
One could also let $y$ be complex, but we examine only real values of $y$. 

Unlike the recent quantum MCMC algorithms, this Lindbladian will not satisfy exact or approximate detailed balance, and thus the fixed point of the evolution is not guaranteed to be the Gibbs state. Nevertheless, the Lindblad operators $\Kbf^a$ can be thought of as a truncation of the detailed-balanced Lindblad operators, to first order in $\comm{\Abf^a}{\Hbf}$. The more complex Lindbladians of those works, such as that of \cite{chen2023QThermalStatePrep,chen2023efficient}, typically define a function $\gamma(\omega)$, and have Lindblad operators $\Kbf^{\omega,a}$ labeled by values of $a$ and $\omega$, and defined by
\begin{align}
    \Kbf^{a,\omega} &=  \sqrt{\gamma(\omega)}\int_{-\infty}^{\infty} f(t)\e^{-i\omega t} \e^{i\Hbf t} \vA^a \e^{-i\Hbf t} \rd t
\end{align}
for some suitable function $f(t)$ (which might depend on the inverse temperature $\beta$ or other parameters).
If we expand to leading order in $[\vH,\cdot]$, we obtain
\begin{align}
    \Kbf^{a,\omega} &\propto  \Abf^a - \frac{i \int_{-\infty}^\infty f(t) t\e^{-i\omega t} \rd t}{\int_{-\infty}^\infty f(t)e^{-i\omega t}\rd t}\comm{\Abf^a}{\Hbf} + \ldots. 
\end{align}
By truncating this expansion and identifying the coefficient of $\comm{\Abf^a}{\Hbf}$ with $y$, we recover \cref{eq:Kagamma}. While the $\vA^a$ operator is independent of the Hamiltonian, the local commutator $[\vA^a,\vH]$ contains terms that overlap with $\vA^a$ and already hint at how the energy may change after applying jump $\vA^a.$ Remarkably, this crude approximation is already enough to produce a state that achieves a good approximation ratio to the optimum, giving strong evidence that the more refined algorithms should enjoy a similar performance. Indeed, another construction~\cite{ding2024single} where $\{\vK\} = \int_{-\infty}^{\infty} f(t)\vA^a(t) \rd t$ also gives a similar leading order piece.
\begin{lemma}[Lindbladian simulation cost]
    The dynamics $\e^{\lind t }$ at times $t\ge 0$ can be simulated with $\epsilon$ distance at gate complexity 
    \begin{align}
    \poly\L(n,y,\binom{n}{k},t, \log\L(\frac{1}{\epsilon}\R)\R), \tag{Gaussian models}\\
    \poly\L(n,y,m,t, \log\L(\frac{1}{\epsilon}\R)\R). \tag{Sparsified models}
    \end{align}
\end{lemma}
We did not attempt to minimize the algorithmic cost as it does not change the qualitative findings of this work. The polynomial scaling can be obtained by invoking the algorithm for black-box Lindbladian simulation in, e.g.,~\cite{cleve2016efficient,chen2023QThermalStatePrep} with suitable block-encoding for the jumps $\vK^a$~\cite{gilyen2019quantum}.

\section{The general theorem}\label{sec:general_requirement}
In this section, we write the considered models and the main algorithmic guarantee in a more abstract form, which will imply \Cref{thm:intro_algo_performance}.
\begin{condition}[Symmetric Random Hamiltonians]\label{cond:random_Hamiltonians}
Consider a random Hamiltonian ensemble parameterized by
\begin{align}
    \Hbf = \sum_{\gamma \in \Gamma} s_{\gamma} \Hbf_{\gamma} \quad \text{where} \quad s_\gamma \stackrel{i.i.d.}{\sim} \mathrm{Unif}(\{+1,-1\}). \label{eq:Hgamma}
\end{align}
That is, the terms $\Hbf_\gamma$ can be fixed, or chosen from some (potentially correlated) distribution, but the Rademacher signs $s_\gamma$ are symmetrically distributed and independent. Furthermore, suppose that the terms $\Hbf_\gamma$ are distinct and their squares are proportional to the identity (by a factor denoted $h_\gamma^2$), that is
\begin{align}\label{eq:H_square}
    \Hbf_{\gamma}^2 = h_{\gamma}^2 \Ibf.
\end{align}
\end{condition}
The set $\Gamma$ contains the labels for the terms of $\Hbf$. This condition captures the part of randomness that we actually use in the proof---that is we only require random signs on each term. While we can very likely relax the random Hamiltonian families, the above seems to simplify the calculations and readily captures a large family of interesting instances (see also~\cite{chen2023sparse} and the references therein).

We define the following useful quantities associated with the connectivity of the Hamiltonian. For qubit models, consider the Hilbert space with qubits labeled by $i \in [n] = \{1,\cdots,n\}$, and for fermionic models, consider the Majoranas labeled by $i \in [n] = \{1,\cdots,n\}$.\footnote{Note that they give different dimensions, $2^n$ for qubits and $2^{n/2}$ for Majorana fermions.} We define the \textit{support}
\begin{align}
    S(\gamma)\subset [n] \quad \text{for each Hamiltonian term}\quad \gamma \in \Gamma,
\end{align}
as the set of sites on which the term $\Hbf_{\gamma}$ acts nontrivially. For example, 
\begin{align}
    S(\vsigma^x_1\vsigma^y_2\vI_3) &= \{1,2\}\subset [3].\\
    S(\vI_1\chi_2\vI_3) &= \{2\}\subset [3].
\end{align}

For bookkeepping, we impose the following commutativity condition, which indicates whether a Hamiltonian term commutes with the jump operator $\vA^a$ or not, and whether two distinct Hamiltonian terms commute or anticommute. 

\begin{condition}[Commuting or anticommuting]
\label{cond:commute}
Consider a set of Hamiltonian terms $\{\Hbf_\gamma: \gamma \in \Gamma\}$ and a set of jump operators $\{\vA^a: a \in A\}$ such that for each $a \in A$ we have $\Abf^a = \Abf^{a \dag}$ and $(\Abf^{a})^2 = \Ibf$, and furthermore 
for each $ a \in A$ and $\gamma, \gamma' \in \Gamma$, we have
\begin{align}
    [\Abf^a , \Hbf_{\gamma}] = b_{a\gamma}\cdot 2\Abf^a \Hbf_{\gamma} \quad &\text{where} \quad b_{a\gamma} \in \{0, 1\}, \\
    [\Hbf_{\gamma} , \Hbf_{\gamma'}] = b_{\gamma\gamma'}\cdot 2 \Hbf_{\gamma}\Hbf_{\gamma'} \quad &\text{where} \quad b_{\gamma\gamma'} \in \{0, 1\}.
\end{align}
Additionally, if the support of the Hamiltonian terms is disjoint $S(\gamma)\cap S(\gamma')=\emptyset,$ then $b_{\gamma\gamma'}=0$, and similarly if the support of $\Abf^a$ is disjoint from $S(\gamma)$, then $b_{a\gamma} = 0$. 
\end{condition}
If $b_{a\gamma} = 0$, then $\Abf^a$ and $\Hbf_\gamma$ commute; otherwise, they anticommute (and similar for $b_{\gamma\gamma'}$). For any subset $Z \subset [n]$, let $A_Z \subset A$ contain labels $a$ for which the support of operator $\Abf^a$ has nonempty intersection with $Z$. 

\begin{condition}[Locality of jumps and Hamiltonian terms]\label{cond:aloc}
Suppose that each of the jump operators $\{\vA^a: a \in A\}$ is supported on one site. Furthermore, suppose that there is a fixed locality parameter $k$ for which $|S(\gamma)| = k$ for all $\gamma$, and that there are absolute constants $\aloc$ and $\aac$ (with $\aac \leq \aloc$) such that 
\begin{align}\label{eq:b_agamma_def}
    \sum_{a\in A} b_{a \gamma} &= \aac k \quad \text{for each}\quad \gamma. \\
    |A_Z| &= \aloc |Z|\quad \text{for each}\quad Z \subset [n].
\end{align}
\end{condition}
Roughly, the first condition counts the number of jumps that anticommute with a term $\vH_{\gamma}$, which scales with $k$, and the second counts how many single site operators there are.
The following lemma shows that, given the models we consider and our choice of jump operators, all of these conditions are simultaneously satisfied.
\begin{lemma}
    Suppose that either (i) the set of jump operators $\{\Abf^a: a \in A\}$ is the set of $3n$ single-site Pauli operators, and $\Hbf$ is chosen from the random Pauli ensemble (\cref{eq:random_Pauli}) or the sampled Pauli ensemble (\cref{eq:sampled_Pauli}), or (ii) the set of jump operators $\{\Abf^a: a \in A\}$ is the set of $n$ single-site Majorana operators, and $\Hbf$ is chosen from the random fermion ensemble (\cref{eq:random_Fermion}) or the sampled fermion ensemble (\cref{eq:sampled_Fermion}) with even $k$.  Then Conditions \ref{cond:random_Hamiltonians}, \ref{cond:commute}, and \ref{cond:aloc} are satisfied, and
    \begin{align}
        \aloc &= \begin{cases}
            3 & \text{if } \Hbf \text{ chosen from \cref{eq:random_Pauli} or \eqref{eq:sampled_Pauli}} \\
            1 & \text{if } \Hbf \text{ chosen from \cref{eq:random_Fermion} or \eqref{eq:sampled_Fermion}}
        \end{cases}\\
        \aac &= \begin{cases}
            2 & \text{if } \Hbf \text{ chosen from \cref{eq:random_Pauli} or \eqref{eq:sampled_Pauli}} \\
            1 & \text{if } \Hbf \text{ chosen from \cref{eq:random_Fermion} or \eqref{eq:sampled_Fermion}}
        \end{cases}
    \end{align}
\end{lemma}
\begin{proof}
    The symmetry of $\Hbf$ required by \Cref{cond:random_Hamiltonians} is true in all cases by construction. For the Pauli models, each $\Hbf_\gamma$ is proportional to a multiqubit Pauli operator; thus, each single-site Pauli operator $\Abf^a$ acting on site $i$ will anticommute with $\Hbf_\gamma$ if and only if the Pauli operator in position $i$ of $\Hbf$ is not the same as that of $\Abf^a$ and not the identity. For each $i \in S(\gamma)$ there are two single-site Pauli operators $\Abf^a$ that meet this criteria, so $\aac =2$. For the fermionic models, each $\Hbf_\gamma$ is proportional to a product of an even number of Majorana operators; thus, the single-site Majorana operator $\Abf^a$ acting on site $i$ will anticommute with $\Hbf_\gamma$ if and only if $i \in S(\gamma)$, so $\aac = 1$. Meanwhile, the total number of Pauli operators that act on a site within a subset $Z$ is precisely $3|Z|$, and the number of Majoranas acting on $Z$ is $|Z|$, confirming the quoted values of $\aloc$.  Thus, in all cases \Cref{cond:commute} and \Cref{cond:aloc} are also met. Note that in the fermionic case, if the supports of two even-weight terms are disjoint, then the commutator vanishes (which is false if both have odd weights), just like the Pauli case. Similarly, any single-site Majorana that acts on a site not in $S(\gamma)$ commutes with $\Hbf_{\gamma}$ when $\Hbf_{\gamma}$ is of even weight.  
\end{proof}

It is also useful to define the local and global energies:
\begin{definition}[local and global energies]\label{def:local_global_energies}
For a random Hamiltonian (\cref{cond:random_Hamiltonians}), define the local and global quantities
\begin{align}\label{eq:Hloc_Hglo_def}
\Hloc: = \max_{1\le i \le n}  \sqrt{\sum_{\gamma: i \in S(\gamma)} h_{\gamma}^2} \quad \text{and} \quad\Hglo: = \sqrt{\sum_{\gamma\in \Gamma} h_{\gamma}^2}.
\end{align}
\end{definition}
The global energy $\Hglo$ is the root-mean-square of the strength of every Hamiltonian term, and the local energy $\Hloc$ is a root-mean-square over terms whose support contains the site $i$ (maximizing over $i$). 

Now, we can give the abstract algorithmic guarantee, given the abstract conditions.
\begin{theorem}[Algorithmic performance]\label{thm:algo_performance}
    Consider a random Hamiltonian $\Hbf$ and choice of jumps $\{\Abf^a\}$ satisfying \cref{cond:random_Hamiltonians}, \cref{cond:commute}, and \cref{cond:aloc}. Then, there exist universal constants\footnote{The constants $c_y$, $c_t$, and $c_H$ have dependence on the constants $\aloc$ and $\aac$ defined in \Cref{cond:commute}, but not on $k$, $m$, or $n$.} $c_y$, $c_t$, and $c_{H}$ such that, setting
    \begin{align}\label{eq:y_t_choice}
    y = -\frac{c_y}{\sqrt{k} \Hloc}, \quad t = \frac{c_t}{k},
\end{align} 
starting with the maximally mixed state $\vec{\mu}$, and evolving by the Lindbladian $\lind$ in \cref{eq:lindbladian_def}, 
produces a state achieving average energy
    \begin{align}
        \EV \Tr[ \mubf \e^{\lind^\dag t }(\Hbf) ] \ge \frac{c_H}{\sqrt{k}} \frac{\HgloPower{2}}{\Hloc}.
    \end{align}
\end{theorem}
Remarkably, our results treat the fermionic and Pauli models quite equally, despite the possibly very different physics (\cite{baldwin2020quenched,anshuetz2024strongly}). We recover the earlier~\cref{thm:intro_algo_performance} by evaluating $\Hglo$ and $\Hloc$ for the four ensembles, using the following proposition (see \cref{sec:compuing_localnorm} for a direct computation). 
\begin{proposition}\label{prop:Hglo_Hloc_bound}
Suppose that $\Hbf$ is drawn randomly from the ensembles in \cref{eq:random_Pauli} or \cref{eq:random_Fermion} with $1<k<n$, or it is drawn randomly from the ensembles in \cref{eq:sampled_Pauli} or \cref{eq:sampled_Fermion} with $m\ge cn\log(n)/k$ for some constant $c$. Then
\begin{align}
    \BE_{\Hbf} \L(\frac{\HgloPower{2}}{\Hloc}\R) =\Omega\left( \sqrt{\frac{n}{k}} \right),
\end{align}
where the expectation value is taken over random choice of $\Hbf$. 
\end{proposition}

\section{Proof of the main theorem}\label{sec:algo_performance_proof}
We provide the proof of our main result (\Cref{thm:algo_performance}). The main idea is to expand the expected energy (averaged also over random choice of $\Hbf$) as the sum of a zeroth order term that vanishes under expectation, a linear-in-$t$ term that can be computed exactly, and a $\CO(t^2)$ term for which the magnitude can be upper bounded. At small $t$, the $\CO(t)$ term dominates, and it is possible to choose a value of $t$ and $y$ that yields the desired result. This balancing of the first- and second- order terms is depicted schematically in \cref{fig:balancing}. 

\begin{figure}[h]
\centering
\tikz[scale=1.0,domain=0:14,samples=50]{
    \begin{axis}[
    axis equal image,
    ymin=0,
    ymax=8.5,
    xmin=0,
    xmax=12,
    xtick={6.82}, 
    xticklabel = {$\scriptstyle \Theta(1/k)$},
    ytick=\empty, 
    every tick/.style={
        black,
        semithick,
      },
    axis on top = false,
    axis x line=middle, 
    axis y line=middle, 
    stack plots=y, 
    xlabel = $t$, 
    ylabel = $\BE \bTr\;\e^{\lind^\dag t}(\Hbf)$,
    x label style={at={(axis description cs:1.05,0.03)},anchor=north},
    y label style={at={(axis description cs:0,1.0)},rotate=0,anchor=south},
    extra description/.code={
\node [left] at (axis cs:0,1.022) {feasible $\scriptstyle \Theta(\sqrt{n}/k)$};
\node [left] at (axis cs:0,3) {true optimum};
\node [left] at (axis cs:0,7.5) {upper bound $\scriptstyle \Theta(\sqrt{n})$};
}
    ]
        \filldraw[red,fill=red] (axis cs:6.82,1.022) circle [radius=1.2];
        \addplot[mark=none,draw=black,ultra thick] {7.5};
        \addplot[mark=none,draw=black,thick,dotted] {-7.5+1.022};
        \addplot[mark=none,draw=black,thick] {-1.022+0.3*x};
        \addplot[mark=none, draw = none] {0.022*x^2} ;
        \addplot[mark=none,fill=cyan,draw=white] {-0.044*x^2} \closedcycle;
    \end{axis}
}
\caption{\label{fig:balancing} Depiction of the balancing of first- and second- order terms in our proof. The expected maximum energy $\BE \lambda_{\max}(\Hbf)$ is upper bounded by $\CO(\sqrt{n})$ (thick black line) and conjectured to be $\Theta(\sqrt{n}/k)$. The lower bound on the expected achievable energy is derived by expanding $\bTr[\e^{\lind^\dag t}(\Hbf)]$ as a sum of a linear-in-$t$ term (solid black line) and an $\CO(t^2)$ remainder term. The linear-in-$t$ term is computed exactly and the magnitude of the $\CO(t^2)$ term is upper bounded (indicated by shaded blue region). Since we seek a lower bound on $\bTr[\e^{\lind^\dag t}(\Hbf)]$, we choose $t$ such that the bottom edge of the blue region is maximized, indicated by the red dot, occurring at $t= \Theta(1/k)$ and $\bTr[\e^{\lind^\dag t}(\Hbf)] = \Theta(\sqrt{n}/k)$. If $t$ is chosen larger than this, the bound becomes worse; this contrasts with the expected behavior of detailed-balance Gibbs sampling algorithms, which are expected to improve strictly over time.}
\end{figure}

The magnitude of the $\CO(t^2)$ term is upper bounded in \Cref{prop:second_term_final_bound}; its proof requires a subproposition (\Cref{prop:superoperator_norm}), which is proved with a diagrammatic method in \Cref{sec:diagrammatic_calculus}. We provide intuition for \Cref{prop:superoperator_norm} in \Cref{sec:intuition_second_order_term}.

\subsection{Proof of \Cref{thm:algo_performance}}

In this section, we provide the proof of~\cref{thm:algo_performance}, modulo the assertion of some propositions, whose proofs are deferred. We begin by emphasizing the introduction of the Rademacher variables $s_\gamma \in \{+1,-1\}$ for $\gamma \in \Gamma$ that make the random signs explicit in \cref{cond:random_Hamiltonians}:
\begin{align}
    \Hbf = \sum_\gamma s_\gamma \Hbf_\gamma\quad \text{where}\quad s_\gamma \stackrel{i.i.d.}{\sim} \mathrm{Unif}(\{+1,-1\}).
\end{align}
In fact, in our proofs, we will only make use of the randomness on the signs. Henceforth, all expectation values $\BE$ that appear randomize only over the $s_\gamma$ variables, and we may consider the quantities $\Hbf_\gamma$ as fixed. 
Consequently, we can organize the Lindbladian in \cref{eq:lindbladian_def} as a degree-two polynomial of the Rademachers
\begin{align}\label{eq:lind^a_as_Rademacher_polynomial}
    \lind^{a \dag} = \lind^{a \dag}_0 + \sum_{\gamma \in \Gamma} s_{\gamma}\lind^{a\dag}_{\gamma} + \sum_{\substack{\gamma,\gamma' \in \Gamma \\ \gamma\ne\gamma'}} s_{\gamma}s_{\gamma'}\lind^{a\dag}_{\gamma \gamma'},
\end{align}
where
\begin{align}
    \lind_0^{a \dag} (\Obf) & = \Abf^{a \dag} \Obf \Abf^a + \sum_\gamma y^2 \comm{\Abf^a}{\Hbf_\gamma}^\dag \Obf \comm{\Abf^a}{\Hbf_\gamma} -\frac{1}{2} \acomm{\id + y^2 \sum_\gamma \comm{\Abf^a}{\Hbf_\gamma}^\dag \comm{\Abf^a}{\Hbf_\gamma}}{\Obf}, \label{eq:lind_0^a}\\
    \lind_\gamma^{a \dag} (\Obf) & = y \comm{\Abf^a}{\Hbf_\gamma}^\dag \Obf \Abf^a + y \Abf^{a \dag} \Obf \comm{\Abf^a}{\Hbf_\gamma} -\frac{1}{2} y \acomm{\comm{\Abf^a}{\Hbf_\gamma}^\dag \Abf^a}{\Obf} -\frac{1}{2} y\acomm{\Abf^{a \dag} \comm{\Abf^a}{\Hbf_\gamma}}{\Obf}, \label{eq:lind_gamma^a}\\
    \lind_{\gamma \gamma'}^{a \dag} (\Obf) & = \delta_{\gamma\neq\gamma'} y^2 \left( \comm{\Abf^a}{\Hbf_\gamma}^\dag \Obf \comm{\Abf^a}{\Hbf_{\gamma'}} -\frac{1}{2} \acomm{\comm{\Abf^a}{\Hbf_\gamma}^\dag \comm{\Abf^a}{\Hbf_{\gamma'}}}{\Obf} \right). \label{eq:lind_gammagamma'^a}
\end{align}
Note that $\lind^{a \dag}_0$ contains the $\lind_{\gamma\gamma}^{a\dag}$ terms since $s_{\gamma}^2 =1$. 

The quantity of interest is the expectation value of $\Hbf$ in the evolved state, averaged over the random choice of the Rademachers, that is
\begin{equation}
    \BE \bTr[\e^{\lind^\dag t}(\Hbf)]  = \sum_{\gamma \in \Gamma} \BE s_\gamma \bTr[\e^{\lind^\dag t}(\Hbf_\gamma)]. 
\end{equation}
Note that $\lind^\dag$ has hidden dependence on the Rademachers $s_\gamma$ through its dependence on $\Hbf$. By invoking the fundamental theorem of calculus twice, we may write
\begin{align}
    \e^{\lind^\dag t}(\Hbf) &= \Hbf + \int_{0}^t \rd t_2 \left[\frac{ \rd}{ \rd u} \e^{\lind^\dag u}(\Hbf)\right]_{u=t_2} = \Hbf + \int_{0}^t \rd t_2 \e^{\lind^\dag t_2}(\lind^\dag(\Hbf)) \\
    &= \Hbf + {t\lind^\dag(\Hbf)} +\int_{0}^t \rd t_2 \int_{0}^{t_2} \rd t_1 \e^{\lind^\dag t_1}(\lind^\dag(\lind^\dag(\Hbf))).
\end{align}
Thus, we can expand the average energy into a zeroth-, first-, and second-order term in $t$:
\begin{align}
    \EV \bTr[\e^{\lind^\dag t}(\Hbf)] &= \sum_{\gamma \in \Gamma} \EV \bTr[\e^{\lind^\dag t}(s_\gamma \Hbf_\gamma)] \nonumber  \\
    &= \undersetbrace{=0}{\sum_{\gamma \in \Gamma} \EV s_\gamma \bTr [\Hbf_\gamma]} +\undersetbrace{\text{first-order term } T_1}{ \sum_{\gamma \in \Gamma} \EV s_\gamma \bTr [\lind^\dag( \Hbf_\gamma)] t} + \undersetbrace{\text{second-order term } T_2}{\sum_{\gamma \in \Gamma} \EV s_\gamma \bTr \left[ \int_{t>t_2>t_1>0} \e^{\lind^\dag t_1} (\lind^\dag)^2 (\Hbf_\gamma) 
    \rd t_1 \rd t_2 \right]}. \label{eq:linear_quadratic_terms}
\end{align}
The zeroth order term vanishes under averaging over the Rademachers. The other two terms are handled separately. 

First, we compute the first-order term $T_1$ exactly. 
\begin{proposition}[First order term]\label{prop:first_order_term}
We can evaluate the quantity $T_1$ (see \cref{{eq:linear_quadratic_terms}}) as
    \begin{align}
         T_1 &= -8 yt \HgloPower{2} \aac k.
    \end{align}
\end{proposition}
\begin{proof}
We may write
\begin{align}
    T_1 := t \sum_{\gamma \in \Gamma} \EV \bTr [\lind^\dag(s_\gamma \Hbf_\gamma)] & = t\sum_{\gamma \in \Gamma} \EV \bTr \left[ \sum_{a \in A} \left( \lind_0^{a \dag} + \sum_{\beta \in \Gamma} s_\beta \lind_\beta^{a \dag} + \sum_{\substack{\beta,\beta' \in \Gamma \\ \beta \neq \beta'}} s_\beta s_{\beta'} \lind_{\beta \beta'}^{a \dag} \right) (s_\gamma \Hbf_\gamma) \right] \\
    & = t\sum_{\gamma \in \Gamma}  \sum_{a \in A} \bTr \lind_\gamma^{a \dag} (\Hbf_\gamma),
\end{align}
where the averaging over Rademachers killed most terms:
\begin{align}
\EV s_\beta s_\gamma =\delta_{\beta\gamma}\quad \text{and}\quad \EV s_\beta s_{\beta'} s_\gamma = 0\quad \text{for all} \quad \beta, \beta', \gamma.  
\end{align}
Here, $\delta_{\beta\gamma}$ denotes the Kronecker delta symbol. 
Applying \cref{eq:lind_gamma^a}, we find that
\begin{align}
    \bTr [\lind^{a\dag}_\gamma(\Hbf_\gamma)] & = y \comm{\Abf^a}{\Hbf_\gamma}^\dag \Hbf_\gamma \Abf^a + y \Abf^{a \dag} \Hbf_\gamma \comm{\Abf^a}{\Hbf_\gamma} -\frac{1}{2} y \acomm{\comm{\Abf^a}{\Hbf_\gamma}^\dag \Abf^a}{\Hbf_\gamma} -\frac{1}{2} y\acomm{\Abf^{a \dag} \comm{\Abf^a}{\Hbf_\gamma}}{\Hbf_\gamma}
\end{align}
which is easily seen to vanish in the case that $\Abf^a$ and $\Hbf_{\gamma}$ commute (i.e., $b_{a\gamma} = 0$). Therefore, the expression can be evaluated for the case that $\Abf^a$ and $\Hbf_\gamma$ anticommute (i.e.,~$b_{a\gamma}=1$), and expressed as 
\begin{align}
    \bTr [\lind^{a\dag}_\gamma(\Hbf_\gamma)]  &= -8 b_{a\gamma} h_\gamma^2 y \bTr[\Abf^{a \dagger} \Abf^a] = -8 b_{a\gamma} h_\gamma^2 y
\end{align}
where we have used the assumption that $\Hbf_\gamma^2 = h_\gamma^2 \Ibf$ from \Cref{cond:random_Hamiltonians} and $\Abf^{a \dagger} \Abf^a = \Ibf $ from \Cref{cond:commute}.
Finally, summing over $a \in A$ and $\gamma \in \Gamma$, we can exactly evaluate the first-order term as
\begin{align}\label{eq:first_term_radem}
    T_1 = \sum_{\gamma \in \Gamma} \sum_{a \in A} \bTr [ \lind_\gamma^{a \dag} (\Hbf_\gamma) t ]&= -8 y t \sum_{\gamma \in \Gamma} h_\gamma^2 \left(\sum_{a \in A} b_{a\gamma} \right) \\
    &= -8 yt \HgloPower{2} \aac k,
\end{align}
as advertised, where we have utilized the definition of $\aac$ in \cref{cond:aloc}.
\end{proof}

For the second-order term, we state the upper bound and defer its proof to the following section.

\begin{proposition}[Bounds on the second order term]\label{prop:second_term_final_bound}
Suppose that $y$ and $t$ satisfy the relations
\begin{align}
 t &\ge 0, \\
    \aloc k t &< 1, \\
   y^2\HlocPower{2} \aloc k &< 1/8.
\end{align}
Then, there is a universal constant $c_2$, such that, with $T_2$ defined as in \cref{eq:linear_quadratic_terms}, we have
    \begin{equation}
        |T_2| \leq c_2 |y| t^2 \alocPower{2} k^2 \HgloPower{2}.
    \end{equation}
\end{proposition} 

Now, we explain how $y$ and $t$ can be chosen in such a way to entail the theorem statement. We let $y = -c_y/(\sqrt{k} \Hloc)$, $t = c_t / k$ for universal constants $c_y$ and $c_t$ chosen as
\begin{align}
    c_t &= \min(1/(2\aloc)\,,\; 4\aac /(c_2 \alocPower{2})) \\
    c_y &= 1/(3\sqrt{\aloc})\,.
\end{align}
These choices ensure that the relations above are met, and also that
\begin{equation}
    |T_2| \leq 4|y| t \aac k \HgloPower{2}  = |T_1|/2.
\end{equation}
Since $|T_2| \leq |T_1|/2$, we can conclude that 
\begin{equation}\label{eq:T1_T2_balancing}
    T_1 + T_2 \geq \frac{1}{2} T_1 = -4 y t \HgloPower{2} \aac k = \frac{4c_y c_t \aac}{\sqrt{k}}\frac{\HgloPower{2}}{\Hloc} = \frac{c_H}{\sqrt{k}}\frac{\HgloPower{2}}{\Hloc}\,,
\end{equation}
where we choose $c_H = 4c_y c_t \aac$, 
proving \Cref{thm:algo_performance}.

\subsection{Proof of bound on second-order term (\cref{prop:second_term_final_bound})}\label{sec:proof_second_term_bound}

The goal is to upper bound the magnitude of the second-order term defined in \cref{eq:linear_quadratic_terms}.
This term can be expressed as
\begin{align}\label{eq:all_cases}
&\sum_\gamma\EV \bTr \int_{0}^t \rd t_2 \int_{0}^{t_2} \rd t_1 \e^{\lind^\dag t_1} \Bigg[ \nonumber\\
&\qquad \left( \lind_0^\dag+ \sum_{\gamma_2 \in \Gamma} s_{\gamma_2} \lind_{\gamma_2}^\dag + \frac{1}{2}\sum_{\substack{\gamma_2, \gamma_2' \in \Gamma \\\gamma_2\neq\gamma'_2}} s_{\gamma_2} s_{\gamma'_2} \lind_{\gamma_2\gamma'_2}^\dag \right) 
\left( \lind_0^\dag+ \sum_{\gamma_1} s_{\gamma_1} \lind_{\gamma_1}^\dag + \frac{1}{2}\sum_{\substack{\gamma_1,\gamma_1' \in \Gamma \\\gamma_1\neq\gamma'_1}} s_{\gamma_1} s_{\gamma'_1} \lind_{\gamma_1\gamma'_1}^\dag \right) (s_\gamma\Hbf_\gamma) \Bigg],
\end{align}
where we abbreviated
\begin{align}\label{eq:lind_0/gamma/gammagamma'_def}
    \lind_0^\dag:= \sum_{a \in A} \lind_0^{a \dag}, \qquad \lind_\gamma^\dag := \sum_{a \in A} \lind_\gamma^{a \dag}, \qquad \lind_{\gamma\gamma'}^\dag := \sum_{a \in A} \L(\lind_{\gamma\gamma'}^{a \dag} + \lind_{\gamma'\gamma}^{a \dag}\R).
\end{align}
By picking a term from each parenthesis, we arrive at $9$ terms to evaluate. If we choose the final term in each parenthesis, we end up with a term with 5 Rademacher variables $s_{\gamma_2}s_{\gamma_2'}s_{\gamma_1}s_{\gamma_1'}s_\gamma$ that can be brought outside of the integral---this is the maximum number of Rademachers for any term. 

We may group terms in the expression for $T_2$ from \cref{eq:all_cases} based on which Rademacher variables appear in the term:
\begin{equation}\label{eq:expansion_T2_rademachers}
    T_2 = \sum_{\substack{U \subset \Gamma \\ |U| \leq 5}} \BE\left[\left(\prod_{\alpha \in U} s_\alpha\right) \bTr \int_{t > t_2> t_1> 0} \rd t_1 \rd t_2 \e^{\lind^\dag t_1}\Obf_U \right].
\end{equation}
Each term is labeled by a subset $U \subset \Gamma$ with at most 5 elements, denoting which Rademachers $s_\gamma$ appear (note that the Lindbladian $\lind$ also implicitly contains a dependence on all the Rademachers, but these are in the exponent and cannot be brought outside the integral). The operators $\Obf_U$ depend on Lindbladian terms, but they are independent of the random Rademacher variables. For example for $U = \{\gamma_1,\gamma_2\}$, the associated term looks like
\begin{align}
    \Obf_{\{\gamma_1,\gamma_2\}} &= \lind_{\gamma_1}^\dag\lind_0^\dag(\Hbf_{\gamma_2}) +\lind_{\gamma_2}^\dag\lind_0^\dag(\Hbf_{\gamma_1}) + \lind_0^\dag\lind_{\gamma_1}^\dag (\Hbf_{\gamma_2}) +\lind_{0}^\dag\lind_{\gamma_2}^\dag(\Hbf_{\gamma_1}) \nonumber \\
    &\qquad + \sum_{\gamma_3 \in \Gamma} \left(\lind_{\gamma_1\gamma_2}^\dag\lind_{\gamma_3}^\dag(\Hbf_{\gamma_3}) + \lind_{\gamma_1\gamma_3}^\dag\lind_{\gamma_2}^\dag(\Hbf_{\gamma_3}) + \lind_{\gamma_2\gamma_3}^\dag\lind_{\gamma_1}^\dag(\Hbf_{\gamma_3}) +\lind_{\gamma_1\gamma_3}^\dag\lind_{\gamma_3}^\dag(\Hbf_{\gamma_2}) + \lind_{\gamma_2\gamma_3}^\dag\lind_{\gamma_3}^\dag(\Hbf_{\gamma_1})\right) \nonumber\\
&\qquad + \sum_{\gamma_3 \in \Gamma} \left(\lind_{\gamma_3}^\dag\lind_{\gamma_1\gamma_2}^\dag(\Hbf_{\gamma_3}) + \lind_{\gamma_2}^\dag\lind_{\gamma_1\gamma_3}^\dag(\Hbf_{\gamma_3}) + \lind_{\gamma_1}^\dag\lind_{\gamma_2\gamma_3}^\dag(\Hbf_{\gamma_3}) +\lind_{\gamma_3}^\dag\lind_{\gamma_1\gamma_3}^\dag(\Hbf_{\gamma_2}) + \lind_{\gamma_3}^\dag\lind_{\gamma_2\gamma_3}^\dag(\Hbf_{\gamma_1})\right).
\end{align}
and for $U = \varnothing$, we have
\begin{align}
    \Obf_{\varnothing} = \sum_{\gamma \in \Gamma}\L(\lind^\dag_0\lind^{\dag}_\gamma (\Hbf_{\gamma}) + \lind^\dag_\gamma\lind^{\dag}_0 (\Hbf_{\gamma})\R) + \sum_{\substack{\gamma_1,\gamma_2\in \Gamma \\ \gamma_1 \neq \gamma_2} }\L(\lind^\dag_{\gamma_2}\lind^{\dag}_{\gamma_1 \gamma_2} (\Hbf_{\gamma_1}) + \lind^\dag_{\gamma_1\gamma_2}\lind^{\dag}_{\gamma_2} (\Hbf_{\gamma_1})\R)
\end{align}
Since these $\Obf_U$ do not depend on the random Rademachers, we may upper bound the value of each term by a product of norms:
\begin{equation}\label{eq:abs_T2_two_norms}
    |T_2| \leq \sum_{\substack{U \subset \Gamma \\ |U| \leq 5}}\norm{ \BE\left[\left(\prod_{\alpha \in U} s_\alpha\right) \bTr \int_{t > t_2> t_1> 0} \rd t_1 \rd t_2 \e^{\lind^\dag t_1}(\cdot)\right]}\cdot\norm{\Obf_U} 
\end{equation}
where the first norm denotes a norm for the linear functional
\begin{align}
    \norm{\bTr[\CM^{\dagger}(\cdot)]} :=\sup_{\vO}\frac{\labs{\bTr[\CM^{\dagger}(\vO)]}}{\norm{\vO}} = \sup_{\vO}\frac{\labs{\Tr[\CM(\mubf) \Obf]}}{\norm{\vO}} = \norm{\CM(\vec{\mubf})}_1,
\end{align}
for superoperator $\CM$ being the expected-integral over signed Lindbladian dynamics.\footnote{If there were no signs, this would have norm exactly one. The random signs further decrease the norm, which we need to capture carefully.} 

This expression can be bounded in two steps. First, we are able to upper bound the norm of the trace-expected-integral superoperator, as follows. (Recall the definition of $h_\alpha$ as the strength of Hamiltonian term $\Hbf_\alpha$ in \cref{cond:random_Hamiltonians}.)
\begin{proposition}\label{prop:superoperator_norm}
Fix a subset $U \subset \Gamma$ with $|U| \leq 5$. Suppose that $y$ and $t$ satisfy the relations     
$t\ge0$, $\aloc k t < 1$, and $y^2 \HlocPower{2} \aloc k < 1/8 $. Then 
\begin{align}
    \norm{ \BE\left[\left(\prod_{\alpha \in U} s_\alpha\right) \bTr \int_{t > t_2> t_1> 0} \rd t_1 \rd t_2 \e^{\lind^\dag t_1}(\cdot)\right]} \leq \CO\L(t^2\R) \cdot |y|^{|U|} \cdot \prod_{\alpha \in U} h_\alpha
\end{align}
\end{proposition}

Note that the Rademacher signs are all we needed to obtain a good estimate, and not the particular operators $\vO_U$. We sketch the proof of this in the following subsection, and give the formal proof in \Cref{sec:proof_superoperator_norm}.
Using this proposition, we can return to \cref{eq:abs_T2_two_norms} and write
\begin{equation}
    |T_2| \leq \CO(t^2) \cdot  \sum_{\substack{U \subset \Gamma \\ |U| \leq 5}}\left(\prod_{\alpha \in U} h_\alpha\right) |y|^{|U|} \norm{\Obf_U}.
\end{equation}

The final step is to evaluate this sum, using the following proposition. Here, there are no Lindbladian exponentials anymore, but merely some finitely many terms arising from $\CL^{\dagger}(\CL^{\dagger}(\vH))$ (with the averaging over Radamachers already handled in~\Cref{prop:superoperator_norm}). We sum over the terms in an almost brute-force fashion, paying special attention to the commutation structures of $\vH$ and $\{\vA^a\}$. 
\begin{proposition}\label{prop:sum_over_operators}
Let $\Obf_U$ be defined as in the expansion of \cref{eq:expansion_T2_rademachers} for each subset $U\subset \Gamma$, with $|U|\leq 5$. Suppose that $y$ satisfies the relation $y^2 \HlocPower{2} \aloc k < 1/8 $. Then we have
\begin{align}
&\sum_{\substack{U \subset \Gamma \\ |U| \leq 5}}\left(\prod_{\alpha \in U} h_\alpha\right) |y|^{|U|} \norm{\Obf_U} \nonumber 
\leq{} \CO(1) \cdot |y|\alocPower{2} k^2 \HgloPower{2}.
\end{align}
\end{proposition}
This proposition is proved in \Cref{sec:proof_sum_over_operators}. 
Utilizing the proposition, we obtain the desired result:
\begin{equation}
    |T_2| \leq \CO(t^2) \cdot |y| \alocPower{2} k^2 \HgloPower{2}\,.
\end{equation}

\subsection{Expansion of Rademacher expressions and proof intuition}\label{sec:intuition_second_order_term}

Here we sketch the idea behind the calculation of the second-order term $T_2$, captured formally in \Cref{prop:superoperator_norm} and \Cref{prop:sum_over_operators}. To illustrate, we first examine a representative contribution to $T_2$ from \cref{eq:expansion_T2_rademachers}, taking $U = \{\gamma_1,\gamma_2\}$ and considering only one of the terms of $\Obf_{U}$: 
\begin{align}
    \EV \left[ s_{\gamma_1}s_{\gamma_2} \bTr \int_{t>t_2>t_1>0} \rd t_1 \rd t_2 \e^{\lind^\dag t_1} \lind^\dag_{\gamma_2}(\lind_0^\dag (\Hbf_{\gamma_1}))\right].
\end{align}
A term like this will appear for all choices of $\gamma_1$ and $\gamma_2$, so we would like to be able to compute quantities where these are summed over, for example
    \begin{align}\label{eq:CT}
    \CT := \sum_{\substack{\gamma_1,\gamma_2 \in \Gamma \\ \gamma_1 \neq \gamma_2}}\EV \left[ s_{\gamma_1}s_{\gamma_2} \bTr \int_{t>t_2>t_1>0} \rd t_1 \rd t_2 \e^{\lind^\dag t_1} \lind^\dag_{\gamma_2}(\lind_0^\dag (\Hbf_{\gamma_1}))\right].
\end{align}
This expression has complicated Rademacher dependence in the exponent of the $\e^{\lind^\dag t}$ term, and it is unclear how it can be bounded. The following lemma will play a part in that.

\begin{lemma}[Lindbladian evolution is norm contractive]
\label{lem:L_contract}
A Lindbladian evolution in the Heisenberg picture contracts the operator norm:
\begin{align}
\norm{\e^{\CL^\dag t}(\vO)} \le \norm{\vO} \quad \text{for any}\quad t \ge 0\quad \text{and operator}\quad \vO.
\end{align}
\end{lemma}
This holds since quantum channels are trace norm contractive in the Schr\"odinger picture, or operator norm contractive in the Heisenberg picture (see, e.g., ~\cite{perez2006contractivity}.)

\textbf{A naive attempt.} Using the triangle inequality and the fact that the normalized trace satisfies $\bTr(\Obf) \leq \norm{\Obf}$, we could try bounding \cref{eq:CT} as
\begin{align}\label{eq:T_bad_bound}
    \labs{\CT} &\le \sum_{\substack{\gamma_1,\gamma_2 \in \Gamma \\ \gamma_1 \neq \gamma_2}}\EV \left[ \abs{s_{\gamma_1}} \abs{s_{\gamma_2}} \int_{t>t_2>t_1>0} \rd t_1 \rd t_2 \norm{\e^{\lind^\dag t_1}} \norm{\lind^\dag_{\gamma_2}(\lind_0^\dag (\Hbf_{\gamma_1}))} \right] \\
    &\leq \frac{t^2}{2} \sum_{\substack{\gamma_1,\gamma_2 \in \Gamma \\ \gamma_1 \neq \gamma_2}} \norm{\lind^\dag_{\gamma_2}(\lind_0^\dag (\Hbf_{\gamma_1}))},
\end{align}
where on the second line we used \cref{lem:L_contract} and the facts that $|s_{\gamma_1}| = |s_{\gamma_2}| = 1$ and $\int_{t>t_2>t_1>0}\rd t_1 \rd t_2 = t^2/2$.
We now give a rough sense of the scaling of the quantity $\norm{\lind^\dag_{\gamma_2}(\lind_0^\dag (\Hbf_{\gamma_1}))}$. The Lindbladian terms $\lind_0^\dag$ and $\lind_{\gamma_2}^\dag$ are defined in \cref{eq:lind_0/gamma/gammagamma'_def} together with \cref{eq:lind_0^a,eq:lind_gamma^a}. Observe that each occurrence of a Hamiltonian term $\Hbf_{\gamma}$ brings a factor of $h_{\gamma}$ into the scaling, and that for a given $\Hbf_\gamma$, only $\CO(k)$ of the values of $a$ will have $[\Abf^a, \Hbf_\gamma] \neq 0$. Furthermore, $\lind^\dag_{\gamma}$ is linear in $y$, and $\lind^{\dag}_0$ has both $y^0$ and $y^2$ components---the $y^2$ components can be shown to be subleading.  Thus, we expect that $\norm{\lind_0^\dag (\Hbf_{\gamma_1})}$ scales roughly as $\CO(h_{\gamma_1} k)$ and then $\norm{\lind^\dag_{\gamma_2}(\lind_0^\dag (\Hbf_{\gamma_1}))}$ scales roughly as $\CO(|y|h_{\gamma_2} h_{\gamma_1}k^2)$.
Assuming $\Hbf$ has $m$ terms and $h_{\gamma} = 1/\sqrt{m}$ for all $\gamma$ (as in \cref{eq:sampled_Pauli} and \cref{eq:sampled_Fermion}), we end up finding 
\begin{align}
    \labs{\CT} \leq t^2 |y| m  k^2 \cdot \CO(1) \,\tag{Loose bound}.
\end{align}
This is problematic due to the dependence on $m$, the number of Hamiltonian terms, which can be much larger than the system size $n$. If the terms composing $T_2$ were this large, then they would dominate the first-order term, which is bounded as $|T_1| \le |y| t k^2 \cdot \CO(1)$. We would not be able to achieve the balancing between $T_1$ and $T_2$ as in \cref{eq:T1_T2_balancing}. This naive bound has not utilized the randomness of the Rademachers, which induces significant cancellations between positive and negative contributions to $\CT$. 

\textbf{The more careful approach.} We pursue a more delicate computation which successfully captures some of these cancellations by expanding the time-ordered integrals---such as that which appears in the expression for $\CT$---via a recursion relation. We will utilize the following lemmas.

\begin{lemma}[Integration by parts for functions of a Rademacher]\label{lem:shift_rule}
For any function of a Rademacher variable $s \sim \Unif(\{\pm 1\})$, it holds that
\begin{align}
    \BE_s[ f(s) s ] = \BE_s\L[ D_sf(s) s \R] \quad \text{where} \quad D_s f(s) :=  f(s) - f(0).
\end{align}
\end{lemma}
\begin{proof}
Note that
    \begin{align}
        \EV_s[D_sf(s)s] = \EV_s[f(s)s - f(0)s] = \EV_s[f(s)s] - \EV_s[f(0)s],
    \end{align}
    and the last term averages to zero.
\end{proof}

\begin{lemma}[Duhamel's identity]\label{lem:duhamel}
For any square matrices $\vA$ and $\vB$, and $t\geq 0$, we have
\begin{align}
    \e^{(\vA+\vB)t} = \e^{\vA t} + \int_0^t \e^{(\vA+\vB) (t-t_1)} \vB \e^{\vA t_1} \rd t_1.
\end{align}
\end{lemma}
\begin{proof}
    Both sides satisfy the differential equation $d\vC(t)/dt = (\vA + \vB)\vC$ with the initial condition $\vC(0) = \vI$.
\end{proof}

We decompose $\lind^\dag = \CA^\dag + \CB^\dag$, where 
\begin{align}
    \CB^\dag &= s_{\gamma_1}\lind^\dag_{\gamma_1} + \sum_{\gamma': \gamma' \neq \gamma_1} s_{\gamma_1}s_{\gamma'}\lind^\dag_{\gamma_1\gamma'} \\
    \CA^\dag &= \lind^\dag - \CB^\dag.
\end{align}
Note that $\CB^\dag$ has linear dependence on $s_{\gamma_1}$, $\CA^\dag$ has no dependence on $s_{\gamma_1}$, and $\CA$ remains a Lindbladian (and thus norm contractive). Applying \cref{lem:shift_rule} to $\CT$ with $s = s_{\gamma_1}$ yields 
\begin{align}
    \CT = \sum_{\substack{\gamma_1,\gamma_2 \in \Gamma \\ \gamma_1 \neq \gamma_2}}\Bigg(&\EV \left[ s_{\gamma_1}s_{\gamma_2} \bTr \int_{t>t_2>t_1>0} \rd t_1 \rd t_2 \;  \e^{\CA^\dag t_1 + \CB^\dag t_1} \lind^\dag_{\gamma_2}(\lind_0^\dag (\Hbf_{\gamma_1}))\right] \nonumber \\
    -&\EV \left[ s_{\gamma_1}s_{\gamma_2} \bTr \int_{t>t_2>t_1>0} \rd t_1 \rd t_2 \; \e^{\CA^\dag t_1} \lind^\dag_{\gamma_2}(\lind_0^\dag (\Hbf_{\gamma_1}))\right]\Bigg).
\end{align}
Next, after applying \cref{lem:duhamel} to the exponential that appears in the first term, we get
\begin{align}
    \CT = \sum_{\substack{\gamma_1,\gamma_2 \in \Gamma \\ \gamma_1 \neq \gamma_2}}\Bigg(&\EV \left[ s_{\gamma_1}s_{\gamma_2} \bTr \int_{t>t_2>t_1>0} \rd t_1 \rd t_2 \; \e^{\CA^\dag t_1} \lind^\dag_{\gamma_2}(\lind_0^\dag (\Hbf_{\gamma_1}))\right] \nonumber \\
    +&\EV \left[ s_{\gamma_1}s_{\gamma_2} \bTr \int_{t>t_2>t_1>\theta>0} \rd t_1 \rd t_2 \rd \theta \; \e^{(\CA^\dag +\CB^\dag)(t_1 -\theta)} \CB^\dag \e^{\CA^\dag \theta} \lind^\dag_{\gamma_2}(\lind_0^\dag (\Hbf_{\gamma_1}))\right] \nonumber\\
    -&\EV \left[ s_{\gamma_1}s_{\gamma_2} \bTr \int_{t>t_2>t_1} \rd t_1 \rd t_2 \; \e^{\CA^\dag t_1} \lind^\dag_{\gamma_2}(\lind_0^\dag (\Hbf_{\gamma_1}))\right]\Bigg). 
\end{align}
The first and the third terms cancel out, yielding only the second term, which is a time-ordered integral of order 3. Furthermore, the appearance of $\CB^\dag$ in the second term brings a linear factor of $s_{\gamma_1}$, which can be brought outside the integral to cancel out the other appearance of $s_{\gamma_1}$. Substituting $\CA^\dag$ and $\CB^\dag$, relabeling $(\theta,t_1,t_2)$ to $(t_1,t_2,t_3)$, and noting that $s_{\gamma_1}^2 = s_{\gamma_2}^2 = 1$, we have
\begin{align}\label{eq:2_Rademachers_first_recursion}
    \CT = \sum_{\substack{\gamma_1,\gamma_2 \in \Gamma \\ \gamma_1 \neq \gamma_2}}\Bigg(&
    \EV \left[ \bTr \int_{t>t_3>t_2>t_1>0} \rd t_1 \rd t_2 \rd t_3 \; \e^{\lind^\dag(t_2 -t_1)} \lind_{\gamma_1\gamma_2}^\dag \e^{(\lind^\dag - \CB^\dag) t_1} \lind^\dag_{\gamma_2}(\lind_0^\dag (\Hbf_{\gamma_1}))\right] \nonumber \\
    + & \EV \left[s_{\gamma_2} \bTr \int_{t>t_3>t_2>t_1>0} \rd t_1 \rd t_2 \rd t_3 \; \e^{\lind^\dag(t_2 -t_1)} \lind_{\gamma_1}^\dag \e^{(\lind^\dag - \CB^\dag) t_1} \lind^\dag_{\gamma_2}(\lind_0^\dag (\Hbf_{\gamma_1}))\right] \nonumber \\
    + & \sum_{\substack{\gamma' \in \Gamma \\ \gamma' \not\in \{\gamma_1,\gamma_2\}}}\EV \left[s_{\gamma'}s_{\gamma_2} \bTr \int_{t>t_3>t_2>t_1>0} \rd t_1 \rd t_2 \rd t_3  \; \e^{\lind^\dag(t_2 -t_1)} \lind_{\gamma_1\gamma'}^\dag \e^{(\lind^\dag - \CB^\dag) t_1} \lind^\dag_{\gamma_2}(\lind_0^\dag (\Hbf_{\gamma_1}))\right]\Bigg).
\end{align}
Let us pause to explain what has happened. We began in \cref{eq:CT} with an expression for $\CT$ that was an order-2 time-ordered integral, with 2 Rademacher factors outside the integral and 1 Lindbladian evolution $\e^{\lind^\dag t_2}$ inside the integral. After applying the recursion relation once, we now have an expression with multiple \textit{order-3} time-ordered integrals, each of which has 0, 1, or 2 Rademacher factors, and 2 applications of Lindbladian evolution within the integral. 

The most interesting term above is the first term, which has 0 Rademachers outside the integral. Both Rademachers were canceled out by the descent of the $s_{\gamma_1}s_{\gamma_2}\lind^\dag_{\gamma_1 \gamma_2}$ term from the exponent in the application of \cref{lem:duhamel}! The pairing up of Rademacher variables in this fashion is key for fully capturing the correlations between the random Rademachers outside the integral and those in the exponents of $\e^{\lind^\dag t_2}$ inside the integral, and the associated cancellations between positive and negative contributions. There is nothing more to do with this term, except to bound it by invoking the triangle inequality and the contractivity of the $ \e^{\lind^\dag t}$ map. That is, we proceed in a similar fashion as we did for \cref{eq:T_bad_bound}:
\begin{align}\label{eq:T_contribution_good_bound}
    &\labs{\sum_{\substack{\gamma_1,\gamma_2 \in \Gamma \\ \gamma_1 \neq \gamma_2}}\EV \left[ \bTr \int_{t>t_3>t_2>t_1>0} \rd t_1 \rd t_2 \rd t_3 \; \e^{\lind^\dag(t_2 -t_1)} \lind_{\gamma_1\gamma_2}^\dag \e^{(\lind^\dag - \CB^\dag) t_1} \lind^\dag_{\gamma_2}(\lind_0^\dag (\Hbf_{\gamma_1}))\right]} \nonumber \\
    &\leq \frac{t^3}{6}\sum_{\substack{\gamma_1,\gamma_2 \in \Gamma \\ \gamma_1 \neq \gamma_2}} \norm{\lind^\dag_{\gamma_1 \gamma_2}} \norm{\lind^\dag_{\gamma_2}(\lind_0^\dag (\Hbf_{\gamma_1}))}.
\end{align}
We now attempt to get a sense of the scaling of the resulting expression. From \cref{eq:lind_gammagamma'^a}, we can see that the operator $\lind^{a \dag}_{\gamma_1\gamma_2}$ vanishes unless both $[\Abf^a, \Hbf_{\gamma_1}]\neq 0$ and $[\Abf^a, \Hbf_{\gamma_2}]\neq 0$. Since $\Abf^a$ is a single-site operator, this can only occur if both $\Hbf_{\gamma_1}$ and $\Hbf_{\gamma_2}$ have support on that site. Hence, thinking of $\Hbf_{\gamma_1}$ as fixed, the only nonvanishing choices of $\gamma_2$ and $a$ occur when both $\Hbf_{\gamma_2}$ and $\Abf^a$ have support that overlaps with the $k$ sites in $S(\gamma_1)$, the support of  $\Hbf_{\gamma_1}$. This logic reduces the sum, as follows.
\begin{align}
\frac{t^3}{6}\sum_{\substack{\gamma_1,\gamma_2 \in \Gamma \\ \gamma_1 \neq \gamma_2}} \norm{\lind^\dag_{\gamma_1 \gamma_2}} \norm{\lind^\dag_{\gamma_2}(\lind_0^\dag (\Hbf_{\gamma_1}))} 
=
\frac{t^3}{6}\sum_{\substack{\gamma_1\in \Gamma}}\sum_{i \in S(\gamma_1)}\sum_{\substack{\gamma_2: i \in S(\gamma_2)\\ \gamma_2 \neq \gamma_1}} \norm{\lind^\dag_{\gamma_1 \gamma_2}} \norm{\lind^\dag_{\gamma_2}(\lind_0^\dag (\Hbf_{\gamma_1}))}
\end{align}
The quantity $\norm{\lind^\dag_{\gamma_1 \gamma_2}}$ scales roughly as $\CO(|y|^2 h_{\gamma_1}h_{\gamma_2}k)$ for choices of $\gamma_1, \gamma_2$ where it is nonvanishing (the $k$ comes from the at most $\CO(k)$ choices of $a$ that could cause $\Abf^a$ to anticommute with both $\Hbf_{\gamma_1}$ and $\Hbf_{\gamma_2}$). Combined with the previous observation that $\norm{\lind^\dag_{\gamma_2}(\lind_0^\dag (\Hbf_{\gamma_1}))}$ scales as $\CO(|y| h_{\gamma_1} h_{\gamma_2}k^2)$, we find that the bound on $\CT$ scales as 
\begin{align}
    \CT &\leq \CO(t^3 |y|^3 k^3)\sum_{\substack{\gamma_1\in \Gamma}} \CO(h_{\gamma_1}^2) \sum_{i \in S(\gamma_1)}\sum_{\substack{\gamma_2: i \in S(\gamma_2)\\ \gamma_2 \neq \gamma_1}} \CO( h_{\gamma_2}^2) \\
    &\leq \CO(t^3 |y|^3 k^3) \HgloPower{2} \HlocPower{2}
\end{align}
where $\Hglo$ and $\Hloc$ are defined in \cref{eq:Hloc_Hglo_def}. Crucially, both $\Hglo$ and $\Hloc$ are independent of the number of terms $m$; for each appearance of a sum over $m$ choices of $\gamma \in \Gamma$, there is an accompanying factor of $h_{\gamma}^2 \sim 1/m$. The residual $m$ dependence was the main issue with the naive attempt above, which did not utilize the randomness of the Hamiltonian.  An additional important feature is that this calculation has incorporated the fact that $\Hbf_{\gamma_2}$ must have overlapping support with $\Hbf_{\gamma_1}$ in order to greatly reduce the number of terms in the sum, resulting in a dependence $\HgloPower{2}\HlocPower{2}$, rather than $\HgloPower{4}$. Ultimately, the values of $t$ and $y$ are chosen such that $t^2 |y|^2 k^2 \HlocPower{2} <1$, allowing the above to contribute as $\CO(t |y| k \HgloPower{2})$, that is, at the same order as the first-order term (\cref{prop:first_order_term}.)

Yet, while we have gained control over the first term in the expression for $\CT$ in \cref{eq:2_Rademachers_first_recursion}, there are two additional terms that were produced by the recursion step, which have 1 or 2 Rademachers outside the integral. To handle these new order-3 terms, we apply the same strategy and express each of them as a sum over several new time-ordered integral terms of order 4. For any of the order-4 terms with Rademachers outside the integral, we again apply the same strategy, yielding order-5 terms. After we have recursed $p-2$ times, what remains are order-$p$ integrals without any Rademachers prefactors, for $p=3,4,5,\ldots$, and a remainder term of high order that includes all the terms that still have Rademachers outside the integral. We show how to sum all the terms without Rademacher prefactors, and we also show that the remainder term approaches $0$ as one performs more recursion levels (see \cref{lem:V_hat_bound_FGPhi}.) 

To account for all the terms, we utilize a diagrammatic method, which allows us to see in what way the Rademachers outside the integral are paired up. Specifically, we equate each of the time-ordered integrals encountered in the expansion with a unique diagram, and then rewrite equations in terms of diagram expansions. For example, the time-ordered integral equality in \cref{eq:2_Rademachers_first_recursion} is rewritten in our diagram language as 
\begin{align}
    \L(\adj{-3pt}{0.7}{\text{\drawDiagramIndexed{{{$\gamma_{1}$/0}/0,{$\gamma_{2}$/0}/0}}}} \R) 
    ={}& 
    \L(\adj{-12pt}{0.7}{\text{\drawDiagramIndexed{{{$\gamma_{1}$/0, $\gamma_{2}$/0}/2,{$\gamma_{2}$/0}/0}}}} \R)
    +\L(\adj{-8pt}{0.7}{\text{\drawDiagramIndexed{{{$\gamma_{1}$/0, $\gamma_{1}$/0}/1,{$\gamma_{2}$/0}/0}}}} \R)
    + \sum_{\gamma' \in \Gamma \setminus\{\gamma_1,\gamma_2\}}  \L(\adj{-12pt}{0.7}{\text{\drawDiagramIndexed{{{$\gamma_{1}$/0, $\gamma'$/0}/0,{$\gamma_{2}$/0}/0}}}} \R)
\end{align}
See \cref{sec:proof_superoperator_norm} for details. Nodes with arrows represent Rademachers outside the integral. The dashed line represents the pairing up of two Rademachers in one step, as occurred in the first term of \cref{eq:2_Rademachers_first_recursion}, reducing the number of Rademachers by 2. Filled nodes represent the pairing up of one of these Rademachers with itself, reducing the number of Rademachers by 1 compared to the left-hand side, as occurred in the second term of \cref{eq:2_Rademachers_first_recursion}. 

By using this diagram language, we can exhaustively count all the terms and thus compute their contribution to the upper bound on $|T_2|$. Fortunately, we need only count diagrams in which all the arrows have been removed, either by pairing up with other arrows, or pairing up with themselves (via a filled node), since we can show that the contribution from the remainder diagrams vanishes with the number of recursion steps. This restricts the combinatorics and makes the calculation tractable. As we can have as many as $|U|=5$ Rademachers, we can end up with diagrams as complex as, for example,
\begin{align}
    \L(\adj{-47pt}{0.7}{\text{\drawDiagramIndexed{{{$\gamma_{1}$/0, $\alpha_{1}$/0, $\alpha_{2}$/1, $\alpha_3$/0, $\alpha_4$/3, $\gamma_4$/3}/4,{$\gamma_{2}$/0, $\alpha_5$/5, $\gamma_5$/1}/5,{$\gamma_{3}$/0, $\alpha_6$/0, $\alpha_7$/4, $\alpha_8$/2, $\alpha_8$/9}/3,{$\gamma_4$/0}/1,{$\gamma_5$/0}/1}}}} \R)
\end{align}
When applying the recursion relation, we choose to extend the leftmost path with an arrow. So the above diagram represents a term where the first Rademacher paired up with the fourth Rademacher on the fifth recursion step, then the second Rademacher paired up with the fifth Rademacher on the seventh recursion step, and finally the third Rademacher paired up with itself on the eleventh recursion step, yielding a diagram without any arrows. 

Ultimately, the contribution of a diagram to the sum has a factor that decays exponentially with the length (i.e.,~number of solid edges) in the diagram. Thus, the scaling of the full sum is captured by the lower-order terms, which have favorable bounds, as was found for the example of the order-3 term in \cref{eq:T_contribution_good_bound}.  

\section{Conclusion}
In this work, we consider the low-energy state preparation problem in the setting of quantum optimization problems and studied  random, all-to-all connected models of strongly interacting many-body quantum systems---quantum analogues of classical spin glasses. Our main finding is an average-case quantum algorithmic guarantee stating that a simplified version of quantum Gibbs samplers can prepare quantum states achieving an $\Omega(\frac{1}{k})$-fraction of the optimum energy---this is proved in a general approach that handles both the spin and fermionic cases, and both the sparse and dense cases all at once. This guarantee is  an exponential improvement over prior classical and quantum algorithms (an $\frac{1}{e^{\Omega(k)}}$-fraction) for the same family of models. Combined with existing circuit lower bounds, our refined guarantee suggests that the low-energy states of sparsified local fermionic and quasilocal spin models are classically nontrivial and quantumly easy.

In fact, we conjecture that our $\Omega(\frac{\sqrt{n}}{k})$-scaling for the achieved energy actually coincides with the true optimum, up to an absolute ($k$-independent) constant. Such scaling is consistent with the physics heuristics $\norm{\vH} \approx \frac{\sqrt{2n}}{k}$ for the even-$k$ SYK model (see, e.g.,~\cite{hastings2022optimizing}). However, a rigorous understanding of the spectrum, physical properties, and complexity at low energies remains very limited, and further progress here may illuminate the aspect of quantum advantage in these systems. 

Meanwhile, our algorithmic guarantees set a quantitative challenge for quantum and classical algorithm designers, and these random models are natural testbeds. Currently, this simplified version of the quantum Gibbs sampler is the only systematic algorithm that achieves a $\Omega(\frac{1}{k})$-approximation. We suspect that the generic classical approach cannot exceed a $\e^{-\Omega(k)}$-fraction, and we are hopeful that $\Omega(\frac{1}{k})$-fractions should be accessible by other quantum algorithmic principles.

This work has mostly been restricted to the context of optimization problems, but thermodynamic simulation often involves the quantum Gibbs state. Of course, quantum MCMC algorithms are most natural for this task, but so far, there have been very limited tools to study mixing times and spectral gaps of quantum Lindbladians. Interestingly, it was recently proved~\cite{anshuetz2024strongly} that the annealed approximation for the free energy in these spin and fermionic models continue to hold at very low temperatures ($\beta = e^{\Omega(k)}$ for random Paulis, $\beta \sim \Omega_k(n^{k/4})$ for fermions); this hints that the model may be glassy and remain rapidly-mixing for a wide range of temperatures.

\section*{Acknowledgements}
We thank Thiago Bergamaschi, Aram Harrow, Jonas Helsen, Yaroslav Herasymenko, Lin Lin, Sam McArdle, Alexander Schmidhuber, Nikhil Srivastava, Leo Zhou, and Alexander Zlokapa for insightful discussions. We thank Eric Anshuetz, Bobak Kiani, and Robbie King for discussions on related work~\cite{anshuetz2024strongly}. CFC is supported by a Simons-CIQC postdoctoral fellowship through NSF QLCI Grant No. 2016245. AD is grateful to the AWS Center for Quantum Computing for its support. 

\bibliographystyle{halphalab}
\bibliography{main}

\appendix

\section{A diagrammatic calculus for upper bounding the norm of the time-ordered integrals}\label{sec:diagrammatic_calculus}
In this section, we work toward establishing \cref{prop:superoperator_norm}. To do so, we require some setup and definitions over several subsections. This framework is used to finally prove the proposition in \Cref{sec:proof_superoperator_norm}. 

\subsection{Compact notation for time-ordered integrals (TOIs)}

To prove the proposition, we will repeatedly expand time-ordered integrals (TOIs) into TOIs of higher order and more complicated integrands. To organize the expansion, we must address individual terms in the global Lindbladian $\CL$, which are labeled as follows.

\begin{definition}[Labels]\label{def:label}
    We let  $\CM = \Gamma \cup (\Gamma \times \Gamma) \cup \{0\}$ denote the set of {\bf labels} for Lindbladian terms, as defined in \cref{eq:lind_0^a,eq:lind_gamma^a,eq:lind_gammagamma'^a}. For example, if $l = (\gamma_1,\gamma_2) \in \Gamma \times \Gamma \subset \CM$ is a label, then $\lind_{l}^\dag := \lind_{\gamma_1 \gamma_2}^\dag$.
\end{definition}

\begin{definition}
    For a subset of indices $S \subseteq \Gamma$, we write
    \begin{align}\label{eq:lind_subset_notation}
        \lind_S^\dag = \lind_0^\dag + \sum_{\gamma \in S} s_\gamma \lind_\gamma^\dag + \frac{1}{2} \sum_{\substack{ \gamma,\gamma' \in S \\ \gamma\neq\gamma'}} s_\gamma s_{\gamma'} \lind_{\gamma\gamma'}^\dag,
    \end{align}
    where $\lind_0^\dag$, $\lind_\gamma^\dag$, and $\lind_{\gamma\gamma'}^\dag$ include the sum over $a \in A$, and $\lind_{\gamma\gamma'}^\dag$ is symmetrized over its suscripts, as defined in \cref{eq:lind_0/gamma/gammagamma'_def}. 
    That is, $\lind_S^\dag$ includes Lindbladian terms only involving indices from $S$ (compare with \cref{eq:lind^a_as_Rademacher_polynomial}.)
\end{definition}

The following lemma guarantees that this restriction preserves the Lindbladian structure.
\begin{lemma}\label{lem:Lind_subset_is_Lindbladian}
    For all $S \subseteq \Gamma$, $\CL_{S} $ is also a Lindbladian. 
\end{lemma}
\begin{proof}
Observe that $\CL_{S}= \BE_{\Gamma\setminus S}\CL$ where the expectation is taken over all $s_{\gamma}:\gamma \in \Gamma \setminus S$. Since a convex combination of Lindbladians is still a Lindbladian, the lemma follows.
\end{proof}

Now we define the general form of a time-ordered integral, where the intermediate Lindbladians include particular choices of indices.
\begin{definition}[Compact notation for time-ordered integrals (TOIs)]\label{def:TOIs}
Given a non-negative integer $\ell>0$, sets $S_0,\ldots,S_\ell \subseteq \Gamma$, and labels $l_1,\ldots,l_\ell \in \CM$, we define the object $\bb*{S_\ell,l_{\ell},\ldots,S_1,l_1,S_0}$ to be a superoperator (mapping operators to real numbers), defined by its action on an arbitrary operator $\Obf$:
\begin{align}
    &\bb*{S_\ell,l_{\ell},S_{\ell-1},l_{\ell-1}, \ldots,S_1,l_1,S_0}(\Obf) \nonumber\\
    :={}& \bTr\L( \int_{t>t_{\ell+2}>\ldots>t_1>0} \rd t_{\ell+2} \ldots \rd t_1 \; 
    \e^{(t_{\ell+1}-t_\ell) \lind_{\Gamma \setminus S_\ell}^\dag} \lind_{l_{\ell}}^\dag 
    \ldots
    \e^{(t_{3} - t_{2})\lind_{\Gamma\setminus S_{2}}^\dag}\lind_{l_2}^{\dag}
    \e^{(t_{2} - t_{1})\lind_{\Gamma\setminus S_{1}}^\dag}\lind_{l_1}^\dag \e^{t_1\lind_{\Gamma \setminus S_0}^\dag} (\Obf) \R). \label{eq:TOI_notation}
\end{align}
When necessary, we identify $t_{\ell+3} = t$ and $t_0 = 0$. We say that $\ell$ is the \emph{length} of the TOI. 
\end{definition}

The TOI objects depend on the Rademachers through the appearance of the $s_\gamma\lind_\gamma^\dag$ and $s_\gamma s_{\gamma'}\lind_{\gamma \gamma'}^\dag$ terms ``upstairs'' in the exponent. We also need to consider expressions with TOIs weighted by Rademachers ``downstairs'' and then average over the Rademachers. 

\begin{definition}[Rademacher-averaged TOI]\label{def:Rademacher-averaged_TOI}
Instantiante the definitions in \cref{def:TOIs}. Let $h\ge 1$ be an integer and define the Rademacher random variables $s_{\gamma_1},\ldots,s_{\gamma_h} \sim \Unif(\{\pm 1\})$ which are indexed by elements of $\Gamma$.
A Rademacher-averaged TOI is a superoperator of the form
\begin{equation}
\EV s_{\gamma_1}\ldots s_{\gamma_h}\bb*{S_\ell,l_\ell,\ldots,S_1,l_1,S_0}.
\end{equation}
These are labeled by elements from the set
\begin{align}
    \CW := \{0\} \cup \bigcup_{\ell=0}^\infty\bigcup_{h=0}^{\infty} \CW_{h,\ell}
\end{align}
with
\begin{align}
    \CW_{h,\ell} :=\L\{(  (\gamma_1, \cdots, \gamma_{h}), (S_\ell,l_\ell,\ldots S_1,l_1,S_0)) ) : \gamma_i \in \Gamma, S_k \subseteq \Gamma, l_k \in \CM \R\}.
\end{align} 
When $w \in \CW_{h,\ell}$, we say that $h$ is the number of ``downstairs'' Rademachers, and $\ell$ is the length of the Rademacher-averaged TOI. We also define the value $V(w,\Obf)$ of a Rademacher-averaged TOI $w \in \CW$ acting on operator $\Obf$ by $V(0,\Obf) = 0$ for all $\Obf$ and 
\begin{equation}\label{eq:V_wO_def}
    V\L( ((\gamma_1,\ldots,\gamma_h),(S_\ell,l_\ell,\ldots S_1,l_1,S_0)) , \Obf\R) = \EV s_{\gamma_1}\ldots s_{\gamma_h}\bb*{S_\ell,l_\ell,\ldots,S_1,l_1,S_0}(\Obf) \in \mathbb{R}.
\end{equation} 
\end{definition}
Note that $V$ is a linear function in its second argument. We have included the element $0$ in the set $\CW$ to represent the superoperator that maps all inputs to 0. This is important later to map the invalid choices of ordering to 0, and have this not leave the set.

\begin{lemma}[Bound on Rademacher-averaged TOI]\label{lem:bouding_rad_avg_toi}
    Let $h\ge 0$ be an integer and $\EV s_{\gamma_1}\ldots s_{\gamma_h}\bb*{S_\ell,l_\ell,\ldots,S_1,l_1,S_0}$ a Rademacher-averaged TOI as in \cref{def:Rademacher-averaged_TOI} applied to an operator $\Obf$. Then we can bound its absolute value as
    \begin{align}
        \labs{\BE s_{\gamma_1}\ldots s_{\gamma_h}\bb*{S_\ell, l_\ell, S_{\ell-1},\ldots, S_1,l_1,S_0}(\Obf)} \leq \frac{t^{\ell+2}}{(\ell+2)!}\norm{\lind^\dag_{l_\ell}}\norm{\lind^\dag_{l_{\ell-1}}}\ldots \norm{\lind^\dag_{l_2}}\norm{\lind^\dag_{l_1}}\norm{\Obf}\,.
    \end{align}
\end{lemma}
\begin{proof}
    First, when bounding a TOI, we can use the fact that the normalized trace satisfies $\labs{\bTr(\Obf)} \leq \norm{\Obf}$, as well as triangle inequality and submultiplicativity of the operator norm to get
    \begin{align}
        &\labs{\bTr\L( \int_{t>t_{\ell+2}>\ldots>t_1>0} \rd t_{\ell+2} \ldots \rd t_1 \; 
    \e^{(t_{\ell+1}-t_\ell) \lind_{\Gamma \setminus S_\ell}^\dag} \lind_{l_{\ell}}^\dag 
    \ldots
    \e^{(t_{3} - t_{2})\lind_{\Gamma\setminus S_{2}}^\dag}\lind_{l_2}^{\dag}
    \e^{(t_{2} - t_{1})\lind_{\Gamma\setminus S_{1}}^\dag}\lind_{l_1}^\dag \e^{t_1\lind_{\Gamma \setminus S_0}^\dag} (\Obf) \R)}\\
    \leq{}& \norm{\int_{t>t_{\ell+2}>\ldots>t_1>0} \rd t_{\ell+2} \ldots \rd t_1 \; 
     \e^{(t_{\ell+1}-t_\ell) \lind_{\Gamma \setminus S_\ell}^\dag} \lind_{l_{\ell}}^\dag 
    \ldots
    \e^{(t_{3} - t_{2})\lind_{\Gamma\setminus S_{2}}^\dag}\lind_{l_2}^{\dag}
    \e^{(t_{2} - t_{1})\lind_{\Gamma\setminus S_{1}}^\dag}\lind_{l_1}^\dag \e^{t_1\lind_{\Gamma \setminus S_0}^\dag} (\Obf)} \\
    \leq{}& \int_{t>t_{\ell+2}>\ldots>t_1>0} \rd t_{\ell+2} \ldots \rd t_1 \; 
     \norm{\e^{(t_{\ell+1}-t_\ell) \lind_{\Gamma \setminus S_\ell}^\dag} \lind_{l_{\ell}}^\dag 
    \ldots
    \e^{(t_{3} - t_{2})\lind_{\Gamma\setminus S_{2}}^\dag}\lind_{l_2}^{\dag}
    \e^{(t_{2} - t_{1})\lind_{\Gamma\setminus S_{1}}^\dag}\lind_{l_1}^\dag \e^{t_1\lind_{\Gamma \setminus S_0}^\dag} (\Obf)} \\
    \leq{}& \int_{t>t_{\ell+2}>\ldots>t_1>0} \rd t_{\ell+2} \ldots \rd t_1 \; 
    \norm{\lind^\dag_{l_\ell}}\ldots\norm{\lind^\dag_{l_2}}\norm{\lind^\dag_{l_1}}\norm{\Obf} \\
    ={}& \frac{t^{\ell+2}\norm{\lind^\dag_{l_\ell}}\norm{\lind^\dag_{l_{\ell-1}}}\ldots \norm{\lind^\dag_{l_2}}\norm{\lind^\dag_{l_1}}\norm{\Obf}}{(\ell+2)!},
    \end{align}
    where we used the formula for the volume of the rescaled simplex:
    \begin{align}
    \int_{t>t_{\ell+2}>\ldots>t_1>0} \rd t_{\ell+2} \ldots \rd t_1 = \frac{t^{\ell+2}}{(\ell+2)!}.
    \end{align}
    We have also used the fact that $\norm{\e^{(t_{j+1}-t_j} \lind^{\dag}_{\Gamma \setminus S_j}} \leq 1$, owing to the fact that $\lind^{\dag}_{\Gamma \setminus S_j}$ is a Lindbladian (\cref{lem:Lind_subset_is_Lindbladian}), and that Lindbladian evolution of operators is operator-norm contractive (\cref{lem:L_contract}).
    
    The Rademacher-averaged TOI is an average over the $2^h$ settings of the Rademachers, each weighted by $+1$ or $-1$, but each setting gives rise to a TOI that can be upper bounded as above. Thus we have
    \begin{equation}
        \labs{\BE s_{\gamma_1}\ldots s_{\gamma_h}\bb*{S_\ell, l_\ell, S_{\ell-1},\ldots, S_1,l_1,S_0}(\Obf)} \leq \frac{t^{\ell+2}}{(\ell+2)!}\norm{\lind^\dag_{l_\ell}}\norm{\lind^\dag_{l_{\ell-1}}}\ldots \norm{\lind^\dag_{l_2}}\norm{\lind^\dag_{l_1}}\norm{\Obf}\,.
    \end{equation}
\end{proof}

\subsection{Expansion rules for TOIs}
The main technique that enables the calculation is a recursion relation that expands one TOI of length $\ell$ as a sum over various TOIs of length $\ell+1$. We establish the following lemmas to capture this relation.
\begin{lemma}[Integration by parts for any TOI]\label{lem:shift_rule_compact_notation}
Let $(\gamma_1,\ldots, \gamma_h)$ be any tuple of distinct elements of  $\Gamma$. Then
    \begin{align}
        &\EV s_{\gamma_1}\ldots s_{\gamma_h} \bb*{S_\ell,l_\ell,\ldots,S_1,l_1,S_0} \\
        ={}& \EV s_{\gamma_1}\ldots s_{\gamma_h} \bb*{S_\ell,l_\ell,\ldots,S_1,l_1,S_0} - \EV s_{\gamma_1}\ldots s_{\gamma_h}\bb*{S_\ell\cup\{\gamma_1\},l_\ell,\ldots,S_1\cup\{\gamma_1\},l_1,S_0\cup\{\gamma_1\}}.
    \end{align}
\end{lemma}
\begin{proof}
    Referring to the notational definitions in \cref{eq:TOI_notation,eq:lind_subset_notation}, we can see that the expression
    $\bb*{S_\ell\cup\{\gamma_1\},l_\ell,\ldots,S_1\cup\{\gamma_1\},l_1,S_0\cup\{\gamma_1\}}$ has no dependence on the Rademacher $s_{\gamma_1}$. Thus, we have that $\BE s_{\gamma_1}\ldots s_{\gamma_h}\bb*{S_\ell\cup\{\gamma_1\},l_\ell,\ldots,S_1\cup\{\gamma_1\},l_1,S_0\cup\{\gamma_1\}} = 0$, since $s_{\gamma_1}$ is uniformly chosen from $\{+1,-1\}$. This proves the lemma.
\end{proof}
Note that this is analogous to \cref{lem:shift_rule}.

\begin{lemma}[Duhemal's identity applied to one exponential in a TOI]\label{lem:duhemal_compact_notation}
    For any $\gamma \in \Gamma$ and $k \in \{0,1,2,\ldots,\ell\}$ for which $\gamma \not\in S_k$, it holds that
    \begin{align}
        \bb*{S_\ell,l_\ell,\ldots,S_1,l_1,S_0} &= \bb*{S_\ell,l_\ell,\ldots,S_{k+1},l_{k+1},S_{k}\cup\{\gamma\},l_k,S_{k-1},l_{k-1},\ldots,S_1,l_1,S_{0}} \nonumber \\
        &+ s_\gamma \bb*{S_\ell,l_\ell,\ldots,S_{k},\gamma,S_{k}\cup\{\gamma\},l_k,S_{k-1},l_{k-1},\ldots,S_1,l_1,S_{0}} \nonumber \\
        &+ s_\gamma \sum_{\substack{\gamma' \in \Gamma \\ \gamma' \not\in S_k \cup \{\gamma\}}} s_{\gamma'}\bb*{S_\ell,l_\ell,\ldots,S_{k},\gamma\gamma',S_{k}\cup\{\gamma\},l_k,S_{k-1},l_{k-1},\ldots,S_1,l_1,S_{0}}.
    \end{align}
\end{lemma}
\begin{proof}
To further condense notation, denote the set of coordinates $\tau_\ell = (t_1,\ldots,t_\ell)$ for which $0 = t_0 < t_1 < \ldots < t_\ell < t$ by $\Delta_{\ell}$, and let the differential $\rd \tau_{\ell} = \rd t_{\ell}\ldots \rd t_{1}$.   We now apply \cref{lem:duhamel} to the exponential operator $\e^{(t_{k+1}-t_k)\lind_{\Gamma\setminus S_k}^\dag}$ from \cref{eq:TOI_notation}.
Specifically, referring to \cref{eq:lind_subset_notation}, we may decompose 
\begin{align}
    \lind_{\Gamma\setminus S_k}^\dag &= \lind_{\Gamma\setminus (S_k \cup \{\gamma\})}^\dag + s_\gamma \lind_\gamma^\dag + \frac{1}{2} \sum_{\gamma' \in \Gamma \setminus (S_k \cup \{\gamma\})}\L(s_\gamma s_{\gamma'} \lind_{\gamma \gamma'}^\dag + s_{\gamma'} s_{\gamma} \lind_{\gamma' \gamma}^\dag \R) \\
    &= \underbrace{\lind_{\Gamma\setminus (S_k \cup \{\gamma\})}^\dag}_{\vA} + \underbrace{s_\gamma \lind_\gamma^\dag +  \sum_{\gamma' \in \Gamma \setminus (S_k \cup \{\gamma\})}s_\gamma s_{\gamma'} \lind_{\gamma \gamma'}^\dag}_{\vB}
\end{align}
where in the last line we have used the symmetry of $\lind_{\gamma\gamma'}^\dag$ with respect to its subscripts. \cref{lem:duhamel} then implies that 
\begin{align}
    & \e^{(t_{k+1}-t_k)\lind_{\Gamma\setminus S_k}^\dag} = \e^{(t_{k+1}-t_k)(\vA + \vB)} = \e^{(t_{k+1}-t_k)\vA} + \int_0^{t_{k+1}-t_k} \rd \theta \; \e^{(t_{k+1}-t_k-\theta)(\vA + \vB)} \vB \e^{\theta \vA} \\
    ={}&\e^{(t_{k+1}-t_k)\lind_{\Gamma\setminus (S_k \cup \{\gamma\})}^\dag} + \int_0^{t_{k+1}-t_k} \rd \theta \; \e^{(t_{k+1}-t_k-\theta)\lind^\dag_{\Gamma \setminus S_k}} \L(s_\gamma \lind_\gamma^\dag +  \sum_{\gamma' \in \Gamma \setminus (S_k \cup \{\gamma\})}s_\gamma s_{\gamma'} \lind_{\gamma \gamma'}^\dag\R) \e^{\theta \lind_{\Gamma\setminus (S_k \cup \{\gamma\})}^\dag}
\end{align}
Plugging this into \cref{eq:TOI_notation} yields
\begin{align}
    &\bTr \int_{\Delta_{\ell+2}} \rd \tau_{\ell+2} \, \e^{(t_{\ell+1}-t_{\ell}) \lind_{\Gamma \setminus S_{\ell}}^\dag} \lind_{l_\ell}^\dag \ldots \Bigg[ \e^{(t_{k+1} - t_k)\lind_{\Gamma\setminus (S_k\cup\{\gamma\})}^\dag} \nonumber \\
    &\qquad  + \int_{0}^{t_{k+1}-t_k} \rd \theta \, \e^{(t_{k+1}-t_k-\theta)\lind^\dag_{\Gamma\setminus S_k}} s_\gamma\left(\lind_\gamma^\dag + \sum_{\gamma' \in \Gamma \setminus (S_k \cup\{\gamma\})} s_{\gamma'} \lind_{\gamma\gamma'}^\dag \right) \e^{\theta\lind_{\Gamma\setminus (S_k\cup\{\gamma\})}^\dag} \Bigg] \ldots \lind_{l_1}^\dag \e^{(t_1-t_0)\lind_{\Gamma \setminus S_0}^\dag} (\Obf) \\
    ={}&\bTr \int_{\Delta_{\ell+2}} \rd \tau_{\ell+2} \, \e^{(t_{\ell+1}-t_{\ell}) \lind^\dag_{\Gamma \setminus S_{\ell}}} \lind^\dag_{l_\ell} \ldots \e^{(t_{k+1}-t_k)\lind^\dag_{\Gamma\setminus (S_k \cup \{\gamma\}) }} \ldots \lind^\dag_{l_1} \e^{(t_1-t_0)\lind^\dag_{\Gamma \setminus S_0}} (\Obf) \\
    & +s_\gamma \bTr \int_{\Delta_{\ell+3}} \rd \tau_{\ell+3} \, \e^{(t_{\ell+2}-t_{\ell+1}) \lind^\dag_{\Gamma \setminus S_{\ell}}} \lind^\dag_{l_\ell} \ldots \e^{(t_{k+2} - t_{k+1})\lind^\dag_{\Gamma\setminus S_k}} \lind_\gamma^\dag \e^{(t_{k+1}-t_k)\lind^\dag_{\Gamma\setminus (S_k\cup\{\gamma\})}} \ldots \lind^\dag_{l_1} \e^{(t_1-t_0)\lind^\dag_{\Gamma \setminus S_0}} (\Obf) \nonumber \\
    &+s_\gamma \sum_{\gamma' \in \Gamma \setminus (S_k \cup\{\gamma\})} s_{\gamma'} \bTr \int_{\Delta_{\ell+3}} \rd \tau_{\ell+3} \nonumber \Bigg[\\
    &\qquad \qquad \qquad \e^{(t_{\ell+2}-t_{\ell+1}) \lind^\dag_{\Gamma \setminus S_{\ell}}} \lind^\dag_{l_\ell} \ldots \e^{(t_{k+2} - t_{k+1})\lind^\dag_{\Gamma\setminus (S_k)}} \lind_{\gamma\gamma'}^\dag \e^{(t_{k+1}-t_k)\lind^\dag_{\Gamma\setminus (S_k\cup\{\gamma\})}} \ldots \lind^\dag_{l_1} \e^{(t_1-t_0)\lind^\dag_{\Gamma \setminus S_0}} (\Obf)\Bigg]
\end{align}
where we performed the following change of variables: $t_{j+1} \leftarrow t_j$ for $k+1\le j\le \ell+2$ and then reassign $t_{k+1} \leftarrow \theta + t_k$. Changing back to the notation in \cref{eq:TOI_notation} proves the lemma. 
\end{proof}

\begin{lemma}[Master recursion rule for Rademacher-averaged TOIs]\label{lem:master_recursion_rule}
    Let $(\gamma_1,\ldots, \gamma_h)$ be any tuple of distinct elements of  $\Gamma$, with $h \geq 1$. Then    \begin{align}\label{eq:one_variable_recursion_unsimplified}
        &\qquad \EV s_{\gamma_1}s_{\gamma_2}\ldots s_{\gamma_h} \bb*{S_\ell,l_\ell,\ldots S_1,l_1,S_0} \nonumber \\
        ={}& \sum_{\substack{k=0 \\ \gamma_1 \not\in S_k}}^\ell \EV s_{\gamma_2}\ldots s_{\gamma_h} \bb*{S_\ell \cup \{\gamma_1\},l_{\ell},\ldots,S_{k+1}\cup\{\gamma_1\},l_{k+1},S_k,\gamma_1,S_{k}\cup\{\gamma_1\},l_k,S_{k-1},l_{k-1},\ldots, S_1,l_1,S_0 } \nonumber \\
        +& \sum_{\substack{k=0 \\ \gamma_1 \not\in S_k}}^\ell \sum_{\substack{\gamma' \in \Gamma \\ \gamma'\not\in S_k \cup \{\gamma_1\}}} \EV s_{\gamma'}s_{\gamma_2}\ldots s_{\gamma_h} \bb*{S_\ell \cup \{\gamma_1\},l_{\ell},\ldots,S_{k+1}\cup\{\gamma_1\},l_{k+1},S_k,\gamma\gamma',S_{k}\cup\{\gamma_1\},l_k,S_{k-1},l_{k-1},\ldots, S_1,l_1,S_0 } .
    \end{align}
    Note that for terms where $\gamma' = \gamma_i$ for some $i \in\{2,\ldots,h\}$, then $s_{\gamma'}s_{\gamma_i} = 1$ and one pair of Rademachers will cancel in the final term. 
\end{lemma}

\begin{proof}
    For simplicity, first suppose that $\gamma_1 \not\in S_1 \cup S_{2} \cup \cdots \cup {S_\ell}$. We start by applying \cref{lem:shift_rule_compact_notation}, followed by applying \cref{lem:duhemal_compact_notation} (Duhamel's identity) with $\gamma = \gamma_1$ and $k=\ell$, arriving at
    \begin{align}
        &\EV s_{\gamma_1}s_{\gamma_2}\ldots s_{\gamma_h} \bb*{S_{\ell},l_\ell,\ldots S_1,l_1,S_0} \\
        ={}& \EV s_{\gamma_1}s_{\gamma_2}\ldots s_{\gamma_h} \Bigg[ \bb*{S_\ell\cup\{\gamma_1\},l_\ell,S_{\ell-1},l_{\ell-1}\ldots,l_1,S_{0}} 
        +s_{\gamma_1} \bb*{S_{\ell},\gamma_1, S_{\ell} \cup \{\gamma_1\},l_{\ell},S_{\ell-1},l_{\ell-1}\ldots,l_1,S_{0}} \nonumber \\
        & + s_{\gamma_1} \sum_{ S_\ell \not\ni \gamma'\neq\gamma_1} s_{\gamma'} \bb*{S_\ell,\gamma \gamma', S_\ell \cup \{\gamma\},l_\ell,S_{\ell-1},l_{\ell-1}\ldots,l_1,S_{0}} \nonumber - \bb*{S_\ell \cup \{\gamma_1\},l_\ell,\ldots,l_1,S_{0} \cup \{\gamma_1\}} \Bigg].\label{eq:One_Recursion_Step_Bracket_Notation}
    \end{align}
    We now apply \cref{lem:duhemal_compact_notation} to the second exponential of the first term i.e., choosing $\gamma = \gamma_1$ and $k=\ell-1$:
        \begin{align}
        \EV s_{\gamma_1}s_{\gamma_2}\ldots s_{\gamma_h} & \bb*{S_{\ell},l_{\ell},\ldots S_1,l_1,S_0} = \EV s_{\gamma_1}s_{\gamma_2}\ldots s_{\gamma_h} \Bigg[ \bb*{S_{\ell}\cup\{\gamma_1\},l_{\ell},S_{\ell-1}\cup\{\gamma_1\},l_{\ell-1},S_{\ell-2},l_{\ell-1},\ldots,l_1,S_{0}} \nonumber \\
        &+ s_{\gamma_1} \bb*{S_{\ell} \cup \{\gamma_1\},\gamma_1,S_{\ell-1},l_{\ell},S_{\ell-1}\cup\{\gamma_1\},l_{\ell-1},S_{\ell-2}, \ldots,l_1,S_{0}} \nonumber \\
        &\quad\quad + s_{\gamma_1} \sum_{ S_{\ell} \not\ni \gamma'\neq\gamma_1} s_{\gamma'} \bb*{S_{\ell} \cup \{\gamma_1\},\gamma_1\gamma',S_{\ell-1},l_{\ell},S_{\ell-1}\cup\{\gamma_1\},l_{\ell-1},S_{\ell-2}, \ldots,l_1,S_{0}} \nonumber \\
        &+s_{\gamma_1} \bb*{S_{\ell},\gamma_1, S_{\ell} \cup \{\gamma_1\},l_{\ell},S_{\ell-1},l_{\ell-1}\ldots,l_1,S_{0} } \nonumber \\
        &\quad\quad + s_{\gamma_1} \sum_{ S_{\ell} \not\ni \gamma'\neq\gamma_1} s_{\gamma'} \bb*{S_{\ell},\gamma_1 \gamma', S_{\ell} \cup \{\gamma_1\},l_{\ell},S_{\ell-1},l_{\ell-1}\ldots,l_1,S_{0} } \nonumber \\
        &- \bb*{S_{\ell} \cup \{\gamma_1\},l_{\ell},\ldots,l_1,S_{0} \cup \{\gamma_1\} } \Bigg].
    \end{align}
    We proceed in this fashion, repeatedly applying \cref{lem:duhemal_compact_notation} to the first term with $\gamma = \gamma_1$ and with $k$ descending to $0$. At the end, the first term will cancel with the last term, and the leftover terms created by choosing $k = \ell, \ell-1,\ldots,0$ are precisely the right-hand side of \cref{eq:one_variable_recursion_unsimplified}.
    
    If $\gamma_1 \in S_i$ for one or more values of $i$, we simply skip that value of $k$ in the sequence; the first and last terms will still cancel, and the theorem still goes through, since those values are omitted in \cref{eq:One_Recursion_Step_Bracket_Notation}. 
\end{proof}

The recursion rule in \cref{lem:master_recursion_rule} relates one Rademacher-averaged TOI as the sum of over many Rademacher-averaged TOIs of longer length. We will call the longer Rademacher-averaged TOIs that compose this sum the ``children'' of the shorter one, indexed by which value of $k$ and $\gamma'$ the correspond to on the right-hand side of the recursion rule. The value of a Rademacher-averaged TOI is equal to the sum over its children. This is captured formally in the following definition and lemma. 

\begin{definition}[Children of Rademacher-averaged TOI]\label{def:child_Rademacher_averaged_TOI}
Let $w= ((\gamma_1,\ldots,\gamma_h), (S_\ell, l_\ell, \ldots,l_1,S_0)) \in \CW$ represent a Rademacher-averaged TOI object, and let $\gamma' \in \Gamma$ and $i \in \{0,1,\ldots,\ell\}$. If $h=0$, define $\CE_{\gamma',i}(w) = 0$. Otherwise, define
\begin{equation}\label{eq:child_of_w}
    \CE_{\gamma',i}(w) = \begin{cases}
        0 & \text{if } \gamma_1 \in S_{i} \text{ or } \gamma' \in S_{i}\\
        (\bar{\gamma}, (S_{\ell} \cup \{\gamma_1\},l_{\ell},\ldots, S_{i+1} \cup \{\gamma_1\}, l_{i + 1}, S_{i}, \bar{\mu}, S_{i}\cup \{\gamma_1\}, l_{i}, S_{i-1}, \ldots,S_1,l_1,S_0))& \text{otherwise}
    \end{cases}
\end{equation}
where
\begin{align}\label{eq:bar_gamma}
     \bar{\gamma} = \begin{cases}
        (\gamma_2,\ldots, \gamma_{h})& \text{if } \gamma' = \gamma_1\\
        (\gamma_2,\ldots,\gamma_{k-1},\gamma_{k+1},\ldots,\gamma_{h})& \text{if } \gamma' = \gamma_k \text{ for } k > 1\\
        (\alpha,\gamma_2,\gamma_3,\ldots,\gamma_{h}) & \text{otherwise}
    \end{cases} \qquad 
\end{align}
and
\begin{align}\label{eq:bar_mu}
    \bar{\mu} = \begin{cases}
        \gamma_1 & \text{if } \gamma' = \gamma_1\\
        \gamma_1 \gamma' & \text{otherwise } 
    \end{cases} \qquad \in \CM.
\end{align}
We say that $\CE_{\gamma',i}(w)$ is the $(\gamma',i)$th \textit{child} of $w$.

We further define the set of \emph{children} of $w$: $\CE(w) = \{\CE_{\gamma',i}(w): i \in \{0,1,\ldots,\ell\}, \gamma' \in \Gamma\}$. 
\end{definition}

\begin{lemma}\label{lem:sum_of_children}
    Let $\Obf$ be an operator, and let $w \in \CW$ be a Rademacher-averaged TOI object. Suppose that $w$ has at least one downstairs Rademacher, i.e., it can be written as $w= ((\gamma_1,\ldots,\gamma_h), (S_\ell, l_\ell, \ldots,l_1,S_0))$ with $h \geq 1$. Then, the action of $w$ on $\Obf$ is given by the sum over all of its children, acting on $\Obf$:
    \begin{align}
        V(w,\Obf) = \sum_{w' \in \CE(w)} V(w',\Obf) = \sum_{\alpha \in \Gamma} \sum_{i=0}^{\ell} V(\CE_{\alpha,i}(w),\Obf)\,,
    \end{align}
    where $V$ was defined in \cref{eq:V_wO_def}.
\end{lemma}
\begin{proof}
    This equation directly follows from \cref{lem:master_recursion_rule}. In particular, if $w= ((\gamma_1,\ldots,\gamma_h), (S_\ell, l_\ell, \ldots,l_1,S_0))$, then by the definition in \cref{eq:V_wO_def}, the quantity $V(w,\Obf)$ is equal to the left-hand side of \cref{eq:one_variable_recursion_unsimplified}. Meanwhile, each term on the right-hand side of \cref{eq:one_variable_recursion_unsimplified} is a Rademacher-averaged TOI associated with the object $\CE_{\alpha,i}(w)$: one identifies $\alpha = \gamma'$ and $i = k$ to get the corresponding terms. Moreover, each setting $(\alpha,i)$ that does not correspond to any term on the right-hand side of \cref{eq:one_variable_recursion_unsimplified} also satisfies $\CE_{\alpha,i}(w) = 0$. Thus, the sum over Rademacher-averaged TOIs, acting on $\Obf$, is equal to the sum over $V(w',\Obf)$. 
\end{proof}

\begin{lemma}[Length and number of Rademachers]\label{lem:children_length_number_downstairs_Rademachers}
    Let $w = ((\gamma_1,\ldots,\gamma_h), (S_\ell,l_{\ell},\ldots,S_1,l_1,S_0)) \in \CW$ be a Rademacher-averaged TOI with length $\ell$ and $h \geq 1$ downstairs Rademachers, and suppose $w' \neq 0$ is a child of $w$. Then the length of $w'$ is $\ell+1$ and the number of downstairs Rademachers of $w'$ lies in the set $\{h, h-1, h-2\}$. 

\end{lemma}
\begin{proof}
    This follows directly from the definition of a child (see \cref{def:child_Rademacher_averaged_TOI}). In particular, each case in \cref{eq:bar_gamma} defines a new tuple $\bar{\gamma}$ that has has length, respectively, $h-1,h-2$ or $h$. Furthermore, the new sequence of sets in \cref{eq:child_of_w} has length $\ell+1$ since $\bar\mu$ was inserted. 
\end{proof} 
The picture that emerges is as follows. We may apply \Cref{lem:master_recursion_rule} to a Rademacher-averaged TOI, yielding a sum over all its children (as clarified in \Cref{lem:sum_of_children}). The number of downstairs Rademachers in any of the children will not exceed that of the parent, as shown in \Cref{lem:children_length_number_downstairs_Rademachers}. The strategy will be to recursively apply \cref{lem:sum_of_children} to the children for which there is at least one downstairs Rademacher. After performing this process many times, the original expression will be given by a sum over many TOIs of varying lengths where all Rademacher prefactors have been canceled, and a remainder term of longer Rademacher-averaged TOIs which still contain Rademachers downstairs. We will show how to sum the former, and bound the latter. Next, we introduce a diagrammatic language to organize this expansion. 

\subsection{Rademacher-averaged TOIs as diagrams: some examples}

In this subsection, we introduce the diagrammatic language by way of example. A more formal definition can be found in the next subsection.

We are interested in bounding the norm of Rademacher-averaged TOIs of length 0, with $h$ downstairs Rademachers. For example, if $h=3$, we are interested in the following superoperator, for which we introduce a diagrammatic representation:
\begin{align}\label{eq:three_Rademacher_example}
    \BE s_{\gamma_1}s_{\gamma_2}s_{\gamma_3} \bTr \int_0^t  \rd t_2 \int_{0}^{t_1} \rd t_1 \e^{\lind^\dag t} = \BE s_{\gamma_1}s_{\gamma_2}s_{\gamma_3} \bb*{\varnothing} = \qquad \raisebox{-6pt}{\text{\drawDiagram{{{$\gamma_{1}$}/0,{$\gamma_{2}$}/0,{$\gamma_{3}$}/0}}}}
\end{align}
Each node of the diagram is labeled with an element of $\Gamma$. The node with an arrow represents a downstairs Rademacher. 

When we apply \Cref{lem:master_recursion_rule} to the expression in \cref{eq:three_Rademacher_example}, we have the equality
\begin{align}
    &\qquad \qquad \BE s_{\gamma_1}s_{\gamma_2}s_{\gamma_3} \bb*{\varnothing}= \BE s_{\gamma_2}s_{\gamma_3} \bb*{\varnothing,\gamma_1,\{\gamma_1\}} + \sum_{\gamma' \neq \gamma_1 } \BE s_{\gamma'}s_{\gamma_2}s_{\gamma_3} \bb*{\varnothing,\gamma_1\gamma',\{\gamma_1\}}\\
    &=\BE s_{\gamma_2}s_{\gamma_3} \bb*{\varnothing,\gamma_1,\{\gamma_1\}} + \BE s_{\gamma_3} \bb*{\varnothing,\gamma_1 \gamma_2,\{\gamma_1\}} + \BE s_{\gamma_2} \bb*{\varnothing,\gamma_1 \gamma_3,\{\gamma_1\}} + \sum_{\gamma' \in \Gamma\setminus \{\gamma_1,\gamma_2,\gamma_3\} } \BE s_{\gamma'}s_{\gamma_2}s_{\gamma_3} \bb*{\varnothing,\gamma_1\gamma',\{\gamma_1\}} \\
    ={}& \L(\adj{-8pt}{0.7}{\text{\drawDiagramIndexed{{{$\gamma_{1}$/0, $\gamma_{1}$/0}/1,{$\gamma_{2}$/0}/0,{$\gamma_{3}$/0}/0}}}}\R)
   + \L(\adj{-14pt}{0.7}{\text{\drawDiagramIndexed{{{$\gamma_{1}$/0, $\gamma_{2}$/0}/2,{$\gamma_{2}$/0}/1,{$\gamma_{3}$/0}/0}}}}\R) + 
   \L(\adj{-20pt}{0.7}{\text{\drawDiagramIndexed{{{$\gamma_{1}$/0, $\gamma_{3}$/0}/3,{$\gamma_{2}$/0}/0,{$\gamma_{3}$/0}/1}}}} \R) + 
   \sum_{\gamma' \in \Gamma \setminus\{\gamma_1,\gamma_2,\gamma_3\}}\L(\adj{-10pt}{0.7}{\text{\drawDiagramIndexed{{{$\gamma_{1}$/0, $\gamma'$/0}/0,{$\gamma_{2}$/0}/0,{$\gamma_{3}$/0}/0}}}} \R), \label{eq:example_expansion_diagram_length1}
\end{align}
where we have associated each Rademacher-averaged TOI in the resulting expression with a unique diagram. Each of these diagrams represents a child of the original diagram. 
In each new diagram, the arrow representing the first Rademacher $\gamma_1$ was replaced with an edge to a new node, which has a label in $\gamma' \in \Gamma$. In the first term, the node label is $\gamma' = \gamma_1$, so the node is filled and that arrow eliminated, representing the cancellation of $s_{\gamma_1}$. In the next two terms, the new node is labeled with a label that agrees with one of the other Rademachers (corresponding to $\gamma' = \gamma_2$ and $\gamma' = \gamma_3$), and these Rademachers cancel together, denoted by the dashed line and the removal of the arrow on both paths. In the final term, there is no cancellation. The 0 label next to the solid edge represents the fact that we are dealing with indexed by $k=0$ in the sum in \cref{eq:one_variable_recursion_unsimplified}---in this case $k=0$ was the only option.

In general, if the new node introduced has label $\gamma'$ and the new edge introduced has label $i$, then the new diagram represents the $(\gamma',i)$th child of the original diagram. This means we brought down $s_{\gamma'}$ by applying Duhemal to the exponential with set indexed by $i$.

We may now recurse and apply \Cref{lem:master_recursion_rule} to each of the terms on the right-hand side of \cref{eq:example_expansion_diagram_length1} that still have unpaired downstairs Rademachers (in this case, all of them). As before, in the applicaton of \Cref{lem:master_recursion_rule}, we choose  $\gamma$ to be the leftmost Rademacher listed downstairs, corresponding the leftmost arrow appearing in the diagram. For example, a particular summand of the final term can be expanded as
\begin{align}
    &\BE s_{\gamma'}s_{\gamma_2}s_{\gamma_3} \bb*{\varnothing,\gamma_1\gamma',\{\gamma_1\}}= \L(\adj{-10pt}{0.7}{\text{\drawDiagramIndexed{{{$\gamma_{1}$/0, $\gamma_{'}$/0}/0,{$\gamma_{2}$/0}/0,{$\gamma_{3}$/0}/0}}}} \R) \nonumber \\
    ={}& \L(\adj{-14pt}{0.7}{\text{\drawDiagramIndexed{{{$\gamma_{1}$/0, $\gamma'$/0, $\gamma'$/0}/1,{$\gamma_{2}$/0}/0,{$\gamma_{3}$/0}/0}}}}\R)
   + \L(\adj{-18pt}{0.7}{\text{\drawDiagramIndexed{{{$\gamma_{1}$/0,$\gamma'$/0 ,$\gamma_{2}$/0}/2,{$\gamma_{2}$/0}/1,{$\gamma_{3}$/0}/0}}}}\R) 
   + \L(\adj{-24pt}{0.7}{\text{\drawDiagramIndexed{{{$\gamma_{1}$/0,$\gamma'$/0 ,$\gamma_{3}$/0}/3,{$\gamma_{2}$/0}/0,{$\gamma_{3}$/0}/1}}}} \R) + \sum_{\gamma'' \in \Gamma \setminus\{\gamma',\gamma_1,\gamma_2,\gamma_3\}}\L(\adj{-16pt}{0.7}{\text{\drawDiagramIndexed{{{$\gamma_{1}$/0, $\gamma'$/0, $\gamma''$/0}/0,{$\gamma_{2}$/0}/0,{$\gamma_{3}$/0}/0}}}} \R) \nonumber \\
   +& \L(\adj{-14pt}{0.7}{\text{\drawDiagramIndexed{{{$\gamma_{1}$/0, $\gamma'$/0, $\gamma'$/1}/1,{$\gamma_{2}$/0}/0,{$\gamma_{3}$/0}/0}}}}\R)
   + \L(\adj{-18pt}{0.7}{\text{\drawDiagramIndexed{{{$\gamma_{1}$/0,$\gamma'$/0 ,$\gamma_{2}$/1}/2,{$\gamma_{2}$/0}/1,{$\gamma_{3}$/0}/0}}}}\R) 
   + \L(\adj{-24pt}{0.7}{\text{\drawDiagramIndexed{{{$\gamma_{1}$/0,$\gamma'$/0 ,$\gamma_{3}$/1}/3,{$\gamma_{2}$/0}/0,{$\gamma_{3}$/0}/1}}}} \R) + \sum_{\gamma'' \in \Gamma \setminus\{\gamma',\gamma_2,\gamma_3\}}\L(\adj{-16pt}{0.7}{\text{\drawDiagramIndexed{{{$\gamma_{1}$/0, $\gamma'$/0, $\gamma''$/1}/0,{$\gamma_{2}$/0}/0,{$\gamma_{3}$/0}/0}}}} \R) \label{eq:diagram_example_expanded_twice}
\end{align}
A few comments are in order. First, we see that for each term on the right-hand side, the first (leftmost) path has grown to have two solid edges by introducing a new node with a new label $\gamma''$. As before, if that label is equal to $\gamma'$, forming an edge with the same label on both adjacent nodes, then the number of Rademachers downstairs decreases by 1 and the new node is filled in. The path can also pair up with one of the other two downstairs Rademachers if $\gamma'' = \gamma_2$ or $\gamma'' = \gamma_3$, indicated by a dashed line. Finally, it can extend without pairing up, preserving the arrow on all three paths. The second solid edge is labeled with a 0 or a 1, indicating the value of $k$ that the term corresponds to in the expansion of \cref{eq:one_variable_recursion_unsimplified}. Note that the $k=0$, $\gamma'' = \gamma_1$ term is missing, because in that case the condition $\gamma'' \not\in S_k$ would be violated (as $S_0 = \{\gamma_1\}$ for this particular term, see left-hand side of \cref{eq:diagram_example_expanded_twice}).

As another example, if we take instead the third term from \cref{eq:example_expansion_diagram_length1}, we expand to obtain
\begin{align}
    &\BE s_{\gamma_2}\bb*{\varnothing, \gamma_1\gamma_3, \{\gamma_1\}} = \L(\adj{-20pt}{0.7}{\text{\drawDiagramIndexed{{{$\gamma_{1}$/0, $\gamma_{3}$/0}/3,{$\gamma_{2}$/0}/0,{$\gamma_{3}$/0}/1}}}} \R)  \\
    ={}& \L(\adj{-20pt}{0.7}{\text{\drawDiagramIndexed{{{$\gamma_{1}$/0, $\gamma_{3}$/0}/3,{$\gamma_{2}$/0,$\gamma_2$/0}/2,{$\gamma_{3}$/0}/1}}}} \R)
    +\sum_{\gamma' \in \Gamma \setminus\{\gamma_1\}}\L(\adj{-20pt}{0.7}{\text{\drawDiagramIndexed{{{$\gamma_{1}$/0, $\gamma_{3}$/0}/3,{$\gamma_{2}$/0,$\gamma'$/0}/0,{$\gamma_{3}$/0}/1}}}} \R) +
    \L(\adj{-20pt}{0.7}{\text{\drawDiagramIndexed{{{$\gamma_{1}$/0, $\gamma_{3}$/0}/3,{$\gamma_{2}$/0,$\gamma_2$/1}/2,{$\gamma_{3}$/0}/1}}}} \R)
    +\sum_{\gamma' \in \Gamma }\L(\adj{-20pt}{0.7}{\text{\drawDiagramIndexed{{{$\gamma_{1}$/0, $\gamma_{3}$/0}/3,{$\gamma_{2}$/0,$\gamma'$/1}/0,{$\gamma_{3}$/0}/1}}}} \R)
\end{align}
Here, we see that the first and third terms have no remaining downstairs Rademachers (i.e., no arrows). The second term, corresponding to $k=0$, omits the $\gamma' = \gamma_1$ option because here the condition $\gamma' \not\in S_k$ would be violated. 

As we will define in the next subsection, if we a draw a diagram that violates $\gamma' \not\in S_k$, we simply associate that diagram with the 0 superoperator, so that it may be included in the sum without any effect.

\subsection{Rademacher-averaged TOIs as diagrams: formal definition}\label{sec:diagrams_def_formal}

We now formally write the rules to generate a diagram and its associated Rademacher-averaged TOI.
\begin{definition}[Generating a diagram]\label{def:algorithm}
    Let $(\beta_{1},\ldots, \beta_g)$ be any tuple of distinct elements of  $\Gamma$, with $g \geq 1$. Let $(\alpha_1,\ldots,\alpha_{\ell})$ be a sequence of elements of $\Gamma$, and let $(i_1,i_2,\ldots,i_{\ell})$ be a tuple of integers satisfying $0 \leq i_j \leq j-1$ for all $j=1,\ldots,\ell$. 
    \begin{enumerate}
        \item Let $w = ((\beta_1,\ldots,\beta_g),(\varnothing))$ be an initial element of $\CW$ with $g$ downstairs Rademachers and length $0$ (associated with the superoperator $\BE s_{\beta_1}\ldots s_{\beta_g} \bb*{\varnothing}$, see \cref{def:Rademacher-averaged_TOI})
        \item Let $d$ be an initial diagram with $g$ open nodes, labeled $\beta_1,\ldots,\beta_g$ from left to right, with an arrow on each node (see \cref{eq:three_Rademacher_example} for the $g=3$ example).
        \item For $j=1$ to $j=\ell$:
        \begin{enumerate}
            \item Let $\beta$ be the label of the leftmost node of $d$ with an arrow. If there are none, then return ``sequence too long.''
            \item Update the diagram $d$ by removing the arrow from the node that is labeled with $\beta$. 
            \item Update the diagram $d$ by introducing a new node  directly above the node labeled by $\beta$, and adding a solid edge connecting the new node with the node labeled by $\beta$. Label the new node with $\alpha_j$. Label the new solid edge with the number $i_j$. 
            \item If $\alpha_j = \beta$, then update the diagram $d$ by filling in the new node. (Do not add any new arrows.)
            \item Else if there exists another node in the diagram that has an arrow and has a label equal to $\alpha_j$, then update the diagram $d$ by removing the arrow on that node, and connecting the two nodes with label $\alpha_j$ with a dotted line. (Do not add any new arrows.)
            \item Else: update diagram $d$ by adding an arrow on the new node labeled by $\alpha_j$. 
            \item Update $w$ to be equal to the $(\alpha_{j},i_j)$th child of $w$, as in \Cref{def:child_Rademacher_averaged_TOI}.
        \end{enumerate}
        \item Output $d$ and $w$.
    \end{enumerate}
\end{definition}
Note that the first recursion step has one exponential to choose from, the second recursion step has two, etc. The index $i_j$ will be the value of $k$ chosen in the recursion equation at step $j$ of recursion (see \cref{lem:master_recursion_rule}).

\begin{lemma}[Properties of diagram]\label{lem:diagram_properties}
    Let $(\beta_1,\ldots,\beta_g)$ be $g \geq 1$ distinct elmements of $\Gamma$. Let $(\alpha_1,\ldots,\alpha_{\ell})$ be a sequence of elements of $\Gamma$ and $(i_1,\ldots,i_\ell)$ be integers satisfying $0 \leq i_j \leq j-1$ for all $j = 1,\ldots,\ell$. 
    Let $d$ and $w$ be the diagram and the element of $\CW$ output by the algorithm defined in \cref{def:algorithm} on these inputs, assuming that the output is not ``sequence too long.'' Then, either $w = 0$, or else the following are true:
    \begin{enumerate}
        \item $d$ has $\ell$ solid edges, and $w$ has length $\ell$
        \item The number of nodes of $d$ with an arrow is equal to the number of downstairs Rademachers for $w$
        \item Denoting $w = ((\gamma_1,\ldots,\gamma_h),(S_\ell,l_\ell,\ldots,l_1,S_0))$, let $\beta$ be the number of values of $j$ for which the label $l_j \in \Gamma$, implying that there are $\ell-\beta$ values of $j$ for which $l_j \in \Gamma \times \Gamma$. Then the number of filled nodes of $d$ is equal to $\beta$. 
    \end{enumerate}
\end{lemma}
\begin{proof}
    Let $w_0, w_1,\ldots, w_\ell$ be the sequence of Rademacher-averaged TOIs generated by the $\ell$ steps of the algorithm (with output $w=w_\ell$), and let $d_0,d_1,\ldots,d_\ell$ be the sequence of diagrams generated (with output $d = d_\ell$). Each step of the algorithm adds a single solid edge to the diagram, and the original diagram has 0 edges, so $d_j$ has $j$ edges for all $j$. Furthermore, $w_j$ is a child of $w_{j-1}$, and the length of $w_0$ is 0, so applying \cref{lem:children_length_number_downstairs_Rademachers}, the length of $w_j$ is $j$. This establishes item 1. 
    
    To verify item 2, examine the generation of $d_j$ from $d_{j-1}$ in \cref{def:algorithm}, together with the definition of $w_j$ as the child of $w_{j-1}$ in \cref{def:child_Rademacher_averaged_TOI}. 
    Assume for induction that the tuple of labels on nodes of $d_{j-1}$ with an arrow (ordered from left to right), match the tuple of downstairs Rademachers of $w_{j-1}$; this is true in the base case when $j=1$ by construction. Then, we see that the three cases in \cref{eq:bar_gamma} correspond exactly to the three if-then-else cases of the algorithm: if $\beta = \alpha_j$ in the algorithm, then the number of arrows decreases by 1, and the tuple $\bar{\gamma}$ for $w_j$ has one fewer downstairs Radmeacher than that of $w_{j-1}$; if $\beta$ matches the label of another node with an arrow, then this corresponds to the $\alpha=\gamma_k$ branch in \cref{eq:bar_gamma}, and the number of arrows in $d_j$ and downstairs Rademachers in $w_j$ both decrease by 2 compared to $d_{j-1}$ and $w_{j-1}$; otherwise, the number of arrows and downstairs Rademachers stays the same. In all cases, the set of nodes with arrows on them in $d_j$, from left to right, matches the tuple $\bar{\gamma}$, confirming the inductive assumption.
    
    To verify item 3, note that when $d_j$ is formed from $d_{j-1}$ and a filled node is added, this corresponds to the first case ($\alpha = \gamma_1$) in the addition of a label $\bar{\mu}$ to $w_{j-1}$ to form $w_{j}$ in \cref{eq:bar_mu}. Thus, the number of filled nodes is equal to the number of labels with only one element, and the rest of the labels of $w$ have two elements.
\end{proof}

    Before proceeding, we provide a bit of intuition about how this formalism will be used in the lemmas that follow.  Let $(\beta_{1},\ldots, \beta_g)$ be any tuple of distinct elements of  $\Gamma$ and let $q \geq 0$ be an integer. We start with the Rademacher-averaged TOI $\BE s_{\gamma_1}\ldots s_{\gamma_g}\bb*{\varnothing}$, apply \Cref{lem:master_recursion_rule}, and then apply it again to every Rademacher-averaged TOI in the resulting expression that has at least one Rademacher prefactor, repeating this a total of $q$ times. We can see that each Rademacher-averaged TOI in the final expression corresponds to applying the algorithm in \Cref{def:algorithm} for some value $\ell \leq q$, some sequence $(\alpha_1,\ldots,\alpha_\ell)$, and some sequence $(i_1,\ldots,i_\ell)$, corresponding to a choice of $\gamma'$ and $k$ in the expansion at each of the $\ell$ steps.

\subsection{Bounding the norm of a diagram}

Continuing the discussion on intuition above, we see that, if we apply \Cref{lem:master_recursion_rule} onto an initial Rademacher-averaged TOI $\BE s_{\gamma_1}\ldots s_{\gamma_g} \bb*{\varnothing}$, generating the children of that Rademacher-averaged TOI, and then we apply \Cref{lem:master_recursion_rule} again onto each of children that have at least one downstairs Rademacher, recursing $q$ times, then every term that results can be associated with a diagram. That is, at the $j$th level of recursion, each child can be associated with a choice $(\alpha_j, i_j)$ in the algorithm, so for any Rademacher-averaged TOI in the final expression, there is a choice of sequence $(\alpha_1,\ldots,\alpha_{\ell})$, $(i_1,\ldots,i_\ell)$ that leads the algorithm to output that Rademacher-averaged TOI, along with the diagram. Each of these diagrams will have $\ell$ solid edges, with $\ell \leq q$. If the diagram has an arrow on it, then it must have $\ell=q$ since at each step of the recursion, it will always have had at least one downstairs Rademacher, and will be expanded by length 1. The $\ell$ solid edges are labeled by integers $i_1,i_2,\ldots, i_\ell$ obeying $0 \leq i_k \leq k-1$ (with $k$ increasing in the order they were added to the diagram). We may now assert the following. 

\begin{lemma}[Norm bound on a diagram]\label{lem:diagram_bound}
    Consider a diagram $d$  generated from the initial Rademacher-averaged TOI $\BE s_{\gamma_1}\ldots s_{\gamma_g} \bb*{\varnothing}$, using the sequence $(\alpha_1,\ldots,\alpha_{\ell})$ and $(i_1,\ldots,i_\ell)$. Suppose the diagram has $\ell$ solid edges, $\beta$ filled in nodes, and $\alpha$ dashed lines. Let $(\omega_1\omega_1', \omega_2\omega_2',\ldots,\omega_{\ell-\beta}\omega'_{\ell-\beta})$ denote the pairs of labels in $\Gamma$ that are adjacent to solid edges that connect two open nodes. Let $(\xi_1,\ldots,\xi_\beta)$ denote the set of labels in $\Gamma$ for the $\beta$ solid edges connecting an unfilled node to a filled node (each with that label). Let $w$ be the Rademacher-averaged TOI that is output along with this diagram by the algorithm in \cref{def:algorithm}. Then the norm of the Rademacher-averaged TOI (as a superoperator) can be bounded as 
    \begin{equation}\label{eq:norm_d_norm_w_definition_and_bound}
        \norm{d} := \norm{w} := \sup_{\Obf} \frac{\labs{V(w,\Obf)}}{\norm{\Obf}} \leq \frac{t^{\ell+2}}{(\ell+2)!}\L(\prod_{j=1}^{\ell-\beta} \norm{\lind^\dag_{\omega_j\omega'_j}} \R)\L(\prod_{i = 1}^\beta \norm{\lind^\dag_{\xi_i}}\R).
    \end{equation}
\end{lemma}
\begin{proof}
    Recall that we are bounding the norm of a Rademacher-averaged-TOI via the definition of $V(w,\Obf)$ (\cref{def:Rademacher-averaged_TOI}.) 
    The diagram corresponds to a Rademacher-averaged TOI of the form 
    \begin{align}
    \BE s_{\gamma_1}\ldots s_{\gamma_h}\bb*{S_\ell, l_\ell, S_{\ell-1},\ldots, S_1,l_1,S_0}.
    \end{align}
    During the creation of the diagram, a solid edge between open nodes $\omega_{j}$ and $\omega'_j$ corresponds to the inclusion of a label $l_r=\omega_{j}\omega'_j$ for one of the values of $r$. 
    A solid edge between an open node and filled node with the same label $\xi_i$ corresponds to  the label $l_r = \xi_i$ for some value of $r$.  Applying \cref{lem:bouding_rad_avg_toi} completes the proof. Note that the choice of $(i_1,\ldots,i_{\ell})$ impacts the order in which the set of labels appear, as well as the sets $S_r$, but these do not impact the bound in \cref{lem:bouding_rad_avg_toi}.
\end{proof}

What this lemma shows is that we can upper bound the norm of a diagram by a product of values assigned to each of its components.
For example, we have

\begin{equation}
\norm{\text{\adj{-36pt}{0.8}{
\drawDiagramIndexed{{{$\gamma_{1}$/0,$\beta_1$/0,$\beta_{2}$/0,$\gamma_4$/2}/4,{$\gamma_{2}$/0,$\gamma_{2}$/1}/2,{$\gamma_{3}$/0,$\beta_{3}$/3,$\beta_{4}$/0}/0,{$\gamma_{4}$}/1,{$\gamma_{5}$}/0}}}}} \leq \frac{t^8\norm{\lind^\dag_{\gamma_1\beta_1}}\norm{\lind^\dag_{\beta_1\beta_2}}\norm{\lind^\dag_{\beta_2\gamma_4}}\norm{\lind^\dag_{\gamma_2}}\norm{\lind^\dag_{\gamma_3\beta_3}}\norm{\lind^\dag_{\beta_3\beta_4}}}{8!}.
\end{equation}
Note that the values of $i_1,\ldots,i_\ell$ that appear on the edges of the diagram do not affect the bound.

\subsection{Proof of \Cref{prop:superoperator_norm}}\label{sec:proof_superoperator_norm}

We define symbols that will serve to compute bounds:
\begin{align}
    \quad F_1(\gamma,\gamma') := \norm{\lind^\dag_{\gamma\gamma'}},  \quad \text{ and, for } p\ge 2, \quad F_p(\gamma,\gamma') := \sum_{\beta_1,\ldots,\beta_{p-1}} \norm{\lind^\dag_{\gamma\beta_1}} \norm{\lind^\dag_{\beta_1\beta_2}} \cdots \norm{\lind^\dag_{\beta_{p-2}\beta_{p-1}}} \norm{\lind^\dag_{\beta_{p-1} \gamma'}}, \label{eq:F_def}
\end{align}
\begin{align}
    G_1(\gamma) := \norm{\lind^\dag_\gamma}, \quad \text{ and, for } p\ge 2, \quad G_p(\gamma) := \sum_{\beta_1,\ldots,\beta_{p-1}} \norm{\lind^\dag_{\gamma\beta_1}} \norm{\lind^\dag_{\beta_1\beta_2}} \cdots \norm{\lind^\dag_{\beta_{p-2}\beta_{p-1}}} \norm{\lind^\dag_{\beta_{p-1}}},\label{eq:G_def}
\end{align}
\begin{align}
    \Phi_0(\gamma) := 1, \quad & \Phi_1(\gamma) := \sum_{\beta'} \norm{\lind^\dag_{\gamma\beta'}},  \quad \text{ and, for } \nonumber \\
    &p\ge 2, \quad \Phi_p(\gamma) := \sum_{\gamma'} \sum_{\beta_1,\ldots,\beta_{p-1}} \norm{\lind^\dag_{\gamma\beta_1}} \norm{\lind^\dag_{\beta_1\beta_2}} \cdots \norm{\lind^\dag_{\beta_{p-2}\beta_{p-1}}} \norm{\lind^\dag_{\beta_{p-1} \gamma'}} \label{eq:Phi_def}.
\end{align}
In what follows, $F_p, G_p, \Phi_p$ will each correspond to a bound on the norm of a sum over diagrams with a single connected path. The quantity $F_p(\gamma,\gamma')$ corresponds to all paths with $p$ solid edges, without any filled nodes or arrows, beginning at $\gamma$ and ending at $\gamma'$.  The quantity $G_p(\gamma)$ corresponds to all paths of length $p$ beginning at $\gamma$ and ending with a filled node. The quantity $\Phi_p(\gamma)$ corresponds to all paths of length $p$ beginning at $\gamma$ and ending with an arrow. 

The following lemma shows how to bound a term of the form $\BE s_{\gamma_{1}}s_{\gamma_{2}}\ldots s_{\gamma_{g}} \bb*{\varnothing}$.
\begin{lemma}[Sum over contributions of all paths]\label{lem:V_hat_bound_FGPhi}
    Let $g,q$ be positive integers. Let  $\gamma_{1},\ldots,\gamma_{g}$ be distinct elements of $\Gamma$. Then the following bound holds:
    \begin{align}\label{eq:V_hat_paths}
        &\norm{\BE s_{\gamma_{1}}s_{\gamma_{2}}\ldots s_{\gamma_{g}} \bb*{\varnothing}} \nonumber \\
        &\leq \sum_{\ell=\alpha+\beta}^q \frac{t^{\ell+2}}{(\ell+1)(\ell+2)} \sum_{\pi \in \symg(g)} \sum_{\substack{\alpha,\beta\in \mathbb{Z}_{\ge 0} \\ 2\alpha + \beta = g}} \sum_{\substack{\ell_1+\ldots+\ell_\alpha+r_1+\ldots+r_\beta=\ell \\ \ell_j, r_k \in \mathbb{Z}_{>0}}} \Bigg[\nonumber \\
        &\quad F_{\ell_1}(\gamma_{\pi(1)},\gamma_{\pi(2)}) \cdots F_{\ell_\alpha}(\gamma_{\pi(2\alpha-1)},\gamma_{\pi(2\alpha)}) G_{r_1}(\gamma_{\pi(2\alpha+1)}) \cdots G_{r_\beta}(\gamma_{\pi(2\alpha+\beta)}) \Bigg]\nonumber \\
        &+ t^{q+2} \sum_{\pi \in \symg(g)} \sum_{\substack{\alpha,\beta \in \mathbb{Z}_{\ge 0} \\ 2\alpha + \beta < g}} \sum_{\substack{\ell_1+\ldots+\ell_\alpha+r_1+\ldots+r_\beta+v=q \\ \ell_j, r_k \in \mathbb{Z}_{>0}, v\in \mathbb{Z}_{\geq 0}}} \Bigg[ \nonumber \\
        & F_{\ell_1}(\gamma_{\pi(1)},\gamma_{\pi(2)}) \cdots F_{\ell_\alpha}(\gamma_{\pi(2\alpha-1)},\gamma_{\pi(2\alpha)}) G_{r_1}(\gamma_{\pi(2\alpha+1)}) \cdots G_{r_\beta}(\gamma_{\pi(2\alpha+\beta)}) \Phi_{v}(\gamma_{\pi(2\alpha+\beta+1)}) \Phi_0(\gamma_{\pi(2\alpha+\beta+2)}) \cdots \Phi_0(\gamma_{\pi(g)})\Bigg].
    \end{align}
    Note that the index $v$ ranges from $0$ to $\ell$ whereas $r_1,\ldots,r_\beta,\ell_1,\ldots,\ell_\alpha$ range from $1$ to $\ell$.
\end{lemma}
\begin{proof}
    By the definition of the Rademacher-averaged TOI (\cref{def:Rademacher-averaged_TOI}), the superoperator $\BE \gamma_{1}\gamma_{2}\ldots \gamma_{g} \bb*{\varnothing}$, acting on $\Obf$, is equal to $V(w_0,\Obf)$, where $w_0:= ((\gamma_1,\ldots,\gamma_g),(\varnothing))$, which has length 0 and $g \geq 1$ downstairs Rademachers. The object $w_0$ is associated with the diagram $d_0$ with $g$ isolated nodes with an arrow, as in \cref{eq:three_Rademacher_example}. 
    Let $\CC_1 \subset \CW$ denote the children of $w_0$ that are not equal to 0, that is, $\CC_1 = \CE(w_0) \setminus \{0\}$, as defined in \Cref{def:child_Rademacher_averaged_TOI}, and let $\CC_1^0$ be the subset of $\CC_1$ that have 0 downstairs Rademachers.
    Then we can recursively extend this definition by letting
    \begin{align}
        \CC_j = \bigcup_{w \in \CC_{j-1} \setminus \CC_{j-1}^0} \CE(w) \setminus \{0\}\,,
    \end{align}
    and then define $\CC_j^0$ to be the subset of $\CC_j$ that has 0 downstairs Rademachers.
    
    Then, with repeated applications of  \cref{lem:sum_of_children}, we can write
    \begin{align}
        V(w_0,\Obf) = \sum_{w \in \CC_1} V(w,\Obf) &= \sum_{w \in \CC_1^0} V(w,\Obf) + \sum_{w \in \CC_1 \setminus \CC_1^0} V(w,\Obf) \\
        &=\sum_{w \in \CC_1^0} V(w,\Obf) + \sum_{w \in \CC_2^0} V(w,\Obf) + \sum_{w \in \CC_2 \setminus \CC_2^0} V(w,\Obf) \\
        &= \ldots = \L(\sum_{\ell = 1}^q \sum_{w \in \CC_\ell^0} V(w,\Obf)\R) + \sum_{w \in \CC_q \setminus \CC_q^0} V(w,\Obf).
    \end{align}

    We now use the triangle inequality and assert (noting from the definition of $\norm{w}$ in \cref{eq:norm_d_norm_w_definition_and_bound} that $|V(w,\Obf)| \leq \norm{w}\norm{\Obf}$)
    \begin{align}
        \labs{V(w_0, \Obf)} &\leq \sum_{\ell = 1}^q \sum_{w \in \CC_\ell^0} \labs{V(w,\Obf)} + \sum_{w \in \CC_q \setminus \CC_q^0} \labs{V(w,\Obf)}  \\
        &\leq \norm{\Obf}\sum_{\ell = 1}^q \sum_{w \in \CC_\ell^0} \norm{w} + \norm{\Obf}\sum_{w \in \CC_q \setminus \CC_q^0} \norm{w}\,.
    \end{align}
    As $\norm{w_0} := \sup_{\Obf} \frac{\labs{V(w_0,\Obf)}}{\norm{\Obf}}$, we can drop the $\norm{\Obf}$ and write
    \begin{equation}
        \norm{w_0} \leq \sum_{\ell = 1}^q \sum_{w \in \CC_\ell^0} \norm{w} + \sum_{w \in \CC_q \setminus \CC_q^0} \norm{w}\,.
    \end{equation}
    For each element $w \in \CC_{\ell}^0$, there are sequences $(\alpha_1,\ldots,\alpha_\ell)$ and $(i_1,\ldots,i_{\ell})$ such that, on these inputs, running the algorithm in \Cref{def:algorithm} outputs $w$ and also the diagram $d$. By item 2 of \cref{lem:diagram_properties}, if $w \in \CC_\ell^0$, then $d$ has no arrows. Let $\CD_\ell^0$ denote the set of all diagrams that can be generated by some sequence pair, with length $\ell$ and no arrows, and let $\CD_q$ be the set of all diagrams with $q$ solid edges (and any number of arrows). Some of these diagrams may be generated together with the output $w=0$. Nevertheless, since $\norm{d} = \norm{w} = 0$ when $w=0$, we may still write
    \begin{equation}
        \norm{w_0} \leq \sum_{\ell = 1}^q \sum_{d \in \CD_\ell^0} \norm{d} + \sum_{d \in \CD_q \setminus \CD_q^0} \norm{d}
    \end{equation}
    including those diagrams in the sum. 
    Examining the algorithm in \cref{def:algorithm}, we observe that each of the $g$ initial paths within the diagram will have an arrow unless (i) at some step the path terminates with a filled node, removing 1 arrow, or (ii) at some step the path matches with a dotted line to another path, removing 2 arrows. Let $\beta$ denote the number of filled nodes in the diagram, and $\alpha$ the number of dotted lines. If $d$ has no arrows, it must therefore hold that $2\alpha+\beta = g$. There exists a permutation $\pi \in \symg(g)$, for which the path of index (counting from left to right) $\pi(2j-1)$ matches with the path of index $\pi(2j)$ and has a length $\ell_j$, for $j=1,\ldots,\alpha$, and paths of index $\pi(2\alpha+j)$ terminate in a filled node and have length $r_j$ for $j=1,\ldots,\beta$.  (In fact, there will be multiple such $\pi$.) It must hold that $\ell_1+\ldots+\ell_\alpha + r_1 + \ldots r_\beta = \ell$. Since $\ell_j > 0$ and $r_j > 0$ must hold, we see that $\ell \geq \alpha+\beta$ for this diagram to be possible. 

    We now impose the upper bound in \cref{lem:diagram_bound}. Note that it is independent of the choice of $i_{1},\ldots,i_{\ell}$, which ultimately contributes a factor of $\ell!$ to the bound. Also, it has a factor $t^{\ell+2}/(\ell+2)!$ regardless of the details of the diagram. Finally, the diagram-specific multiplicative factors  can each be associated with an edge of the diagram, and are determined by the labels of the nodes adjacent to that edge. Thus, we can compute independently the contributions of each of the $\alpha + \beta$ paths (of lengths $\ell_1,\ldots,\ell_\alpha, r_1,\ldots, r_\beta)$ that compose the diagram, and multiply them together. All diagrams will be included at least once in the sum if we sum over all choices of $\pi$, $\alpha$, $\beta$, $\ell_1,\ldots,\ell_\alpha$, $r_1,\ldots, r_\beta$,  $i_1,\ldots, i_{\ell}$, and the values of the label of all the open nodes that are adjacent to two edges (the filled nodes must agree with the node that preceded them and thus are not summed over). Consider the path of index $\pi(2j-1)$, with $\ell_j$ such edges, for $j \in \{1,2,\ldots,\alpha\}$. Summing over the $\ell_j$ choices of labels for open nodes of the path gives a contribution to the bound of \cref{lem:diagram_bound} equal to $F_{\ell_j}(\gamma_{\pi(2j-1)}\gamma_{\pi(2j)})$, defined in \cref{eq:G_def}. Similarly, the path of index $\pi(2\alpha+j)$, with $r_j$ such edges, for $j \in \{1,2,\ldots,\beta\}$ contributes to the bound of \cref{lem:diagram_bound} as $G_{r_j}(\gamma_{\pi(1)}\gamma_{\pi(2)})$, defined in \cref{eq:F_def}. Overall, we can say that
    \begin{align}
        &\sum_{d \in \CD_\ell^0} \norm{d} \leq \sum_{\substack{\alpha,\beta \in \mathbb{Z}_{\geq 0}\\ 2\alpha+\beta = g}}\frac{t^{\ell+2}}{(\ell+1)(\ell+2)} \sum_{\pi \in \symg(g)} \sum_{\substack{\alpha,\beta\in \mathbb{Z}_{\ge 0} \\ 2\alpha + \beta = g}} \sum_{\substack{\ell_1+\cdots+\ell_\alpha+r_1+\cdots+r_\beta=\ell \\ \ell_j, r_k \in \mathbb{Z}_{>0}}} \Bigg[\nonumber \\
        &\quad F_{\ell_1}(\gamma_{\pi(1)},\gamma_{\pi(2)}) \cdots F_{\ell_\alpha}(\gamma_{\pi(2\alpha-1)},\gamma_{\pi(2\alpha)})) G_{r_1}(\gamma_{\pi(2\alpha+1)}) \cdots G_{r_\beta}(\gamma_{\pi(2\alpha+\beta)})\Bigg].
    \end{align}

    Similarly, we may examine elements $d \in \CD_{q} \setminus \CD_{q}^0$, which have at least one path that ends in an arrow. Let the number of paths with an arrow be $\theta$, and $\alpha$ and $\beta$ defined as before, so that $2\alpha + \beta + \theta = g$.  The rules of the algorithm in \cref{def:algorithm} ensure that the leftmost path will continue growing until its arrow is removed. Thus, there can be at most one path with an arrow of length greater than 0. Denote the length of this path by $v \geq 0$, and note that there exists a $\pi$ for which its index is $\pi(2\alpha + \beta + 1)$ (while maintaining the index relationships above as well).  We thus have $\ell_1+ \ldots + \ell_\alpha + r_1 + \ldots + r_\beta + v = \ell$, and the contributions of the paths without arrows as above. The contribution of the path of length $v$ with an arrow is given by $\Phi(\gamma_{\pi(2\alpha + \beta + 1)})$. The other $\ell-2\alpha-\beta-1$ paths contribute only a factor of $\Phi_0(\cdot) = 1$. This allows us to say that
\begin{align}
     &\sum_{d \in \CD_{q} \setminus \CD_{q}^0} \norm{d} \leq t^{q+2} \sum_{\pi \in \symg(g)} \sum_{\substack{\alpha,\beta \in \mathbb{Z}_{\ge 0} \\ 2\alpha + \beta  < g}} \sum_{\substack{\ell_1+\ldots+\ell_\alpha+r_1+\ldots+r_\beta+v=q \\ \ell_j, r_k \in \mathbb{Z}_{>0}, v\in \mathbb{Z}_{\geq 0}}} \Bigg[ \nonumber \\
        & F_{\ell_1}(\gamma_{\pi(1)},\gamma_{\pi(2)}) \cdots F_{\ell_\alpha}(\gamma_{\pi(2\alpha-1)},\gamma_{\pi(2\alpha)}) G_{r_1}(\gamma_{\pi(2\alpha+1)}) \cdots G_{r_\beta}(\gamma_{\pi(2\alpha+\beta)}) \Phi_{v}(\gamma_{\pi(2\alpha+\beta+1)}) \Bigg].
\end{align}
Combining these contributions yields the lemma.

\end{proof}

We can now use the bounds on $F_l,G_l,\Phi_l$ from \cref{sec:lindbladin_bounds} to arrive at the following lemma:
\begin{customthm}{\Cref{prop:superoperator_norm}}[restated]\label{prop:superoperator_norm_appendix}
Fix a subset $U \subset \Gamma$ with $|U| \leq 5$. Suppose that $y$ and $t$ satisfy the relations     
$t\ge0$, $\aloc k t < 1$, and $y^2 \HlocPower{2} \aloc k < 1/8 $. Then 
\begin{align}
    \norm{ \BE\left[\left(\prod_{\alpha \in U} s_\alpha\right) \bTr \int_{t > t_2> t_1> 0} \rd t_1 \rd t_2 \e^{\lind^\dag t_1}(\cdot)\right]} \leq \CO\L(t^2\R) \cdot |y|^{|U|} \cdot \prod_{\alpha \in U} h_\alpha.
\end{align}
\end{customthm}
\begin{proof}
    Let $U = \{\gamma_1,\ldots,\gamma_g\}$ with $|U| = g$. Then, the left-hand side is equivalent to the quantity $\norm{\BE \gamma_{1}\gamma_{2}\ldots \gamma_{g} \bb*{\varnothing}}$. Applying the bounds offered by \cref{lem:bounds_on_paths} to \cref{lem:V_hat_bound_FGPhi}, we get
    \begin{align}
        \norm{\BE \gamma_{1}\gamma_{2}\ldots \gamma_{g} \bb*{\varnothing}} &\leq \left(\prod_{j=1}^g h_{\gamma_{j}} \right) \sum_{\pi \in \symg(g)} \sum_{\substack{\alpha,\beta,\theta \in \mathbb{Z}_{\ge 0} \\ 2\alpha + \beta = g}} \sum_{\ell=\alpha+\beta}^q \frac{t^{\ell+2}}{(\ell+1)(\ell+2)} \nonumber \\
        &\quad\quad\times \sum_{\substack{\ell_1+\ldots+\ell_\alpha+r_1+\ldots+r_\beta=\ell \\ \ell_j, r_k \in \mathbb{Z}_{>0}}} \left( 8|y|^2 k\aloc \right)^{\alpha} \left( 8|y| k\aloc \right)^{\beta} \left( 8|y|^2 k\aloc \Hloc^2 \right)^{\ell - \alpha - \beta} \nonumber \\
        &+ \frac{t^{q+2}}{(q+1)(q+2)} \sum_{\pi \in \symg(g)}  \sum_{\substack{\alpha,\beta\in \mathbb{Z}_{\ge 0} \\ 2\alpha + \beta < g}} \left(\prod_{j=1}^{g-\theta} h_{\gamma_{\pi(j)}} \right) \sum_{\substack{\ell_1+\ldots+\ell_\alpha+r_1+\ldots+r_\beta+v=q \\ \ell_j, r_k \in \mathbb{Z}_{>0}, v\in \mathbb{Z}_{\geq 0}}} \Bigg[  \nonumber\\
        &\qquad \qquad \left(\sum_\gamma h_\gamma \right) \left( 8|y|^2 k\aloc \right)^{\alpha+1} \left( 8|y| k\aloc \right)^{\beta} \left( 8|y|^2 k\aloc \Hloc^2 \right)^{q - \alpha - \beta - 1}\Bigg]. \label{eq:terminated_diagrams_and_remainder}
    \end{align}
    We see now that each term in the sum over $\pi$ is identical, so this can be replaced with a $g! = \CO(1)$, since $g \leq 5$. Furthermore, each term in the sum over $\ell_1,\ldots,\ell_\alpha, r_1,\ldots,r_\beta$ is identicall. The number of ways to choose integers $\ell_1,\ldots,\ell_\alpha,r_1,\ldots,r_\beta$, each at least 1, such that they sum to $\ell$, is upper bounded by $\binom{\ell-1}{\alpha+\beta-1} \leq \ell^{\alpha+\beta}$. This allows the first term above to be rewritten as
    \begin{align}
        \CO(t^2) \cdot \left(\prod_{j=1}^g h_{\gamma_{j}} \right)  \sum_{\substack{\alpha,\beta,\theta \in \mathbb{Z}_{\ge 0} \\ 2\alpha + \beta = g}} |y|^{2\alpha+\beta} k^{\alpha+\beta} \alocPower{\alpha+\beta} t^{\alpha+\beta} \sum_{\ell=\alpha+\beta}^q \frac{\ell^{\alpha+\beta}\L(8|y|^2k\aloc t \HlocPower{2}\R)^{\ell-\alpha-\beta}}{(\ell+1)(\ell+2)} 
    \end{align}
    Under the assumptions on $y$, $t$, we have $|8|y|^2k\aloc t \HlocPower{2}|<1$, and thus, for any fixed $\alpha$ and $\beta$, the sum over $\ell$ converges and can be upper bounded by a $\CO(1)$ value (note that $\alpha,\beta \leq 5 = \CO(1)$).  Since $2\alpha+\beta = g$, we can pull out a $y^g$ factor. Furthermore, the sum over $\alpha,\beta$ of the quantity $(k \aloc t)^{\alpha+\beta}$ has a finite number of terms, and each term is upper bounded by 1, owing to the fact that we have assumed $\aloc k t < 1$. Thus the overall contribution can be absorbed into the $\CO(t^2)$.  Finally, we recover
    \begin{equation}
        \CO(t^2) \cdot \left(\prod_{j=1}^g h_{\gamma_{j}} \right) \cdot |y|^g.
    \end{equation}

    Now, we consider the second term in \cref{eq:terminated_diagrams_and_remainder}. We employ similar reasoning as above. Here we have $\ell_1+\ldots + \ell_\alpha + r_1 + \ldots+ r_\beta + v = q$, with $v\geq 0$. The number of ways to choose these integers is $\binom{q}{\alpha+\beta} \leq q^{\alpha+\beta}$. Thus, we can rewrite the second term as 
    \begin{align}
        &\CO(t^2) \cdot \left(\prod_{j=1}^g h_{\gamma_{j}} \right)  \L(\sum_\gamma h_\gamma \R)\frac{1}{\HlocPower{2}} \sum_{\substack{\alpha,\beta \in \mathbb{Z}_{\ge 0} \\ 2\alpha + \beta < g}} |y|^{2\alpha+\beta} k^{\alpha+\beta} \alocPower{\alpha+\beta} t^{\alpha+\beta} \frac{q^{\alpha+\beta}\L(8|y|^2k\aloc t \HlocPower{2}\R)^{q-\alpha-\beta}}{(q+1)(q+2)} \\
        &=\CO(t^2) \cdot \left(\prod_{j=1}^g h_{\gamma_{j}} \right)  \L(\sum_\gamma h_\gamma \R)\frac{1}{\HlocPower{2}} |y|^{g} \frac{q^{\alpha+\beta}\L(8|y|^2k\aloc t \HlocPower{2}\R)^{q-\alpha-\beta}}{(q+1)(q+2)}\,.
    \end{align}
    We observe that as we increase $q$ to $\infty$, the magnitude of the second term approaches 0. As the statement is true for all $q$, this implies that the norm of the Rademacher-averaged TOI is at most the first term, proving the theorem. 
\end{proof}

\section{Proof of \Cref{prop:sum_over_operators}}\label{sec:proof_sum_over_operators}

In this section, we prove~\Cref{prop:sum_over_operators}.
In the previous sections, the subscripts $\gamma$ in symbols like $\Hbf_\gamma$ have been drawn from $\Gamma$. Here, we allow the subscript to also be 0: let $\Hbf_0 := \Ibf$, $h_0=1$, $S(0) = \varnothing$, and $s_0 := 1$ (not random, unlike the other $s_\gamma$), and let $\bar{\Gamma} = \Gamma \cup \{0\}$. For $\gamma \in \bar{\Gamma}$, let
\begin{align}\label{eq:bar_b_agamma_def}
    \bar{b}_{a\gamma} := \begin{cases}
        1 & \text{if } \gamma = 0 \\
        2y & \text{if }\gamma\neq 0 \text{ and } [\Abf^a, \Hbf_\gamma] \neq 0 \\
        0 & \text{if } \gamma\neq 0 \text{ and } [\Abf^a, \Hbf_\gamma] = 0
    \end{cases}.
\end{align}
Since if $[\Abf^a, \Hbf_\gamma] \neq 0$ then $[\Abf^a, \Hbf_\gamma] = 2\Abf^a \Hbf_\gamma$, we can write the definition of $\Kbf^a$ in \cref{eq:Kagamma} as
\begin{equation}
    \Kbf^a = \sum_{\gamma \in \bar{\Gamma}} s_{\gamma} \bar{b}_{a\gamma} \Abf^{a} \Hbf_{\gamma},
\end{equation}
implying that (see \cref{eq:lindbladian_def})
\begin{align}
    \lind^{a \dag}(\Obf) ={}& \Kbf^{a \dag} \Obf \Kbf^a - \frac{1}{2} \Kbf^{a \dag} \Kbf^a \Obf - \frac{1}{2} \Obf \Kbf^{a \dag} \Kbf^a \\
    ={}&\sum_{\gamma_1,\gamma_2 \in \bar{\Gamma}} s_{\gamma_1} s_{\gamma_2} \bar{b}_{a\gamma_1}\bar{b}_{a\gamma_2} \L( \Hbf_{\gamma_1} \Abf^{a \dag}\Obf  \Abf^a \Hbf_{\gamma_2} - \frac{1}{2}\Hbf_{\gamma_1} \Hbf_{\gamma_2}\Obf  -\frac{1}{2} \Obf\Hbf_{\gamma_1} \Hbf_{\gamma_2} \R) \\
    ={}& \sum_{\gamma_1,\gamma_2 \in \bar{\Gamma}}  s_{\gamma_1}s_{\gamma_2}\CR_{\gamma_1\gamma_2}^{a \dag}(\Obf) \label{eq:lind_a_with_calR}
\end{align}
where we used the fact that $\Abf^{a \dag} \Abf^a = \Ibf$, from \cref{cond:commute}, and defined
\begin{align}
    & \qquad\CR_{\gamma_1\gamma_2}^{a \dag}(\Obf) := \nonumber \\
    &
     \frac{1}{2}\bar{b}_{a\gamma_1}\bar{b}_{a \gamma_2} \L(  \Hbf_{\gamma_1} \Abf^{a \dag} \Obf \Abf^a \Hbf_{\gamma_2}  + \Hbf_{\gamma_2} \Abf^{a \dag}\Obf  \Abf^a \Hbf_{\gamma_1} - \frac{1}{2}\Hbf_{\gamma_1} \Hbf_{\gamma_2}\Obf - \frac{1}{2}\Hbf_{\gamma_2} \Hbf_{\gamma_1}\Obf  -\frac{1}{2} \Obf\Hbf_{\gamma_1} \Hbf_{\gamma_2}- \frac{1}{2} \Obf\Hbf_{\gamma_2} \Hbf_{\gamma_1} \R), \label{eq:def_R_lind}
\end{align}
noting the explicit symmetrization of $\CR_{\gamma_1 \gamma_2}^{a\dag}$ with respect to its subscripts. The object $\CR_{\gamma_1\gamma_2}^{a \dag}$ is similar, but not identical, to the previously defined object $\lind_{\gamma_1\gamma_2}^{a \dag}$ (see \cref{eq:lind_gammagamma'^a}), but it always has two subscripts, and allows those subscripts to be drawn from $\bar{\Gamma}$ rather than $\Gamma$, making for simpler bookkeeping in this section.

For any integer $g\geq 0$, define
\begin{equation}\label{eq:def_h(gammas)}
    h(\gamma_0, \gamma_1,\gamma_2,\ldots,\gamma_{g-1}) = \prod_{\alpha \in \Gamma} (|y| h_\alpha)^{|\{j: \gamma_j = \alpha\}| \bmod 2}\,,
\end{equation}
where $h_{\alpha}$ is the strength of the associated term of $\Hbf$, as defined in \cref{cond:random_Hamiltonians}. 
That is, for each $\alpha \in \Gamma$, a factor of $(|y| h_{\alpha})$ is included if the tuple $(\gamma_0,\ldots,\gamma_{g-1})$ has an odd number of appearances of $\alpha$. 

\begin{customthm}{\cref{prop:sum_over_operators}}[restated]
Let $\Obf_U$ be defined as in the expansion of \cref{eq:expansion_T2_rademachers} for each subset $U\subset \Gamma$, with $|U|\leq 5$. Suppose that $y$ satisfies the relation $y^2 \HlocPower{2} \aloc k < 1/8 $. Then we have
\begin{align}
&\sum_{\substack{U \subset \Gamma \\ |U| \leq 5}}\left(\prod_{\alpha \in U} h_\alpha\right) |y|^{|U|} \norm{\Obf_U} \nonumber 
\leq{} \CO(1) \cdot |y|\alocPower{2} k^2 \HgloPower{2}
\end{align}
\end{customthm}
\begin{proof}
Observe that $\Obf_U$, as defined in \cref{eq:expansion_T2_rademachers}, is the operator obtained by expanding $\lind^\dag(\lind^\dag(\Hbf))$ as a polynomial in the Rademachers $s_{\gamma}$, and then (after cancelling any double appearances $s_{\gamma}^2=1$) setting $s_\gamma = 1$ when $\gamma \in U$ and $s_\gamma = 0$ when $\gamma \not\in U$.  To do this explicitly, we may write (see \cref{eq:lind_a_with_calR})
\begin{align}
    \lind^\dag(\Obf) = \sum_{a \in A}\sum_{\gamma_1,\gamma_2 \in \bar{\Gamma}} s_{\gamma_1}s_{\gamma_2}\CR_{\gamma_2\gamma_1}^{a \dag}(\Obf).
\end{align}
Recalling that $\Hbf = \sum_{\gamma_{0} \in \Gamma} s_{\gamma_0} \Hbf_{\gamma_0}$, we have
\begin{align}
    \lind^\dag(\lind^\dag(\Hbf)) = \sum_{a,a' \in A}\sum_{\gamma_0 \in \Gamma} \sum_{\gamma_1,\gamma_2,\gamma_3,\gamma_4 \in \bar{\Gamma}} s_{\gamma_0}s_{\gamma_1}s_{\gamma_2}s_{\gamma_3}s_{\gamma_4}\CR_{\gamma_4 \gamma_3}^{a' \dag}\CR_{\gamma_2 \gamma_1}^{a \dag}(\Hbf_{\gamma_0}) = \sum_{\substack{U \subset \Gamma \\ |U| \leq 5}} \L(\prod_{\alpha \in U} s_\alpha \R) \Obf_U
\end{align}
where $\Obf_U$ can be expressed as a sum over some set of terms of the form $\CR_{\gamma_4 \gamma_3}^{a' \dag}\CR_{\gamma_2 \gamma_1}^{a \dag}(\Hbf_{\gamma_0})$, as follows,
\begin{align}
    \Obf_U = \sum_{(\gamma_0,\gamma_1,a,\gamma_2,\gamma_3,a',\gamma_4)\in \CC_U} \CR_{\gamma_4 \gamma_3}^{a' \dag}\CR_{\gamma_2 \gamma_1}^{a \dag}(\Hbf_{\gamma_0})\,.
\end{align}
Here, the subset $\CC_U$ is defined implicitly to contain all choices of $(\gamma_0,\gamma_1,a,\gamma_2,\gamma_3,a',\gamma_4)$ that lead there to be an odd number of appearances of $s_\alpha$ for each $\alpha \in U$ and an even number of appearances of $s_{\alpha}$ for each $\alpha \not\in U$. The sets $\CC_U$ partition without overlap the set of all choices in $\Gamma \times \bar{\Gamma} \times A \times \bar{\Gamma}\times \bar{\Gamma}\times A \times \bar{\Gamma}$. By the triangle inequality, we may state that
\begin{align}
    \norm{\Obf_U} \leq \sum_{(\gamma_0,\gamma_1,a,\gamma_2,\gamma_3,a',\gamma_4)\in \CC_U} \norm{\CR_{\gamma_4 \gamma_3}^{a' \dag}\CR_{\gamma_2 \gamma_1}^{a \dag}(\Hbf_{\gamma_0})}\,.
\end{align}
We now observe that for a fixed $U$ and a choice $(\gamma_0,\gamma_1,a,\gamma_2,\gamma_3,a',\gamma_4) \in \CC_U$, it holds that $\prod_{\alpha \in U} |y|h_{\alpha} = h(\gamma_0,\gamma_1,\gamma_2,\gamma_3,\gamma_4)$, by construction. Hence, we may express the left-hand side of the quantity in the proposition as follows. 
\begin{align}\label{eq:sum_over_operators_as_sum_over_R}
    \sum_{\substack{U \subset \Gamma \\ |U| \leq 5}} \L(\prod_{\alpha \in U} |y|h_{\alpha}\R) \norm{\Obf_U} &\leq  \sum_{\substack{U \subset \Gamma \\ |U| \leq 5}} \L(\prod_{\alpha \in U} |y|h_{\alpha}\R) \sum_{(\gamma_0,\gamma_1,a,\gamma_2,\gamma_3,a',\gamma_4)\in \CC_U} \norm{\CR_{\gamma_4 \gamma_3}^{a' \dag}\CR_{\gamma_2 \gamma_1}^{a \dag}(\Hbf_{\gamma_0})}
    \\
    &=  \sum_{\substack{U \subset \Gamma \\ |U| \leq 5}}  \sum_{(\gamma_0,\gamma_1,a,\gamma_2,\gamma_3,a',\gamma_4)\in \CC_U}h(\gamma_0,\gamma_1,\gamma_2,\gamma_3,\gamma_4) \norm{\CR_{\gamma_4 \gamma_3}^{a' \dag}\CR_{\gamma_2 \gamma_1}^{a \dag}(\Hbf_{\gamma_0})} \\
    &= \sum_{a,a' \in A}\sum_{\gamma_0 \in \Gamma} \sum_{\gamma_1,\gamma_2,\gamma_3,\gamma_4 \in \bar{\Gamma}} h(\gamma_0, \gamma_1,\gamma_2,\gamma_3,\gamma_4) \norm{\CR_{\gamma_4 \gamma_3}^{a' \dag}\CR_{\gamma_2 \gamma_1}^{a \dag}(\Hbf_{\gamma_0})}. \label{eq:sum_at_beginning}
\end{align}
  
  The sum  above is the starting point for our analysis. Let $\hat{\gamma} = (\gamma_0,\gamma_1,a,\gamma_2,\gamma_3,a',\gamma_4)$ denote a choice for the tuple of parameters in the sum, drawn from the set $\Gamma \times \bar{\Gamma} \times A \times \bar{\Gamma} \times \bar{\Gamma} \times A \times \bar{\Gamma}$. We will identify a subset $\CC$ such that if $\hat{\gamma} \not\in \CC$, the associated term in the sum in \cref{eq:sum_over_operators_as_sum_over_R} is guaranteed to vanish, modulo a permutation of the coordinates. Then, we will evaluate the sum that results. We will utilize the following lemmas. The first lemma establishes some facts about the action of $\mathcal{R}^{a\dag}_{\gamma_2\gamma_1}$ onto an operator $\Obf$; namely, a few conditions under which it vanishes, and a restriction on the growth of the support. 

\begin{lemma}[Locality estimate]\label{lem:R_applied_to_O}
    Suppose an operator $\Obf$ can be decomposed as a linear combination of products of $\Hbf_\gamma$ for various $\gamma$, and let $Z \subset [n]$ be the union of $S(\gamma)$ for all $\Hbf_\gamma$ appearing in this decomposition. (This implies that $\Obf$ is supported on $Z$ and that one can write $\Obf = \Obf_Z \otimes \Ibf_{[n]\setminus Z}$.) Assume that \cref{cond:random_Hamiltonians}, \cref{cond:commute}, and \cref{cond:aloc} hold. Then, the following statements hold:
    \begin{enumerate}[(i)]
    \item $\CR_{\gamma_2\gamma_1}^{a\dag}(\Obf)$ is supported on $Z \cup S(\gamma_1) \cup S(\gamma_2)$.
    \item $
        \norm{\CR_{\gamma_2\gamma_1}^{a\dag}(\Obf)} \leq 2\labs{\bar{b}_{a \gamma_1}\bar{b}_{a \gamma_2}} h_{\gamma_1} h_{\gamma_2} \norm{\Obf}$.
    \item Suppose that $\gamma_1 = 0$ or that $\gamma_2 = 0$. If $\Abf^a$ is not supported on $Z$, then $\CR_{\gamma_2\gamma_1}^{a\dag}(\Obf) = 0$.
    \item Suppose that $\gamma_1,\gamma_2 \in \bar{\Gamma}$ and $a \in A$, and that either $\gamma_1 \neq 0$ or $\gamma_2 \neq 0$ (or both). If $Z \cap (S(\gamma_1) \cup S(\gamma_2)) =  \varnothing$, then $\CR_{\gamma_2\gamma_1}^{a\dag}(\Obf) = 0$. 
    \end{enumerate}
\end{lemma}
\begin{proof}
    We begin with a general observation that will be relevant in multiple cases. We have assumed $\Obf$ is given by a linear combination of products of $\Hbf_{\gamma}$ factors whose support is contained in $Z$. If $\Abf^a$ is not supported on $Z$, then \cref{cond:commute} implies that $[\Abf^a, \Hbf_{\gamma}] = 0$ for all $\gamma$ appearing in the linear combination, and hence that $[\Abf^a, \Obf] = 0$.
    
    First we consider the case that $\gamma_1 = \gamma_2 = 0$. In this case, referring to the definition in \cref{eq:def_R_lind}, we have $\Hbf_{\gamma_1} = \Hbf_{\gamma_2} = \Ibf$ and $\bar{b}_{a\gamma_1} = \bar{b}_{a\gamma_2}=1$, implying that $\CR_{\gamma_2\gamma_1}^{a\dag}(\Obf) = \Abf^{a \dag} \Obf \Abf^{a \dag} - \Obf$. If $\Abf^a$ is not supported on $Z$, then $[\Abf^a, \Obf] = 0$ and hence $\CR_{\gamma_2\gamma_1}^{a\dag}(\Obf)$ vanishes, verifying item (iii). If $\Abf^a$ is supported on $Z$, then $\Abf^{a \dag} \Obf \Abf^{a \dag} - \Obf$ is supported on $Z$. Either way, $\CR_{\gamma_2\gamma_1}^{a\dag}(\Obf)$ is supported on $Z = Z \cup S(\gamma_1) \cup S(\gamma_2)$, verifying item (i). Furthermore, since $\norm{\Abf^{a}}=1$ (\cref{cond:commute}), we have $\norm{\Abf^{a \dag} \Obf \Abf^{a \dag} - \Obf} \leq 2 \norm{\Obf} = 2\labs{\bar{b}_{a0}\bar{b}_{a0}} h_0 h_0 \norm{\Obf}$ by the triangle inequality, and the fact that the norm of a product is less than the product of norms, verifying item (ii). Item (iv) is not applicable in this case. 

    In all other cases, at least one of $\gamma_1,\gamma_2$ is not 0. Without loss of generality, assume $\gamma_1 \neq 0$. 
    Inspecting the definition of $\CR_{\gamma_1\gamma_2}^{a \dag}$ in \cref{eq:def_R_lind}, we see that $\CR_{\gamma_1\gamma_2}^{a \dag}(\Obf) = 0$ (and all four items hold), unless $\bar{b}_{a\gamma_1} \neq 0$ and $\bar{b}_{a\gamma_2} \neq 0$. Thus, we assume $\bar{b}_{a\gamma_1} \neq 0$, which implies that $\Abf^a$ anticommutes with $\Hbf_{\gamma_1}$ and thus by \cref{cond:commute}, $\Abf^a$ must be supported on $S(\gamma_1)$. Now, in general, the product of operators is supported on the union of the supports of the factors in the product. Since each term of \cref{eq:def_R_lind} is a product of $\Abf^{a}$, $\Hbf_{\gamma_1}$, $\Hbf_{\gamma_2}$ and $\Obf$ in some order, it follows that $\CR_{\gamma_1\gamma_2}^{a \dag}(\Obf)$ is supported on $Z \cup S(\gamma_1) \cup S(\gamma_2)$, proving item (i). To show item (ii), one can again apply the triangle inequality to each of the 6 terms of \cref{eq:def_R_lind} while noting that the norm of a product of operators is upper bounded by the product of the norms. We recall that $\norm{\Hbf_{\gamma}}\leq h_\gamma$ (\cref{cond:random_Hamiltonians}), and $\norm{\Abf^a} = 1$ (\cref{cond:commute}). Summing these term-wise bounds yields the quoted $2 \labs{\bar{b}_{a\gamma_1}\bar{b}_{a\gamma_2}}h_{\gamma_1}h_{\gamma_2}\norm{\Obf}$. 

    Next, suppose that $\gamma_2 = 0$, and that $\Abf^a$ is not supported on $Z$. This implies that $\Hbf_{\gamma_2} = \Ibf$, and that $\Abf^{a \dag} \Obf \Abf = \Obf$, so that 
            \begin{align}
        \CR_{\gamma_1\gamma_2}^{a \dag}(\Obf) =
     \frac{1}{2}\bar{b}_{a\gamma_1}\bar{b}_{a \gamma_2} \L(  \Hbf_{\gamma_1}  \Obf   + \Obf  \Hbf_{\gamma_1} - \frac{1}{2}\Hbf_{\gamma_1} \Obf - \frac{1}{2}\Hbf_{\gamma_1}\Obf  -\frac{1}{2} \Obf\Hbf_{\gamma_1}- \frac{1}{2} \Obf \Hbf_{\gamma_1} \R)
    \end{align}
    which vanishes, verifying item (iii). 

    Finally, suppose that $Z \cap (S(\gamma_1) \cup S(\gamma_2)) = \varnothing$.  If $\Abf^a$ is not supported on $S(\gamma_1)$, then $\bar{b}_{a\gamma_1}= 0$, and hence $\CR_{\gamma_2\gamma_1}^{a\dag}(\Obf) = 0$. If $\Abf^a$ (a single-site operator, by \cref{cond:aloc}) is supported on $S(\gamma_1)$, then it is not supported on $Z$ (since $Z \cap S(\gamma_1) = \varnothing$), and thus again $\Abf^{a\dag}\Obf \Abf^a = \Obf$. In this case, we have
    \begin{align}
        \CR_{\gamma_1\gamma_2}^{a \dag}(\Obf) =
     \frac{1}{2}\bar{b}_{a\gamma_1}\bar{b}_{a \gamma_2} \L(  \Hbf_{\gamma_1}  \Obf  \Hbf_{\gamma_2}  + \Hbf_{\gamma_2} \Obf  \Hbf_{\gamma_1} - \frac{1}{2}\Hbf_{\gamma_1} \Hbf_{\gamma_2}\Obf - \frac{1}{2}\Hbf_{\gamma_2} \Hbf_{\gamma_1}\Obf  -\frac{1}{2} \Obf\Hbf_{\gamma_1} \Hbf_{\gamma_2}- \frac{1}{2} \Obf\Hbf_{\gamma_2} \Hbf_{\gamma_1} \R)
    \end{align}
    If $\gamma_2 = 0$, then $\Hbf_{\gamma_2} = \Ibf$, and the expression vanishes and item (iv) is true. If $\gamma_2 \neq 0$, then the statement $Z \cap (S(\gamma_1) \cup S(\gamma_2)) = \varnothing$ means that the support of $\Obf$ is nonoverlapping with both $S(\gamma_1)$ and $S(\gamma_2)$. By the fact that $\Obf$ is a linear combination of products of $\Hbf_{\gamma}$ terms and the assumption of commutativity in \cref{cond:commute}, it follows that $\Obf$ commutes with each of $\Hbf_{\gamma_1}$ and $\Hbf_{\gamma_2}$. Employing this fact, we move the $\Obf$ to the left of every term, rewriting
    \begin{align}
        \CR_{\gamma_1\gamma_2}^{a \dag}(\Obf) =
     \frac{1}{2}\bar{b}_{a\gamma_1}\bar{b}_{a \gamma_2} \L(  \Obf\Hbf_{\gamma_1}   \Hbf_{\gamma_2}  + \Obf\Hbf_{\gamma_2}   \Hbf_{\gamma_1} - \frac{1}{2}\Obf \Hbf_{\gamma_1} \Hbf_{\gamma_2}- \frac{1}{2}\Obf\Hbf_{\gamma_2} \Hbf_{\gamma_1}  -\frac{1}{2} \Obf\Hbf_{\gamma_1} \Hbf_{\gamma_2}- \frac{1}{2} \Obf\Hbf_{\gamma_2} \Hbf_{\gamma_1} \R)
    \end{align}
    which vanishes, verifying item (iv). 
\end{proof}
Next, we use these facts to establish formal conditions under which the quantity $\CR_{\gamma_4 \gamma_3}^{a' \dag}\CR_{\gamma_2 \gamma_1}^{a \dag}(\Hbf_{\gamma_0})$ vanishes. Here we recall the definition from \cref{cond:aloc} that $A_Z$ is the subset of $A$ containing values of $a$ for which the jump $\Abf^a$ is supported on the subset $Z \subset [n]$. 
  \begin{lemma}\label{lem:CC_set}
        Assume that \cref{cond:random_Hamiltonians}, \cref{cond:commute}, and \cref{cond:aloc} hold. Define the subset $\CC \subset \Gamma \times \bar{\Gamma} \times A \times \bar{\Gamma} \times \bar{\Gamma} \times A \times \bar{\Gamma}$ as follows. A coordinate $\hat{\gamma} = (\gamma_0,\gamma_1,a,\gamma_2,\gamma_3,a',\gamma_4) \in \CC$ if all the following criteria are satisfied:
        \begin{align}
        \gamma_1 &\in \{\gamma_1 \in \Gamma: S(\gamma_1) \cap S(\gamma_0) \neq \varnothing \} \cup \{0\} \label{eq:gamma_1_criterion}\\
        a &\in A_{S(\gamma_0) \cup S(\gamma_1)} \text{ and } \bar{b}_{a\gamma_1}\neq 0 \label{eq:a_criterion}\\
        \gamma_2 &\in \{\gamma_2 \in \Gamma: b_{a\gamma_2} \neq 0\} \cup \{0\} \label{eq:gamma_2_criterion}\\
        \gamma_3 &\in \{\gamma_3 \in \Gamma: S(\gamma_3) \cap Z \neq \varnothing\} \cup \{0\}  \label{eq:gamma_3_criterion}  \\
        a' &\in A_{Z \cup S(\gamma_3)} \label{eq:a'_criterion}\\
    \gamma_4 &\in \{\gamma_4 \in \Gamma: b_{a'\gamma_4} \neq 0\} \cup \{0\},
    \label{eq:gamma_4_criterion}
\end{align}
where $Z \subset [n]$ denotes the support of the operator $\CR_{\gamma_2 \gamma_1}^{a \dag}(\Hbf_{\gamma_0})$ which is guaranteed to satisfy  $|Z| \leq 3k$. Let $\CC'$, $\CC''$, $\CC'''$ be the sets formed by taking points in $\CC$ and swapping coordinates $(\gamma_3,\gamma_4)$, swapping coordinates $(\gamma_1,\gamma_2)$, and swapping both pairs of coordinates, respectively.   If $\hat{\gamma} \not\in \CC \cup \CC' \cup \CC'' \cup \CC'''$, then $\CR_{\gamma_4 \gamma_3}^{a' \dag}\CR_{\gamma_2 \gamma_1}^{a \dag}(\Hbf_{\gamma_0}) = 0$ must hold. 
  \end{lemma}
  \begin{proof}
      First, note that by inpection of \cref{eq:def_R_lind}, and using the fact that $\Abf^a$ commutes or anticommutes with $\Hbf_{\gamma}$ for all $\gamma$ (\cref{cond:commute}), it follows that $\CR_{\gamma_2 \gamma_1}^{a \dag}(\Hbf_{\gamma_0})$ is a linear combination of products of Hamiltonian terms (in each term of \cref{eq:def_R_lind}, we can commute $\Abf^{a \dag}$ through to cancel as $\Abf^{a \dag} \Abf^a = \Ibf$). Thus, the criteria of \cref{lem:R_applied_to_O} are satisfied both for $\Obf = \Hbf_{\gamma_0}$ and for $\Obf=\CR_{\gamma_2 \gamma_1}^{a \dag}(\Hbf_{\gamma_0})$. The fact that $|Z| \leq 3k$ follows from the item (i) of \cref{lem:R_applied_to_O}, since $\Hbf_{\gamma_0}$, $\Hbf_{\gamma_1}$, and $\Hbf_{\gamma_2}$ are each supported on at most $k$ sites. 
      
      We iterate through the conditions, aiming to show that if the condition fails, the quantity vanishes.  For $\Obf = \Hbf_{\gamma_0}$ the support of $\Obf$ is $S(\gamma_0)$ and by \cref{lem:R_applied_to_O}, item (iv), we can assert the following. For at least one choice of $i \in \{1,2\}$ it must hold either that $\gamma_i = 0$, or that $\gamma_i \neq 0$ and $S(\gamma_0) \cap S(\gamma_i) \neq \varnothing$, in order for the quantity to be nonvanishing. Without loss of generality (since we symmetrize by also including $\CC'$, $\CC''$, and $\CC'''$), assume it holds for $i=1$, which implies the first condition, \cref{eq:gamma_1_criterion}. 

      Next, if $\bar{b}_{a\gamma_1} = 0$, then clearly the term vanishes, by \cref{eq:def_R_lind}. Moreover, if $\Abf^a$ is not supported on $S(\gamma_1)$ and $\gamma_1 \neq 0$, then this implies that $\bar{b}_{a\gamma_1}=0$ and thus by \cref{eq:def_R_lind}, the quantity vanishes. If $\Abf^a$ is not supported on $S(\gamma_1)$ and $\gamma_1 = 0$, then by item (iii) of \cref{lem:R_applied_to_O}, $\Abf^a$ must be supported on $S(\gamma_0)$, else the quantity vanishes. We conclude that for the quantity to be nonvanishing, $\Abf^a$ must either be supported on $S(\gamma_0)$ or on $S(\gamma_1)$. This verifies the second condition, \cref{eq:a_criterion}. 

      The third condition, \cref{eq:gamma_2_criterion}, is true since the quantity vanishes whenever $\bar{b}_{a\gamma_2} = 0$. Thus, for it to be nonvanishing, either $\gamma_2 = 0$ or $\gamma_2 \neq 0$ and $b_{a\gamma_2} \neq 0$. 

      The fourth, fifth, and sixth conditions are true for the same reason the first, second, and third conditions are true, respectively, instead using $\Obf=\CR_{\gamma_2 \gamma_1}^{a \dag}(\Hbf_{\gamma_0})$, and with $a'$ in place of $a$, and $\gamma_3,\gamma_4$ in place of $\gamma_1$, $\gamma_2$.   
  \end{proof}
  \begin{lemma}\label{lem:bound_on_RRH}
  Assume \cref{cond:random_Hamiltonians}, \cref{cond:commute}, and \cref{cond:aloc} hold. Then,
      \begin{align}
          \norm{\CR_{\gamma_4 \gamma_3}^{a' \dag}\CR_{\gamma_2 \gamma_1}^{a \dag}(\Hbf_{\gamma_0})} \leq 4\labs{\bar{b}_{a\gamma_1}\bar{b}_{a\gamma_2}\bar{b}_{a'\gamma_3}\bar{b}_{a'\gamma_4}}h_{\gamma_0}h_{\gamma_1}h_{\gamma_2}h_{\gamma_3}h_{\gamma_4}\,.
      \end{align}
  \end{lemma}
  \begin{proof}
      This follows from \cref{lem:R_applied_to_O}, item (ii), and the fact that $\norm{\Hbf_{\gamma}} = h_{\gamma}$ (\cref{cond:random_Hamiltonians}). 
  \end{proof}
 
In what follows, the quantity $Z \subset [n]$ denotes the support of $\CR_{\gamma_2 \gamma_1}^{a \dag}(\Hbf_{\gamma_0})$. By \cref{lem:R_applied_to_O}, item (i) and the fact that $|S(\gamma)| = k$ for all $\gamma$ (\cref{cond:aloc}), we have $|Z|\leq 3k$. 

We build toward the full sum in \cref{eq:sum_at_beginning} by working right to left within the tuple $(\gamma_0,\gamma_1,a,\gamma_2,\gamma_3,a',\gamma_4)$. That is, we consider fixed $\gamma_0,\gamma_1,a,\gamma_2,\gamma_3,a'$ and upper bound the sum over $\gamma_4$. Then, we unfix $a'$ and upper bound the sum over $a'$ and $\gamma_4$. Then we unfix $\gamma_3$ and upper bound the sum over $\gamma_3$, $a'$, and $\gamma_4$, etc. 

\paragraph{Summing over $\gamma_4$.} 
First, for fixed $\gamma_0,\gamma_1,a,\gamma_2,\gamma_3,a'$, from \cref{lem:bound_on_RRH} we have that
\begin{align}
    &\sum_{\gamma_4: \hat{\gamma} \in \CC}  h(\gamma_0, \gamma_1,\gamma_2,\gamma_3,\gamma_4) \norm{\CR_{\gamma_4 \gamma_3}^{a' \dag}\CR_{\gamma_2 \gamma_1}^{a \dag}(\Hbf_{\gamma_0})} \\
    \leq{}& 4\labs{\bar{b}_{a\gamma_1}\bar{b}_{a\gamma_2}\bar{b}_{a'\gamma_3}}h_{\gamma_0}h_{\gamma_1}h_{\gamma_2}h_{\gamma_3}\sum_{\gamma_4: \hat{\gamma} \in \CC} |\bar{b}_{a'\gamma_4}| h_{\gamma_4}h(\gamma_0,\gamma_1,\gamma_2,\gamma_3,\gamma_4). 
\end{align}
Recall that the value of $h(\gamma_0,\gamma_1,\gamma_2,\gamma_3,\gamma_4)$, defined in \cref{eq:def_h(gammas)}, depends on whether the settings of $\gamma_i$ collide or not. Specifically, we can update $h(\gamma_0,\gamma_1,\gamma_2,\gamma_3,\gamma_4)$ based on the setting of $\gamma_4$, and whether it collides with one of the other $\gamma_i$; for example, if $\gamma_4 = \gamma_0$, then $h(\gamma_0,\gamma_1,\gamma_2,\gamma_3,\gamma_4) = h(\gamma_1,\gamma_2,\gamma_3)$, and if there is no collision, then $h(\gamma_0,\gamma_1,\gamma_2,\gamma_3,\gamma_4) = |y|h_{\gamma_4}h(\gamma_0,\gamma_1,\gamma_2,\gamma_3)$.
Following these assertions, we may observe that 
\begin{align}\label{eq:gamma4_h_equalities}
    |\bar{b}_{a'\gamma_4}|h_{\gamma_4} h(\gamma_0,\gamma_1,\gamma_2,\gamma_3,\gamma_4) = 
    \begin{cases}
          h(\gamma_0, \gamma_1,\gamma_2,\gamma_3) & \text{if } \gamma_4 = 0 \\
          h(\gamma_1,\gamma_2,\gamma_3)|\bar{b}_{a'\gamma_0}|h_{\gamma_0} & \text{if } \gamma_4 = \gamma_0 \neq 0 \\
         h(\gamma_0,\gamma_2,\gamma_3)|\bar{b}_{a'\gamma_1}|h_{\gamma_1}  & \text{if } \gamma_4 = \gamma_1 \neq 0 \\
         h(\gamma_0,\gamma_1,\gamma_3)|\bar{b}_{a'\gamma_2}|h_{\gamma_2}  & \text{if } \gamma_4 = \gamma_2 \neq 0 \\
         h(\gamma_0,\gamma_1,\gamma_2)|\bar{b}_{a'\gamma_3}|h_{\gamma_3}  & \text{if } \gamma_4 = \gamma_3 \neq 0 \\
         h(\gamma_0,\gamma_1,\gamma_2, \gamma_3)|\bar{b}_{a'\gamma_4}||y|h_{\gamma_4}^2  & \text{otherwise}
    \end{cases}
\end{align}

Note that in the cases where $\gamma_4 \neq 0$, we can assert that $|\bar{b}_{a'\gamma_i}| \leq 2|y|$ (since the first case in \cref{eq:bar_b_agamma_def} is ruled out). In the final case, in order for $\bar{b}_{a'\gamma_4} \neq 0$, it must hold that $S(\gamma_4)$ intersects with the support of $\Abf^{a'}$ (\cref{cond:commute}), which is a single site (\cref{cond:aloc}), denoted by $\{j'\}$. We recall from the definition of $\Hloc$ in \cref{eq:Hloc_Hglo_def}, that 
\begin{equation}
    \sum_{\gamma_4: j' \in S(\gamma_4)}h_{\gamma_4}^2 \leq \HlocPower{2}.
\end{equation}
Thus, putting it all together, we have the following, where we have labeled which terms arise from each of the cases above.
\begin{align}
    &\sum_{\gamma_4: \hat{\gamma} \in \CC}  h(\gamma_0, \gamma_1,\gamma_2,\gamma_3,\gamma_4) \norm{\CR_{\gamma_4 \gamma_3}^{a' \dag}\CR_{\gamma_2 \gamma_1}^{a \dag}(\Hbf_{\gamma_0})} \\
    \leq{}& \CO(1) \cdot \labs{\bar{b}_{a\gamma_1}\bar{b}_{a\gamma_2}\bar{b}_{a'\gamma_3}}h_{\gamma_0}h_{\gamma_1}h_{\gamma_2}h_{\gamma_3}\Bigg(h(\gamma_0, \gamma_1,\gamma_2,\gamma_3)(\underbrace{1}_{{\sss \gamma_4=0}}+\underbrace{|y|^2 \HlocPower{2}}_{\text{\tiny otherwise}}) \nonumber \\
    &\quad
    + \underbrace{h(\gamma_1,\gamma_2,\gamma_3) |y|h_{\gamma_0}\bar{\delta}_{0\gamma_0}}_{\sss \gamma_4=\gamma_0}
    + \underbrace{h(\gamma_0,\gamma_2,\gamma_3) |y|h_{\gamma_1}\bar{\delta}_{0\gamma_1}}_{\sss \gamma_4=\gamma_1} 
    + \underbrace{h(\gamma_0,\gamma_1,\gamma_3) |y|h_{\gamma_2}\bar{\delta}_{0\gamma_2}}_{\sss \gamma_4=\gamma_2} 
    + \underbrace{h(\gamma_0,\gamma_1,\gamma_2) |y|h_{\gamma_3}\bar{\delta}_{0\gamma_3}}_{\sss \gamma_4=\gamma_3} \Bigg)\label{eq:sum_after_gamma4}
\end{align}
where $\bar{\delta}_{\gamma0} := 1-\delta_{\gamma0}$ is 1 if $\gamma \neq 0$ and 0 otherwise. We have absorbed the factor of 2 on some of the terms into the $\CO(1)$ in front, to keep the expressions as simple as possible (recall that all terms on the right-hand side are non-negative).

\paragraph{Summing over $a'$.} Next, for fixed $\gamma_0,\gamma_1,a, \gamma_2,\gamma_3$, we consider summing the expression in \cref{eq:sum_after_gamma4} over the choice of $a'$, consistent with $\hat{\gamma} \in \CC$. From \cref{lem:CC_set}, \cref{eq:a'_criterion}, it must hold that $a' \in A_{Z \cup S(\gamma_3)}$, which is true for at most $4\aloc k$ values of $a'$ (\cref{cond:aloc}). Thus, we pick up a factor of $\aloc k$ in the expression, and we can write
\begin{align}
    &\sum_{a',\gamma_4: \hat{\gamma} \in \CC}  h(\gamma_0, \gamma_1,\gamma_2,\gamma_3,\gamma_4) \norm{\CR_{\gamma_4 \gamma_3}^{a' \dag}\CR_{\gamma_2 \gamma_1}^{a \dag}(\Hbf_{\gamma_0})} \\
    \leq{}& \CO(1) \cdot (\aloc k ) \labs{\bar{b}_{a\gamma_1}\bar{b}_{a\gamma_2}}(2|y|)^{\bar{\delta}_{0\gamma_3}}h_{\gamma_0}h_{\gamma_1}h_{\gamma_2}h_{\gamma_3}\Bigg(h(\gamma_0, \gamma_1,\gamma_2,\gamma_3)(1+ |y|^2 \HlocPower{2})  \nonumber \\  
    &\quad + h(\gamma_1,\gamma_2,\gamma_3) |y|h_{\gamma_0} \bar{\delta}_{0\gamma_0} 
    + h(\gamma_0,\gamma_2,\gamma_3) |y|h_{\gamma_1}\bar{\delta}_{0\gamma_1}
    + h(\gamma_0,\gamma_1,\gamma_3) |y|h_{\gamma_2}\bar{\delta}_{0\gamma_2}
    + h(\gamma_0,\gamma_1,\gamma_2) |y|h_{\gamma_3}\bar{\delta}_{0\gamma_3} \Bigg)\label{eq:sum_after_a'}
\end{align}
where we have also used the upper bound $|\bar{b}_{a'\gamma_3}| \leq (2|y|)^{\bar{\delta}_{0\gamma_3}}$.

\paragraph{Summing over $\gamma_3$. } Next, for fixed $\gamma_0, \gamma_1, a, \gamma_2$, we consider summing the expression in \cref{eq:sum_after_a'} over $\gamma_3$, consistent with $\hat{\gamma} \in \CC$. Following a similar logic as previously, we observe that
\begin{align}\label{eq:gamma3_h_equalities_1}
    &(2|y|)^{\bar{\delta}_{0\gamma_3}}  h(\gamma_0,\gamma_1,\gamma_2,\gamma_3)h_{\gamma_3}  =
    \begin{cases}
        h(\gamma_0,\gamma_1,\gamma_2) & \text{if } \gamma_3 = 0 \\
        h(\gamma_1,\gamma_2)2|y| h_{\gamma_0} & \text{if } \gamma_3 = \gamma_0 \neq 0 \\
        h(\gamma_0,\gamma_2)2|y| h_{\gamma_1} & \text{if } \gamma_3 = \gamma_1 \neq 0 \\
        h(\gamma_0,\gamma_1)2|y| h_{\gamma_2} & \text{if } \gamma_3 = \gamma_1 \neq 0 \\
        h(\gamma_0,\gamma_1,\gamma_2)2|y|^2 h_{\gamma_3}^2 & \text{otherwise}
    \end{cases}
\end{align}
and that for $i \neq j \in \{0,1,2\}$
\begin{align}\label{eq:gamma3_h_equalities_2}
    (2|y|)^{\bar{\delta}_{0\gamma_3}} |y| h_{\gamma_3} h(\gamma_i,\gamma_j,\gamma_3) = 
    \begin{cases}
        h(\gamma_i,\gamma_j) |y| & \text{if } \gamma_3 = 0 \\
        h(\gamma_j)2|y|^2 h_{\gamma_i} & \text{if } \gamma_3 = \gamma_i \neq 0 \\
        h(\gamma_i)2|y|^2 h_{\gamma_j} & \text{if } \gamma_3 = \gamma_j \neq 0 \\
        h(\gamma_i,\gamma_j)2|y|^3 h_{\gamma_3}^2 & \text{otherwise}
    \end{cases}
\end{align}
Note that $h(\alpha) = 1$ if $\alpha = 0$, and $h(\alpha) = |y|h_{\alpha}$  if $\alpha \neq 0$.  
In the ``otherwise'' cases above, we also note from \cref{lem:CC_set}, \cref{eq:gamma_3_criterion}, that $S(\gamma_3)$ must have overlap with $Z$, a set of size at most $3k$, in order for $\hat{\gamma} \in \CC$ to hold. Thus, 
\begin{align}
    \sum_{\gamma_3: S(\gamma_3) \cap Z \neq \varnothing} h_{\gamma_3}^2 \leq \sum_{j \in Z} \sum_{\gamma_3: j \in S(\gamma_3)} h_{\gamma_3}^2 \leq  3k \HlocPower{2},
\end{align}
where we have again invoked the definition of $\Hloc$ in \cref{eq:Hloc_Hglo_def}. 
This inequality is also useful for summing over the $(2|y|)^{\bar{\delta}_{0\gamma_3}}|y|h_{\gamma_3}^2h(\gamma_0,\gamma_1,\gamma_2)\bar{\delta}_{0\gamma_3}$ term (denoted by the label ``$\gamma_4=\gamma_3$'' in the underbrace below). 
Putting these together, we arrive at the following, where we label each term to help understand which values of $\gamma_3$ and $\gamma_4$ it corresponds to:
\begin{align}
    &\sum_{\gamma_3, a',\gamma_4: \hat{\gamma} \in \CC}  h(\gamma_0, \gamma_1,\gamma_2,\gamma_3,\gamma_4) \norm{\CR_{\gamma_4 \gamma_3}^{a' \dag}\CR_{\gamma_2 \gamma_1}^{a \dag}(\Hbf_{\gamma_0})} \nonumber\\
    \leq{}& \CO(1) \cdot (\aloc k) \labs{\bar{b}_{a\gamma_1}\bar{b}_{a\gamma_2}}h_{\gamma_0}h_{\gamma_1}h_{\gamma_2}
    \Bigg(h(\gamma_0,\gamma_1,\gamma_2)\bigg(\big(\underbrace{1}_{{\sss \gamma_4=0}}+\underbrace{|y|^2 \HlocPower{2}}_{\text{\tiny \eqref{eq:gamma4_h_equalities}, otherwise}}\big)\big(\underbrace{1}_{ \substack{\sss \eqref{eq:gamma3_h_equalities_1} \\ \sss \gamma_3 = 0}}+  \underbrace{k|y|^2 \HlocPower{2}}_{\text{\tiny  \eqref{eq:gamma3_h_equalities_1}, otherwise}}\big) +  \underbrace{k|y|^2\HlocPower{2}}_{\sss \gamma_4 = \gamma_3} \bigg)  \nonumber\\
    &\quad +\bigg(\underbrace{1}_{{\sss \gamma_4=0}}+\underbrace{|y|^2 \HlocPower{2}}_{\text{\tiny \eqref{eq:gamma4_h_equalities}, otherwise}}\bigg) \bigg(\underbrace{h(\gamma_1,\gamma_2)|y|h_{\gamma_0}\bar{\delta}_{0\gamma_0}}_{\sss \eqref{eq:gamma3_h_equalities_1}, \gamma_3=\gamma_0}+ \underbrace{h(\gamma_0,\gamma_2)|y|h_{\gamma_1}\bar{\delta}_{0\gamma_1}}_{\sss \eqref{eq:gamma3_h_equalities_1},  \gamma_3=\gamma_1} + \underbrace{h(\gamma_0,\gamma_1)|y|h_{\gamma_2} \bar{\delta}_{0\gamma_2}}_{\sss \eqref{eq:gamma3_h_equalities_1}, \gamma_3=\gamma_2}\bigg) \nonumber
    \\
    &\quad +\bigg(\underbrace{1}_{\substack{\sss \eqref{eq:gamma3_h_equalities_2} \\ \sss \gamma_3 = 0 }}+ \underbrace{k|y|^2\HlocPower{2}}_{ \text{\tiny \eqref{eq:gamma3_h_equalities_2}, otherwise}}\bigg) \bigg( \underbrace{h(\gamma_1,\gamma_2)|y|h_{\gamma_0}\bar{\delta}_{0\gamma_0}}_{\sss \gamma_4 = \gamma_0}+ \underbrace{h(\gamma_0,\gamma_2)|y|h_{\gamma_1}\bar{\delta}_{0\gamma_1}}_{\sss \gamma_4 = \gamma_1} + \underbrace{h(\gamma_0,\gamma_1)|y|h_{\gamma_2} \bar{\delta}_{0\gamma_2}}_{\sss \gamma_4 = \gamma_2}\bigg) \nonumber
    \\
    &\quad +  \underbrace{h(\gamma_2)|y|^2 h_{\gamma_0}h_{\gamma_1} \bar{\delta}_{0\gamma_0}\bar{\delta}_{0\gamma_1}}_{\sss \eqref{eq:gamma3_h_equalities_2} \substack{ \sss \gamma_4 = \gamma_0 \text{ and } \gamma_3 = \gamma_1 \\ \sss \gamma_4 = \gamma_1 \text{ and } \gamma_3 = \gamma_0}}
    + \underbrace{h(\gamma_1) |y|^2 h_{\gamma_0}h_{\gamma_2} \bar{\delta}_{0\gamma_0}\bar{\delta}_{0\gamma_2}}_{\sss \eqref{eq:gamma3_h_equalities_2} \substack{ \sss \gamma_4 = \gamma_0 \text{ and } \gamma_3 = \gamma_2 \\ \sss \gamma_4 = \gamma_2 \text{ and } \gamma_3 = \gamma_0}} 
    + \underbrace{h(\gamma_0)|y|^2 h_{\gamma_1}h_{\gamma_2}\bar{\delta}_{0\gamma_1}\bar{\delta}_{0\gamma_2}}_{\sss \eqref{eq:gamma3_h_equalities_2} \substack{ \sss \gamma_4 = \gamma_1 \text{ and } \gamma_3 = \gamma_2 \\ \sss \gamma_4 = \gamma_2 \text{ and } \gamma_3 = \gamma_1}}\Bigg).\label{eq:sum_after_gamma3}
    \end{align}
    Grouping like terms, the right-hand side can be rewritten as
    \begin{align}
    & \CO(1) \cdot (\aloc k) \labs{\bar{b}_{a\gamma_1}\bar{b}_{a\gamma_2}}h_{\gamma_0}h_{\gamma_1}h_{\gamma_2}
    \Bigg(h(\gamma_0,\gamma_1,\gamma_2)\bigg(1+  k|y|^2 \HlocPower{2} + k |y|^4 \HlocPower{4}\bigg)  \nonumber\\
    &\quad +\bigg(h(\gamma_1,\gamma_2)|y|h_{\gamma_0}\bar{\delta}_{0\gamma_0}+ h(\gamma_0,\gamma_2)|y|h_{\gamma_1}\bar{\delta}_{0\gamma_1} + h(\gamma_0,\gamma_1)|y|h_{\gamma_2} \bar{\delta}_{0\gamma_2}\bigg)\bigg(1+k|y|^2\HlocPower{2}\bigg)  \nonumber
    \\
    &\quad +  h(\gamma_2)|y|^2 h_{\gamma_0}h_{\gamma_1} \bar{\delta}_{0\gamma_0}\bar{\delta}_{0\gamma_1}
    + h(\gamma_1) |y|^2 h_{\gamma_0}h_{\gamma_2} \bar{\delta}_{0\gamma_0}\bar{\delta}_{0\gamma_2} 
    + h(\gamma_0)|y|^2 h_{\gamma_1}h_{\gamma_2}\bar{\delta}_{0\gamma_1}\bar{\delta}_{0\gamma_2}\Bigg).\label{eq:sum_after_gamma3_grouped}
\end{align}

\paragraph{Summing over $\gamma_2$. } Next, for fixed $\gamma_0, \gamma_1, a$, we consider summing the expression in \cref{eq:sum_after_gamma3_grouped} over $\gamma_2$, consistent with $\hat{\gamma} \in \CC$. The logic here is similar to that of summing over $\gamma_4$, above. First, we observe that
\begin{align}\label{eq:gamma2_h_equalities1}
    |\bar{b}_{a\gamma_2}|h_{\gamma_2} h(\gamma_0,\gamma_1,\gamma_2) = 
    \begin{cases}
          h(\gamma_0, \gamma_1) & \text{if } \gamma_2 = 0 \\
          h(\gamma_1)|\bar{b}_{a\gamma_0}|h_{\gamma_0} & \text{if } \gamma_2 = \gamma_0 \neq 0 \\
         h(\gamma_0)|\bar{b}_{a\gamma_1}|h_{\gamma_1}  & \text{if } \gamma_2 = \gamma_1 \neq 0 \\
         h(\gamma_0,\gamma_1)|\bar{b}_{a\gamma_2}y|h_{\gamma_2}^2  & \text{otherwise}
    \end{cases}
\end{align}
that, for $i \in \{1,2\}$,
\begin{align}\label{eq:gamma2_h_equalities2}
    |\bar{b}_{a\gamma_2}y|h_{\gamma_2} h(\gamma_i,\gamma_2) = 
    \begin{cases}
          h(\gamma_i)|y| & \text{if } \gamma_2 = 0 \\
          |\bar{b}_{a\gamma_i}y|h_{\gamma_i} & \text{if } \gamma_2 = \gamma_i \neq 0 \\
         h(\gamma_i)|\bar{b}_{a\gamma_2}y^2|h_{\gamma_2}^2  & \text{otherwise}
    \end{cases}
\end{align}
and that
\begin{align}\label{eq:gamma2_h_equalities3}
    \bar{b}_{a\gamma_2} |y|^2 h_{\gamma_2}h(\gamma_2)    =
    \begin{cases}
          |y|^2 & \text{if } \gamma_2 = 0 \\
         |\bar{b}_{a\gamma_2}y^3| h_{\gamma_2}^2  & \text{otherwise.}
    \end{cases}
\end{align}
Note that in the cases where $\gamma_2 =\gamma_i \neq 0$, we can assert that $|\bar{b}_{a\gamma_i}| \leq 2|y|$.  
In the ``otherwise'' cases, in order for $\bar{b}_{a\gamma_2} \neq 0$, then $S(\gamma_2)$ must intersect with the support of $\Abf^{a}$ (\cref{cond:commute}), which is a single site (\cref{cond:aloc}), denoted by $\{j\}$. We recall that 
\begin{equation}
    \sum_{\gamma_2: j \in S(\gamma_2)}h_{\gamma_2}^2 \leq \HlocPower{2}.
\end{equation}
This inequality is also useful for summing the $|\bar{b}_{a\gamma_2}y^2|h(\gamma_i)h_{\gamma_j}h_{\gamma_2}^2\bar{\delta}_{0\gamma_j}\bar{\delta}_{0\gamma_2}$ terms (labeled by $*$ below) and the $|\bar{b}_{a\gamma_2}| h(\gamma_0,\gamma_1)|y| h_{\gamma_2}^2\bar{\delta}_{0\gamma_2}$ term (labeled by $**$ below) from \cref{eq:sum_after_gamma3_grouped}.
Thus, we have
\begin{align}
    &\sum_{\gamma_2,\gamma_3, a',\gamma_4: \hat{\gamma} \in \CC}  h(\gamma_0, \gamma_1,\gamma_2,\gamma_3,\gamma_4) \norm{\CR_{\gamma_4 \gamma_3}^{a' \dag}\CR_{\gamma_2 \gamma_1}^{a \dag}(\Hbf_{\gamma_0})} \nonumber\\
    \leq{}& \CO(1) \cdot (\aloc k) |\bar{b}_{a\gamma_1}|h_{\gamma_0}h_{\gamma_1}
    \Bigg(h(\gamma_0,\gamma_1)\bigg(1+  k|y|^2 \HlocPower{2} + k|y|^4 \HlocPower{4}\bigg)\bigg(\underbrace{1}_{\substack{\sss \eqref{eq:gamma2_h_equalities1} \\ \sss \gamma_2 = 0}}+\underbrace{|y|^2 \HlocPower{2}}_{\text{\tiny \eqref{eq:gamma2_h_equalities1}, otherwise}}\bigg)  \nonumber\\
    &\quad
    +  \underbrace{h_{\gamma_0} h_{\gamma_1}|y|^2 \bar{\delta}_{0\gamma_0}\bar{\delta}_{0\gamma_1}}_{\sss \eqref{eq:gamma2_h_equalities2}, \gamma_2 = \gamma_i}\bigg(1 + k |y|^2 \HlocPower{2} \bigg) + |y|\bigg(\underbrace{h(\gamma_1)h_{\gamma_0}\bar{\delta}_{0\gamma_0}}_{\sss \eqref{eq:gamma2_h_equalities1}, \gamma_2=\gamma_0}+\underbrace{h(\gamma_0)h_{\gamma_1}\bar{\delta}_{0\gamma_1}}_{\sss \eqref{eq:gamma2_h_equalities1}, \gamma_2=\gamma_1}\bigg)\bigg(1+k|y|^2 \HlocPower{2} + k|y|^4\HlocPower{4} \bigg) \nonumber \\
    & \quad +  \bigg(\underbrace{h(\gamma_0)|y| h_{\gamma_1}\bar{\delta}_{0\gamma_1} + h(\gamma_1)|y|h_{\gamma_0} \bar{\delta}_{0\gamma_0}}_{\sss \eqref{eq:gamma2_h_equalities2}, \gamma_2 = 0} + \underbrace{h(\gamma_0)|y|^3 h_{\gamma_1}\bar{\delta}_{0\gamma_1} \HlocPower{2} + h(\gamma_1)|y|^3 h_{\gamma_0} \bar{\delta}_{0\gamma_0} \HlocPower{2}}_{\text{\tiny \eqref{eq:gamma2_h_equalities2}, otherwise}} \bigg)\bigg(1 + k |y|^2 \HlocPower{2} \bigg)  \nonumber \\
    & \quad + \underbrace{|y|^2 h_{\gamma_0} h_{\gamma_1} \bar{\delta}_{0\gamma_0} \bar{\delta}_{0\gamma_1}}_{\sss \eqref{eq:gamma2_h_equalities3}, \gamma_2 = 0} +\underbrace{|y|^4 h_{\gamma_0} h_{\gamma_1} \bar{\delta}_{0\gamma_0} \bar{\delta}_{0\gamma_1} \HlocPower{2}}_{\text{\tiny \eqref{eq:gamma2_h_equalities3}, otherwise}} + \underbrace{h(\gamma_1)|y|^3 h_{\gamma_0}\bar{\delta}_{0\gamma_0} \HlocPower{2} + h(\gamma_0)|y|^3 h_{\gamma_1}\bar{\delta}_{0\gamma_1} \HlocPower{2}}_{*} \nonumber \\
    & \quad +\underbrace{h(\gamma_0,\gamma_1) |y|^2 \HlocPower{2} \bigg(1 + k |y|^2 \HlocPower{2}}_{**}\bigg)
    \Bigg).
    \end{align}
    Grouping like terms, the right-hand side can be rewritten as
    \begin{align}
    &\CO(1) \cdot (\aloc k) |\bar{b}_{a\gamma_1}|h_{\gamma_0}h_{\gamma_1}
    \Bigg(h(\gamma_0,\gamma_1)\bigg(1+  k|y|^2 \HlocPower{2} + k|y|^4 \HlocPower{4} + k|y|^6 \HlocPower{6}\bigg)  \nonumber\\
    &\quad
    +  h_{\gamma_0} h_{\gamma_1}\bar{\delta}_{0\gamma_0}\bar{\delta}_{0\gamma_1}\bigg(|y|^2 + k |y|^4 \HlocPower{2} \bigg) + \bigg(h(\gamma_1)h_{\gamma_0}\bar{\delta}_{0\gamma_0}+h(\gamma_0)h_{\gamma_1}\bar{\delta}_{0\gamma_1}\bigg)\bigg(|y|+k|y|^3 \HlocPower{2} + k|y|^5\HlocPower{4} \bigg) \Bigg).\label{eq:sum_after_gamma2}
\end{align}

\paragraph{Summing over $a$.} Next, for fixed $\gamma_0,\gamma_1$, we consider summing the expression in \cref{eq:sum_after_gamma2} over choice of $a$, consistent with $\hat{\gamma} \in \CC$. We note from \cref{lem:CC_set}, \cref{eq:a_criterion},  it must hold that $a \in A_{S(\gamma_0) \cup S(\gamma_1)}$,  which is true for at most $2\aloc k$ choices of $a$ (\cref{cond:aloc}). Thus, we pick up another factor of $\aloc k$ in the expression, and we can write
\begin{align}
    &\sum_{a,\gamma_2,\gamma_3, a',\gamma_4: \hat{\gamma} \in \CC}  h(\gamma_0, \gamma_1,\gamma_2,\gamma_3,\gamma_4) \norm{\CR_{\gamma_4 \gamma_3}^{a' \dag}\CR_{\gamma_2 \gamma_1}^{a \dag}(\Hbf_{\gamma_0})} \nonumber\\
    \leq{}& \CO(1) \cdot (\alocPower{2} k^2) (2|y|)^{\bar{\delta}_{0\gamma_1}}h_{\gamma_0}h_{\gamma_1}
    \Bigg(h(\gamma_0,\gamma_1)\bigg(1+  k|y|^2 \HlocPower{2} + k|y|^4 \HlocPower{4} + k|y|^6 \HlocPower{6}\bigg)  \nonumber\\
    &\quad
    +  h_{\gamma_0} h_{\gamma_1}\bar{\delta}_{0\gamma_1}\bigg(|y|^2 + k |y|^4 \HlocPower{2} \bigg) + \bigg(h(\gamma_1)h_{\gamma_0}+|y|h_{\gamma_0}h_{\gamma_1}\bar{\delta}_{0\gamma_1}\bigg)\bigg(|y|+k|y|^3 \HlocPower{2} + k|y|^5\HlocPower{4} \bigg) \Bigg),\label{eq:sum_after_a}
\end{align}
where we have also made the upper bound $|\bar{b}_{a\gamma_1}| \leq (2|y|)^{\bar{\delta}_{0\gamma_1}}$, and we have substituted $h(\gamma_0) = |y|h_{\gamma_0}$ and $\bar{\delta}_{0\delta_0} = 1$, since $\gamma_0 \in \Gamma$ and thus $\gamma_0 \neq 0$ will always hold. 

\paragraph{Summing over $\gamma_1$. } Next, for fixed $\gamma_0$, we consider summing the expression in \cref{eq:sum_after_a} over $\gamma_1$, consistent with $\hat{\gamma} \in \CC$. We observe that (noting that $\gamma_0 \neq 0$)
\begin{align}\label{eq:gamma1_h_equalities1}
    (2|y|)^{\bar{\delta}_{0\gamma_1}}h_{\gamma_1}h(\gamma_0,\gamma_1) =
    \begin{cases}
        |y|h_{\gamma_0}&\text{if } \gamma_1 = 0 \\
        2|y|h_{\gamma_0}&\text{if } \gamma_1 = \gamma_0 \neq 0 \\
        2|y|^3h_{\gamma_0}h_{\gamma_1}^2 & \text{otherwise}
    \end{cases}
\end{align}
and
\begin{align}\label{eq:gamma1_h_equalities2}
    (2|y|)^{\bar{\delta}_{0\gamma_1}}h_{\gamma_1}h(\gamma_1) =
    \begin{cases}
        1&\text{if } \gamma_1 = 0 \\
        2|y|^2h_{\gamma_1}^2 & \text{otherwise}
    \end{cases}
\end{align}
In the ``otherwise'' cases, we also observe from \cref{lem:CC_set}, \cref{eq:gamma_1_criterion}, that $S(\gamma_1)$ must intersect $S(\gamma_0)$, a set of size $k$.  We have
\begin{align}
    \sum_{\gamma_1: S(\gamma_1) \cap S(\gamma_0) \neq \varnothing} h_{\gamma_1}^2 \leq \sum_{j \in S(\gamma_0)} \sum_{\gamma_1: j \in S(\gamma_1)} h_{\gamma_1}^2 \leq k \HlocPower{2}\,.
\end{align}
This inequality is also helpful for bounding the $(2|y|)^{\bar{\delta}_{0\gamma_1}}h_{\gamma_1}^2 \bar{\delta}_{0\gamma_1}$ term in \cref{eq:sum_after_a}, labeled with $***$ below. 
Thus, we have
\begin{align}
    &\sum_{\gamma_1,a,\gamma_2,\gamma_3, a',\gamma_4: \hat{\gamma} \in \CC}  h(\gamma_0, \gamma_1,\gamma_2,\gamma_3,\gamma_4) \norm{\CR_{\gamma_4 \gamma_3}^{a' \dag}\CR_{\gamma_2 \gamma_1}^{a \dag}(\Hbf_{\gamma_0})} \nonumber\\
    \leq{}& \CO(1) \cdot \alocPower{2} k^2 h_{\gamma_0}
    \Bigg(\bigg(\underbrace{|y| h_{\gamma_0}}_{\sss \eqref{eq:gamma1_h_equalities1}, \gamma_1 = 0} + \underbrace{|y| h_{\gamma_0}}_{\sss \eqref{eq:gamma1_h_equalities1}, \gamma_1 = \gamma_0} + \underbrace{|y|^3 h_{\gamma_0} k \HlocPower{2}}_{\text{\tiny \eqref{eq:gamma1_h_equalities1}, otherwise}}\bigg)\bigg(1+  k|y|^2 \HlocPower{2} +k|y|^4 \HlocPower{4} + k|y|^6 \HlocPower{6}\bigg)  \\
    & \quad + \bigg(\underbrace{h_{\gamma_0}}_{\sss \eqref{eq:gamma1_h_equalities2}, \gamma_1 = 0} + \underbrace{h_{\gamma_0}|y|^2 k \HlocPower{2}}_{\text{\tiny \eqref{eq:gamma1_h_equalities2}, otherwise}} \bigg)\bigg(|y| + k|y|^3\HlocPower{2} + k|y|^5 \HlocPower{4}\bigg) \\
    &\quad + \underbrace{|y|^2h_{\gamma_0}k\HlocPower{2}}_{***} \bigg(|y| + k|y|^3\HlocPower{2} + k|y|^5 \HlocPower{4}\Bigg).
\end{align}
Grouping like terms, we can rewrite the right-hand side as
\begin{align}
    \leq{}& \CO(1) \cdot \alocPower{2} k^2 h_{\gamma_0}^2
    \Bigg(|y|+  k|y|^3 \HlocPower{2} +k^2|y|^5 \HlocPower{4} + k^2|y|^7 \HlocPower{6}+k^2|y|^9\HlocPower{8}\Bigg).  \label{eq:sum_after_gamma1}
\end{align}

\paragraph{Summing over $\gamma_0$. } Finally, we perform the final sum over $\gamma_0$, which simply amounts to noting that $\sum_{\gamma_0 \in \Gamma} h_{\gamma_0}^2 = \HgloPower{2}$. We conclude that
\begin{align}
    &\sum_{\gamma_0, \gamma_1,a,\gamma_2,\gamma_3, a',\gamma_4: \hat{\gamma} \in \CC}  h(\gamma_0, \gamma_1,\gamma_2,\gamma_3,\gamma_4) \norm{\CR_{\gamma_4 \gamma_3}^{a' \dag}\CR_{\gamma_2 \gamma_1}^{a \dag}(\Hbf_{\gamma_0})} \nonumber\\
    \leq{}& \CO(1) \cdot |y|\alocPower{2} k^2 \HgloPower{2}
    \Bigg(1+  k|y|^2 \HlocPower{2} +k^2|y|^4 \HlocPower{4} + k^2|y|^6 \HlocPower{6}+k^2|y|^8\HlocPower{8}\Bigg),\label{eq:sum_after_gamma0}
\end{align}
where we have factored out a $|y|$ for clarity. We finish by utilizing the fact that we have assumed that $\aloc k|y|^2 \HlocPower{2} < 1/8$. This allows us to write factors $k |y|^2 \HlocPower{2} \leq \CO(1)$, $k^2 |y|^4 \HlocPower{4} \leq \CO(1)$, $k^2 |y|^6 \HlocPower{6} \leq \CO(1)$, and $k^2 |y|^8 \HlocPower{8} \leq \CO(1)$. We also note that, by symmetry, the same bound can be shown when summing over sets $\CC'$, $\CC''$, and $\CC'''$, as defined in \cref{lem:CC_set}. Thus, using the result stated in \cref{lem:CC_set}, the full sum is at most a factor of $4 = \CO(1)$ larger than the sum over only coordinates of $\CC$. Thus, we may return to \cref{eq:sum_at_beginning} and conclude
\begin{align}
    \sum_{\substack{U \subset \Gamma \\ |U| \leq 5}} \L(\prod_{\alpha \in U} |y|h_{\alpha}\R) \norm{\Obf_U} \leq \CO(1) \cdot |y|\alocPower{2} k^2 \HgloPower{2},
\end{align}
proving the proposition. 
\end{proof}

\section{Some auxiliary bounds}\label{sec:lindbladin_bounds}
In this section, we compute some combinatorial estimates with inputs from the locality and commutation relation between the jumps $\vA^a$ and the Hamiltonian terms $\vH_{\gamma}$. Recall 
\begin{align}
    \lind^{a \dag} = \lind^{a \dag}_0 + \sum_{\gamma \in \Gamma} s_{\gamma}\lind^{a\dag}_{\gamma} + \sum_{\substack{\gamma,\gamma' \in \Gamma\\ \gamma \neq \gamma'}} s_{\gamma}s_{\gamma'}\lind^{a\dag}_{\gamma \gamma'},
\end{align}
where $\lind_0^{a \dag}$, $\lind_{\gamma}^{a \dag}$ and $\lind_{\gamma\gamma'}^{a\dag}$ are defined in \cref{eq:lind_0^a,eq:lind_gamma^a,eq:lind_gammagamma'^a}. Using \cref{cond:commute}, we can write $[\Abf^a, \Hbf_\gamma] = 2b_{a\gamma}\Abf^a \Hbf_{\gamma}$. We also note from \cref{cond:commute} that $\Abf^{a \dag} \Abf^a = \Ibf$, and from \cref{cond:random_Hamiltonians} that $\Hbf_{\gamma}^2 = h_{\gamma}^2\Ibf$.  Using these facts, we simplify
\begin{align}
    \lind_0^{a \dag} (\Obf) &= \left( \Abf^{a\dag}\Obf\Abf^a - \Obf \right) + 4 y^2 \sum_{\gamma \in \Gamma} b_{a\gamma} \left( \Hbf_\gamma \Abf^{a\dag} \Obf \Abf^a \Hbf_\gamma - h_\gamma^2 \Obf \right), \label{eq:lind_0_expansion} \\
    \lind_\gamma^{a \dag} (\Obf) &= 2y b_{a\gamma} \left( \acomm{\Hbf_\gamma}{\Abf^{a\dag}\Obf\Abf^a} - \acomm{\Hbf_\gamma}{\Obf} \right) , \label{eq:lind_gamma_expansion} \\
    \lind_{\gamma\gamma'}^{a \dag} (\Obf) &= 2y^2 b_{a\gamma} b_{a\gamma'} \left(2 \Hbf_\gamma \Abf^{a\dag} \Obf \Abf^a \Hbf_{\gamma'} - \acomm{\Hbf_\gamma \Hbf_{\gamma'}}{\Obf}\right). \label{eq:lind_gammagamma_expansion}
\end{align}
For ease of notation, we define (consistent with \cref{eq:lind_0/gamma/gammagamma'_def})
\begin{align}
    \lind_0^\dag= \sum_a \lind_0^{a \dag}, &\quad L_0 := \norm{\lind_0^\dag}, \label{eq:L_0} \\
    \lind_\gamma^\dag = \sum_a \lind_\gamma^{a \dag}, &\quad L_\gamma := \norm{\lind_\gamma^\dag}, \label{eq:L_gamma} \\
    \lind_{\gamma\gamma'}^\dag = \sum_a \left( \lind_{\gamma\gamma'}^{a \dag} + \lind_{\gamma'\gamma}^{a \dag} \right), &\quad L_{\gamma\gamma'} := \norm{\lind_{\gamma\gamma'}^\dag}. \label{eq:L_gammagamma}
\end{align}

First, we state and prove the following useful bounds on the Lindbladian terms in \cref{eq:L_gamma} and \cref{eq:L_gammagamma}:
\begin{lemma}\label{lem:L_bounds}
    The following bounds hold.
    \begin{align}
        L_\gamma &\le 8|y| \sum_{a \in A} b_{a\gamma} h_\gamma, \label{eq:L_gamma_bound} \\
        L_{\gamma\gamma'} &\le 8|y|^2 \sum_{a \in A} b_{a\gamma} b_{a\gamma'} h_\gamma h_{\gamma'}. \label{eq:L_gammagamma_bound}
    \end{align}
\end{lemma}
\begin{proof}
Using \cref{eq:lind_gamma_expansion}, the triangle inequality, and the fact that $\norm{\Abf^a} =1$ (\cref{cond:commute}),
\begin{align}
    \norm{\sum_{a \in A} \lind^{a\dag}_\gamma (\Obf)} &\le 2 |y| \sum_{a \in A} b_{a\gamma} \left( \norm{\Hbf_\gamma \Abf^{a\dag} \Obf \Abf^a + \Abf^{a\dag} \Obf \Abf^a \Hbf_\gamma} + \norm{\Hbf_\gamma \Obf + \Obf \Hbf_\gamma} \right) \\
    &\le 2 |y| \sum_{a \in A} b_{a\gamma} \left( \norm{\Hbf_\gamma} \norm{\Abf^{a\dag}} \norm{\Obf} \norm{\Abf^a} + \norm{\Abf^{a\dag}} \norm{\Obf} \norm{\Abf^a} \norm{\Hbf_\gamma} + \norm{\Hbf_\gamma} \norm{\Obf} + \norm{\Obf} \norm{\Hbf_\gamma} \right) \\
    &\le 8 |y| \sum_{a \in A} b_{a\gamma} \norm{\Hbf_\gamma} \norm{\Obf}.
\end{align}

Now, noting that $\norm{\Hbf_{\gamma}} = h_\gamma$ (\cref{cond:random_Hamiltonians}), we find that
\begin{align}
    L_\gamma = \sup_{\norm{\Obf} = 1} \frac{\norm{\lind_{\gamma}^\dag(\Obf)}}{\norm{\Obf}} \le 8 |y| \sum_{a \in A} b_{a\gamma} h_\gamma.
\end{align}
The other bound is analogous.
\end{proof}

Additionally, the following identity will be useful. 
\begin{lemma}\label{lem:sum_b_agamma_gammagamma}
    Let $\gamma' \in \Gamma$ be fixed. Then
    \begin{align}
        \sum_{a \in A} \sum_{\gamma \in \Gamma} b_{a\gamma'} b_{a \gamma} h_\gamma^2 &\le \aloc k \HlocPower{2}, \label{eq:sum_agamma_agamma'}
    \end{align}
    where $\Hloc$ is defined as in~\cref{def:local_global_energies}. 
\end{lemma}

\begin{proof}
We can write
\begin{align}
    \sum_{a \in A} \sum_{\gamma \in \Gamma} b_{a\gamma'} b_{a \gamma} h_\gamma^2 &= \sum_{a \in A} b_{a\gamma'} \sum_{\gamma \in \Gamma} b_{a\gamma} h_\gamma^2 \\
    &\le \left( \sum_{a \in A} b_{a\gamma'} \right) \left( \max_{a' \in A} \sum_{\gamma \in \Gamma} b_{a'\gamma} h_\gamma^2 \right) \\
    &\le \left( \sum_{a \in A} b_{a\gamma'} \right) \left( \max_{j \in [n]} \sum_{\gamma: j \in S(\gamma)} h_\gamma^2 \right) \\
    &= \left( \sum_{a \in A} b_{a\gamma'} \right) \HlocPower{2} \\
    &\le \aac k \HlocPower{2} \le \aloc k \HlocPower{2}.
\end{align}
In the  third line above, we used the fact that $\Abf^{a'}$ is an operator supported on a single-site $\{j\}$ (\cref{cond:aloc}), and that in order for $b_{a'\gamma} \neq 0$ to hold, $j$ must be in the set $S(\gamma)$ (\cref{cond:commute}).  In the fourth line we used the definition of $\Hloc$ from \cref{eq:Hloc_Hglo_def}. In the fifth line, we used \cref{cond:aloc} and the fact that $\aac \leq \aloc$.
\end{proof}

Next, we prove bounds on $F_p, G_p, \Phi_p$ (cf. \cref{eq:F_def,eq:G_def,eq:Phi_def}.)
\begin{lemma}\label{lem:bounds_on_paths}
Let $\gamma,\gamma'\in\Gamma$ and $p \ge 1$. Then the following bounds hold:
\begin{align}
    F_p(\gamma,\gamma') &\le h_\gamma h_{\gamma'} \left( 8 |y|^2 k\aloc \right) \left( 8|y|^2 k\aloc \HlocPower{2} \right)^{p-1}, \label{eq:path_from_f_bound} \\
    G_p(\gamma) &\le h_\gamma \left( 8|y| k\aloc \right) \left( 8|y|^2 k\aloc \HlocPower{2} \right)^{p-1}, \label{eq:path_from_g_bound} \\
    \Phi_p(\gamma) &\le h_\gamma \left( \sum_{\gamma'} h_{\gamma'} \right) \left( 8 |y|^2 k\aloc \right) \left( 8|y|^2 k\aloc \HlocPower{2} \right)^{p-1}.
\end{align}
\end{lemma}
\begin{proof}
First, consider $p=1$. The bound on $G_1$ follows from \cref{eq:L_gamma_bound} and assuming \cref{cond:aloc} (and noting $\aac \leq \aloc$). Similarly, the bound on $F_1$ follows from \cref{eq:L_gammagamma_bound}, making the loose bound $b_{a\gamma'}\leq 1$ and then using the same condition.

Now, let $p \ge 2$. Using \cref{lem:L_bounds}, we can write $G_p(\gamma)$ (see \cref{eq:G_def}) as
\begin{align}
&\sum_{\gamma_1,\gamma_2,\ldots,\gamma_{p-1}} h_\gamma \left( h_\gamma^{-1} L_{\gamma \gamma_1} h_{\gamma_1} \right) \left( h_{\gamma_1}^{-1} L_{\gamma_1\gamma_2} h_{\gamma_2} \right) \cdots \left( h_{\gamma_{p-2}}^{-1} L_{\gamma_{p-2} \gamma_{p-1}} h_{\gamma_{p-1}} \right) \left( h_{\gamma_{p-1}}^{-1} L_{\gamma_{p-1}} \right) \\ 
\le & \sum_{\gamma_1,\gamma_2,\ldots,\gamma_{p-1}} h_\gamma \left(\sum_a 8|y|^2 b_{a\gamma} b_{a\gamma_1} h_{\gamma_1}^2 \right) \left(\sum_a 8|y|^2 b_{a\gamma_1} b_{a\gamma_2} h_{\gamma_2}^2 \right) \cdots \left(\sum_a 8|y|^2 b_{a\gamma_{p-2}} b_{a\gamma_{p-1}} h_{\gamma_{p-1}}^2 \right) \left(\sum_a 8|y| b_{a\gamma_{p-1}} \right) \\
= & \quad h_\gamma \left(\sum_{a,\gamma_1} 8|y|^2 b_{a\gamma} b_{a\gamma_1} h_{\gamma_1}^2 \right) \left(\sum_{a,\gamma_2} 8|y|^2 b_{a\gamma_1} b_{a\gamma_2} h_{\gamma_2}^2 \right) \cdots \left(\sum_{a,\gamma_{p-1}} 8|y|^2 b_{a\gamma_{p-2}} b_{a\gamma_{p-1}} h_{\gamma_{p-1}}^2 \right) \left( 8|y| \aac k \right).
\end{align}

By applying \cref{lem:sum_b_agamma_gammagamma} and noting $\aac\leq \aloc$, we conclude the claimed bound.

We proceed in a similar fashion for $F_p(\gamma,\gamma')$ (see \cref{eq:F_def}), which we can write using \cref{lem:L_bounds} as
\begin{align}
&\sum_{\gamma_1,\gamma_2,\ldots,\gamma_{p-1}} h_\gamma \left( h_\gamma^{-1} L_{\gamma \gamma_1} h_{\gamma_1} \right) \left( h_{\gamma_1}^{-1} L_{\gamma_1\gamma_2} h_{\gamma_2} \right) \cdots \left( h_{\gamma_{p-2}}^{-1} L_{\gamma_{p-2} \gamma_{p-1}} h_{\gamma_{p-1}} \right) \left( h_{\gamma_{p-1}}^{-1} L_{\gamma_{p-1} \gamma'} \right) \\ 
\le& \quad h_\gamma \left( \sum_{a,\gamma_1} 8|y|^2 b_{a\gamma} b_{a\gamma_1} h_{\gamma_1}^2 \right) \left( \sum_{a,\gamma_2} 8|y|^2 b_{a\gamma_1} b_{a\gamma_2} h_{\gamma_2}^2 \right) \cdots \left(\sum_{a,\gamma_{p-1}} 8|y|^2 b_{a\gamma_{p-2}} b_{a\gamma_{p-1}} h_{\gamma_{p-1}}^2 \right) \left(\sum_a 8|y|^2  b_{a\gamma'} h_{\gamma'}\right) \\
\le & \quad h_\gamma \left( 8|y|^2 \aloc k \HlocPower{2} \right)^{p-1} \left( 8 |y|^2 \aloc k h_{\gamma'} \right),
\end{align}
where in the second line we assert $b_{a \gamma_{p-1}}\leq 1$ in the last factor (at the cost of a looser bound).

The bound on $\Phi_p$ follows from the observation that, for all $p \geq 1$,
\begin{equation}
    \Phi_p(\gamma) = \sum_{\gamma'} F_p(\gamma,\gamma').
\end{equation}
\end{proof}

\section{Crude upper bounds on optimum}\label{sec:matrix_hoeffding}
\begin{proposition}[Spectral norm]\label{prop:max_eigenvalues}
For the sparse or dense, spin or fermion ensembles (\eqref{eq:random_Pauli},\eqref{eq:sampled_Pauli},\eqref{eq:random_Fermion},\eqref{eq:sampled_Fermion}), the expected maximal eigenvalue satisfy
\begin{align}
    \BE \lambda_{max}(\vH) =\CO( \sqrt{\log(N)}) = \CO(\sqrt{n}). 
\end{align}
\end{proposition}
\begin{proof}
For the sampled model, the spectral norm concentrates
\begin{align}
\Pr( \norm{\vH} \ge E) \le 2N \exp(-\frac{E^2}{8})
\end{align}
due to Matrix Hoeffding~\cite[Theorem 1.3]{tropp2012user} and that Pauli and Majorana operators both square to identity. Here, the Hilbert space dimension is $N=2^n$ for the Pauli ensembles (\eqref{eq:random_Pauli},\eqref{eq:sampled_Pauli}) and $N={2^{n/2}}$ for the fermionic ensembles (\eqref{eq:random_Fermion},\eqref{eq:sampled_Fermion}). The Guassian models are slightly better concentrated
\begin{align}
\Pr( \norm{\vH} \ge E) \le 2N \exp(-\frac{E^2}{2})
\end{align}
using~\cite[Theorem 4.1]{tropp2012user}. In each case, integrating the Gaussian tail gives the same advertised result. 
\end{proof}

\section{Computing the local norms and global norms}\label{sec:compuing_localnorm}
In this section, we give a standard concentration calculation to show that the local energy quantity is very well concentrated, establishing \cref{prop:Hglo_Hloc_bound}, which we restate for convenience.
\begin{customthm}{\Cref{prop:Hglo_Hloc_bound}}[restated]
Suppose that $\Hbf$ is drawn randomly from the ensembles in \cref{eq:random_Pauli} or \cref{eq:random_Fermion} with $1<k<n$, or it is drawn randomly from the ensembles in \cref{eq:sampled_Pauli} or \cref{eq:sampled_Fermion} with $m\ge cn\log(n)/k$ for some constant $c$. Then
\begin{align}
    \BE_{\Hbf} \L(\frac{\HgloPower{2}}{\Hloc}\R) =\Omega\left( \sqrt{\frac{n}{k}} \right),
\end{align}
where the expectation value is taken over random choice of $\Hbf$. 
\end{customthm}

\begin{proof}
We simplify the expression by standard inequalities
\begin{align}\label{eq:applied_Cauchy_Schwartz}
    \BE_{\Hbf} \L(\frac{\HgloPower{2}}{\Hloc}\R) \ge \frac{(\BE\Hglo)^2 }{\BE\Hloc} \ge \frac{(\BE\Hglo)^2 }{\sqrt{\BE\HlocPower{2}}}
\end{align}
liberally using Cauchy--Schwartz $\BE[\labs{a}]\BE[\labs{b}]\ge \BE[\sqrt{\labs{ab}}]^2$ and $\BE[\labs{a}] \le \sqrt{\BE[a^2]}$, with $a = \Hloc$ and $b = \HgloPower{2}/\Hloc$. 

First, we lower bound the numerator, $(\BE \Hglo)^2$. For the sampled models (\cref{eq:sampled_Pauli} and \eqref{eq:sampled_Fermion}), we have $\HgloPower{2} = 1$ for every instance, and thus $(\BE_{\Hbf} \Hglo)^2 = 1$. For the dense models (\cref{eq:random_Pauli} and \cref{eq:random_Fermion}), the quantity $\HgloPower{2}$ is distributed as the sum of the squares of $M$ centered Gaussian random variables $h_{\gamma}$, each of variance $1/M$, where $M = \binom{n}{k}$ in the case of \eqref{eq:random_Fermion} and $M = \binom{n}{k}3^k$ in the case of \eqref{eq:random_Pauli}. Thus, $\BE_{\Hbf} \HgloPower{2}=1$, and
\begin{align}
    \BE_{\Hbf}\HgloPower{4} = \sum_{\gamma_1, \gamma_2}\BE_{\Hbf} h_{\gamma_1}^2 h_{\gamma_2}^2 = \frac{M+2}{M},
\end{align}
using the value of the fourth moment $\BE h_{\gamma}^4 = \frac{3}{M^2}.$
By the Paley--Zygmund inequality, we have 
\begin{align}
    \Pr_{\Hbf}\L[\Hglo \geq \frac{1}{\sqrt{2}}\R] =\Pr_{\Hbf}\L[\HgloPower{2}  \geq \frac{1}{2}\BE_{\Hbf} \HgloPower{2}\R]\geq \frac{1}{4}\frac{(\BE_{\Hbf} \HgloPower{2})^2}{\BE_{\Hbf}\HgloPower{4}}=   \frac{M}{4M+8}
\end{align}
and hence $(\BE \Hglo)^2 \geq \frac{M^2}{32(M+2)^2} = \Omega(1)$ for $M\ge 1$. Second, we upper bound the denominator, $\sqrt{\BE \HlocPower{2}}$. We write 
\begin{align}
    \BE_{\Hbf}\HlocPower{2} = \BE_{\Hbf}\max_{1 \leq i \leq n} \sum_{\gamma: i \in S(\gamma)} h_\gamma^2 =: \BE_{\Hbf} \max_{1 \leq i \leq n} Y_i\quad \text{where}\quad Y_i := \sum_{\gamma: i \in S(\gamma)} h_\gamma^2.
\end{align}
We can now choose a value of $y_0$ (to be specified later) and write
\begin{align}
    \BE_\Hbf \max_{1 \leq i \leq n} Y_i   
    &=\int_0^{\infty} \rd y\; \BE_\Hbf \indicator\L[y \leq \max_{1 \leq i \leq n} Y_i\R]\\
    &= \int_0^{\infty} \rd y \Pr_{\Hbf}\L[y \leq \max_{1 \leq i \leq n} Y_i\R] \leq y_0 + \int_{y_0}^{\infty} \rd y \Pr_{\Hbf}\L[y \leq \max_{1 \leq i \leq n} Y_i\R] \\
    &\leq y_0 + n\int_{y_0}^{\infty} \rd y \Pr_{\Hbf}\L[y \leq  Y_i\R],\label{eq:y0_int_Yi}
\end{align}
where $\indicator$ is the indicator function, and in the last step we have used the union bound and the symmetry over choice of $i$.

We now consider the sampled and Gaussian models separately, beginning with the models in \cref{eq:random_Pauli,eq:random_Fermion}, where $h_{\gamma}$ is distributed as a Gaussian. In this case, the subset $\{\gamma: i \in S(\gamma)\}$ contains $D=\frac{k}{n}\binom{n}{k}$ elements in the case of \cref{eq:random_Fermion} and $D = \frac{k}{n}\binom{n}{k}3^k$ elements in the case of \cref{eq:random_Pauli}. Either way, the variance of each $h_{\gamma}$ is chosen to be $k/(nD)$. The random variable $Y_i$, being a sum of the squares of $D$ Gaussian random variables, follows a (scaled) chi-squared distribution with mean $k/n$, for which the following concentration bound is known. For any $x>0$ it holds that \cite[Lemma 1]{laurent2000adaptive}
\begin{align}
\Pr_{\Hbf}\left[Y_i \geq \frac{k}{n}\left(1+2\sqrt{\frac{x}{D}} + 2\frac{x}{D}\right)\right] \leq \e^{-x}   \,. 
\end{align}
 Choosing $y_0 = \frac{5k}{n}$ in Eq.~\eqref{eq:y0_int_Yi}, we now make the substitution $y' = (y-y_0)n/k$, invoke the concentration bound, and simplify
\begin{align}
    \BE_\Hbf \max_{1 \leq i \leq n} Y_i &\leq \frac{5k}{n} + k \int_{0}^{\infty} \rd y'\; \Pr_{\Hbf}\L[Y_i \geq \frac{k}{n}(5+y')\R] \\
    &\leq  \frac{5k}{n} + \int_{0}^{\infty} \rd y'\; k\e^{-D}\e^{-Dy'/4}.
\end{align}
The last bound follows from the basic calculation
\begin{align}
    5+y' = 1+2\sqrt{x/D} + 2x/D \leq 1+4x/D \quad \text{whenever}\quad y' >0
\end{align}
and plugging this value of $x$ into the concentration bound. The integral in the final expression evaluates to 
\begin{align}
    \CO(k\e^{-D}/D) \le \CO(k/n)\quad \text{whenever}\quad 1 < k < n.
\end{align}
Thus, for the dense model, $\BE_{\Hbf} \HlocPower{2} \leq \CO(k/n)$.

Next, we consider the models in \cref{eq:sampled_Pauli,eq:sampled_Fermion}, where $m$ terms are chosen (with replacement) and, for those terms, $h_{\gamma} = 1/\sqrt{m}$. We define a Bernoulli random variable $C_{\gamma i}$ to be 1 if $S(\gamma)$ contains site $i$ (which occurs with probablity $k/n$) and $0$ otherwise. Note that $C_{\gamma i}$ is independent from one $\gamma$ to the next. We may then assert that 
\begin{align}
    Y_i = \frac{1}{m}\sum_{\gamma=1}^m C_{\gamma i}\,,
\end{align}
that is, that $Y_i$ is the average of $m$ independent Bernoulli random variables. Using Bernstein's inequality, we have the following concentration bound. For any $x > 0$,
\begin{align}
\Pr_{\Hbf} \left[Y_i \geq \frac{k}{n}\left(1+x\right)\right] \leq \exp\left(-\frac{3m k x^2}{n(6+2x)}\right)  \,. 
\end{align}
We follow the same procedure as above, this time choosing $y_0 = \frac{2k}{n} $ in Eq.~\eqref{eq:y0_int_Yi} and again making the substitution $y' = (y-y_0)n/k$. 
\begin{align}
    \BE_{\Hbf} \max_{1 \leq i \leq n} Y_i &\leq \frac{2k}{n}+ k \int_{0}^{\infty} \rd y' \; \Pr_{\Hbf}\L[Y_i \geq \frac{2k}{n} + \frac{ky'}{n}\R] \\
    &\leq \frac{2k}{n}+ \int_{0}^{\infty} \rd y' \; k \exp\left(-\frac{3mk(1+y')^2}{n(8+2y')}\right)\\
    &\leq \frac{2k}{n}+ \int_{0}^{\infty} \rd y' \; k \e^{-3mk/8n}\e^{-3mky'/8n} \\
    &= \frac{2k}{n} + \frac{8n}{3m}\e^{-3mk/8n}
\end{align}
where in the second to last line, we have used $(1+y')^2/(8+2y') \geq (1+y')/8$ for $y'>0$. Therefore, there exists an absolute constant $c$ such that if $m \geq cn \log(n)/k$, the right-hand side evaluates to $\CO(k/n)$, and thus $\BE_{\Hbf} \HlocPower{2} \leq \CO(k/n)$. 

Combining these bounds on the numerator and denominator of \cref{eq:applied_Cauchy_Schwartz}, we show that the quantity of interest is $\Omega(\sqrt{n}/\sqrt{k})$. 
\end{proof}

\end{document}